\documentclass[11pt,letter]{article}

\usepackage{amssymb}
\usepackage{amsmath}
\usepackage{graphicx}
\usepackage{color}
\usepackage{bigints}
\usepackage{array}
\usepackage[normalem]{ulem}
\usepackage{xspace}
\usepackage{enumerate}
\usepackage{hyperref}
\usepackage{numprint}
\usepackage{subcaption}
\usepackage{changepage}
\usepackage{enumerate}
\usepackage{amsthm}
\usepackage{graphicx}
\usepackage{mathrsfs}
\usepackage[ruled,linesnumbered]{algorithm2e}

\newcommand{\Frechet}{Fr{\'e}chet\xspace}

\newcommand{\Q}{Q}

\newcommand{\real} {\mathbb{R}}
\newcommand{\eps} {\varepsilon}
\newcommand{\grid} {{\mathcal{G}}}
\newcommand{\cancel}[1] {}

\newtheorem{theorem}{\text{Theorem}}
\newtheorem{lemma}{\text{Lemma}}

\DeclareMathOperator{\E}{E}

\graphicspath{{./}{figures/}}

\begin{document}
	
\begin{titlepage}
		
\title{Curve Simplification and Clustering \\ under \Frechet Distance\footnote{Research supported by Research Grants Council, Hong Kong, China (project no.~16203718).}}
		
\author{Siu-Wing Cheng\footnote{Department~of~Computer~Science~and~Engineering,
				HKUST, Hong Kong. Email: {\tt scheng@cse.ust.hk}, {\tt haoqiang.huang@connect.ust.hk}}
		\and 
				Haoqiang Huang\footnotemark[2]}

\date{}

\maketitle

\begin{abstract}
	We present new approximation results on curve simplification and clustering under Fr\'{e}chet distance.  Let $T = \{\tau_i : i \in [n] \}$ be polygonal curves in $\real^d$ of $m$ vertices each.  Let $\ell$ be any integer from $[m]$.  We study a generalized curve simplification problem: given error bounds $\delta_i > 0$ for $i \in [n]$, find a curve $\sigma$ of at most $\ell$ vertices such that $d_F(\sigma,\tau_i) \leq \delta_i$ for $i \in [n]$.  We present an algorithm that returns a null output or a curve $\sigma$ of at most $\ell$ vertices such that $d_F(\sigma,\tau_i) \leq \delta_i + \eps\delta_{\max}$ for $i \in [n]$, where $\delta_{\max} = \max_{i \in [n]} \delta_i$.  If the output is null, there is no curve of at most $\ell$ vertices within a Fr\'{e}chet distance of $\delta_i$ from $\tau_i$ for $i \in [n]$.  The running time is $\tilde{O}\bigl(n^{O(\ell)} \cdot m^{O(\ell^2)} \cdot (d\ell/\eps)^{O(d\ell)}\bigr)$.  This algorithm yields the first polynomial-time bicriteria approximation scheme to simplify a curve $\tau$ to another curve $\sigma$, where the vertices of $\sigma$ can be anywhere in $\real^d$, so that $d_F(\sigma,\tau) \leq (1+\eps)\delta$ and $|\sigma| \leq (1+\alpha)\cdot \min\{|c| : d_F(c,\tau) \leq \delta\}$ for any given $\delta > 0$ and any fixed $\alpha, \eps \in (0,1)$.  The running time is $\tilde{O}\bigl(m^{O(1/\alpha)} \cdot (d/(\alpha\eps))^{O(d/\alpha)}\bigr)$.
	By combining our technique with some previous results in the literature, we obtain an approximation algorithm for $(k,\ell)$-median clustering.   Given $T$, it computes a set $\Sigma$ of $k$ curves, each of $\ell$ vertices, such that $\sum_{i \in [n]} \min_{\sigma \in \Sigma} d_F(\sigma,\tau_i)$ is within a factor $1+\eps$ of the optimum with probability at least $1-\mu$ for any given $\mu, \eps \in (0,1)$.  The running time is $\tilde{O}\bigl(n \cdot m^{O(k\ell^2)} \cdot \mu^{-O(k\ell)} \cdot (dk\ell/\eps)^{O((dk\ell/\eps)\log(1/\mu))}\bigr)$.
\end{abstract}

\thispagestyle{empty}
\end{titlepage}

	\section{Introduction}
	
	The popularity of trajectory data analysis in applications such as wildlife monitoring, delivery tracking, and transportation analysis has generated a lot of interest in curve simplification and clustering under the Fr\'{e}chet distance $d_F$.   Given a polygonal curve $\tau$ of $m$ vertices in $\real^d$ and a value $\delta > 0$, curve simplification calls for computing a polygonal curve $\sigma$ of fewer vertices such that $d_F(\sigma,\tau) \leq \delta$.  Given a set of polygonal curves $T$ and two positive integers $k$ and $\ell$, the $(k,\ell)$-clustering problem is to find a set $\Sigma$ of $k$ curves, each of $\ell$ vertices,  that minimizes some distance measure between $\Sigma$ and $T$.  
	%Both problems are NP-hard in general.  
	We present new approximation results for both problems.
	%curve simplification and median clustering under Fr\'{e}chet distance.
	
	%We briefly survey some previous results in the following; more detailed discussions can be found in~\cite{agarwal2005near,van2019global,van2018optimal}.  
	
	\vspace{8pt}
	
\noindent \textbf{Previous works.}  
Alt and Godau~\cite{alt1995computing} developed the first algorithm for computing $d_F(\sigma,\tau)$; it runs in $O\bigl(|\sigma||\tau|\log (|\sigma||\tau|)\bigr)$ time, where $|\sigma|$ and $|\tau|$ denote their numbers of vertices.  Let $\kappa(\tau,\delta) = \min\{|c| : d_F(c,\tau) \leq \delta\}$.
Let $m$ be $|\tau|$.  Agarwal~et~al.~\cite{agarwal2005near} named this problem as weak Fr\'{e}chet $\delta$-simplification and proposed an $O(m\log m)$-time algorithm in $\real^2$ that returns a curve $\sigma$ such that $d_F(\sigma,\tau) \leq \delta$ and $|\sigma| \leq \kappa(\tau,\delta/4)$ for a given $\delta > 0$. Guibas~et~al.~\cite{guibas1993approximating} presented an $O(m^2\log^2 m)$-time algorithm that minimizes $|\sigma|$ such that $d_F(\sigma, \tau)\le \delta$ in $\real^2$. But in $\real^d$ with $d\ge 3$, no algorithm is known yet.
Van~Kreveld et al.~\cite{van2018optimal} 
can minimize $|\sigma|$ in $O(|\sigma|m^5)$ time under the constraints of $d_F(\sigma,\tau) \leq \delta$ for a given $\delta > 0$ and the vertices of $\sigma$ being a subset of the vertices of $\tau$.  Van~de~Kerkhof~et~al.~\cite{van2019global} 
%showed that the simpification problem under $d_H$ remains NP-hard when the vertices of $\sigma$ can be anywhere in $\real^d$ or on $\tau$.  They also 
improved the running time to $O(m^3)$---a result also obtained by Bringmann and Chaudhury~\cite{bringmann2019polyline}---and that the problem is NP-hard for $d \geq 2$ if the vertices of $\sigma$ can be anywhere on $\tau$.  
Van~de~Kerkhof~et~al.~proposed another algorithm that returns a curve $\sigma$ in $O(\mathrm{poly}(1/\eps) \cdot m^2\log m \log\log m)$ time such that $d_F(\sigma,\tau) \leq (1+\eps)\delta$ and $|\sigma| \leq 2\kappa(\tau,\delta) - 2$, if the vertices of $\sigma$ can be anywhere in $\real^d$.

Let $T$ be a set of $n$ polygonal curves in $\real^d$, each of $m$ vertices.   The \emph{$(k,\ell)$-center clustering} problem is to find a set $\Sigma$ of $k$ curves, each of $\ell$ vertices, such that $\max_{\tau \in T} \min_{\sigma \in \Sigma} d_F(\sigma,\tau)$ is minimized.  The \emph{$(k,\ell)$-median clustering} problem is to minimize $\sum_{\tau \in T} \min_{\sigma \in \Sigma} d_F(\sigma,\tau)$.  Driemel et al.~\cite{driemel2016clustering} initiated the study of $(k,\ell)$-center clustering; they obtained approximation ratios of $1+\eps$ in one dimension and 8 in higher dimensions.
%; the problem is to compute the a set $C$ of $k$ curves of $\ell$ vertices each so that $\max_{\tau \in T} \min_{\sigma \in C} d_F(\sigma,\tau)$ is minimized.  In one dimension, they designed an approximation algorithm that runs in $\tilde{O}(mn)$ time and returns a center set with cost at most $1+\epsilon$ times more than the optimal cost; for $d \geq 2$, they designed an algorithm that runs in $O(\mathrm{poly}(m,n,k,\ell))$ time and returns a center set with cost at most 8 times the optimum.  
Buchin~et~al.~\cite{buchin2019approximating} proved that if $\ell$ is part of the input, there is no polynomial-time approximation scheme for $d \geq 2$ unless P = NP; if both $k$ and $\ell$ are constants, a lower bound of $2.25-\epsilon$ on the approximation ratio is shown.  Buchin~et~al.~also obtained smaller constant factor approximations for $d \geq 2$. It is worth noting that the hardness results~\cite{buchin2019approximating} of the $(1,\ell)$-center problem imply that the generalized curve simplification problem is also NP-hard and hard to approximate with a small constant factor.
%that run in$\tilde{O}(kmn\ell + km^3)$ time; for $d = 2$, the approximation solution cost is with a factor 3 of the optimum; for $d > 2$, the solution cost is within a factor 6 of the optimum.  
For $(k,\ell)$-median clustering, Buchin~et~al.~\cite{BDS20} proved that the problem is NP-hard even if $k = 1$.  Subsequently, Buchin~et~al.~\cite{buchin2021approximating} designed a randomized bicriteria approximation algorithm; it computes a set $\Sigma$ of $k$ curves that has a cost at most $1+\eps$ times the optimum with probability at least $1-\mu$.  Each curve in $\Sigma$ may have up to $2\ell-2$ vertices.  The running time is $\tilde{O}\bigl(n \cdot m^{O(k\ell)} \cdot 2^{O((k^3/\eps^2)\log^2(1/\mu))} \cdot (k/(\mu\eps))^{O(dk\ell)}\bigr)$.
%linear in $nm^{\mathrm{poly}(k,\ell,1/\eps)}$ and exponential in $\mathrm{poly}(d,k,\ell,1/\eps,\ln(1/\mu))$.  
There are some results on coresets for $(k,\ell)$-median clustering under Fr\'{e}chet distance~\cite{BR22}.

	\vspace{8pt}
	
	\noindent \textbf{Our results.} Let $T = \{\tau_i : i \in [n] \}$ be polygonal curves in $\real^d$ of $m$ vertices each.  Let $\ell$ be any integer from $[m]$.  We study a generalized curve simplification problem: given error bounds $\delta_i > 0$ for $i \in [n]$, find a curve $\sigma$ of at most $\ell$ vertices such that $d_F(\sigma,\tau_i) \leq \delta_i$ for $i \in [n]$.  We present an algorithm that returns a null output or a curve $\sigma$ of at most $\ell$ vertices such that $d_F(\sigma,\tau_i) \leq \delta_i + \eps\delta_{\max}$ for $i \in [n]$, where $\delta_{\max} = \max_{i \in [n]} \delta_i$.  If the output is null, there is no curve of at most $\ell$ vertices within a Fr\'{e}chet distance of $\delta_i$ from $\tau_i$ for $i \in [n]$.  The running time is $\tilde{O}\bigl(n^{O(\ell)} \cdot m^{O(\ell^2)} \cdot (d\ell/\eps)^{O(d\ell)} \bigr)$.
	
	This algorithm also yields a polynomial-time bicriteria approximation scheme to simplify a curve $\tau$ to another curve $\sigma$, where the vertices of $\sigma$ can be anywhere in $\real^d$, so that $d_F(\sigma,\tau) \leq (1+\eps)\delta$ and $|\sigma| \leq (1+\alpha)\kappa(\tau,\delta)$ given any $\delta > 0$ and any $\alpha, \eps \in (0,1)$.  The running time is $\tilde{O}\bigl(m^{O(1/\alpha)} \cdot (d/(\alpha\eps))^{O(d/\alpha)}\bigr)$. This is the first polynomial-time bicriteria approximation scheme for simplifying a curve in $\real^d$ with $d\ge 3$.
	
	By combining our technique with the framework in~\cite{buchin2021approximating}, we obtain an approximation algorithm for $(k,\ell)$-median clustering.   Given $T$, it computes a set $\Sigma$ of $k$ curves, each of $\ell$ vertices, such that $\sum_{i \in [n]} \min_{\sigma \in \Sigma} d_F(\sigma,\tau_i)$ is within a factor $1+\eps$ of the optimum with probability at least $1-\mu$ for any given $\mu, \eps \in (0,1)$.  The running time is $\tilde{O}\bigl(n \cdot m^{O(k\ell^2)} \cdot \mu^{-O(k\ell)} \cdot (dk\ell/\eps)^{O((dk\ell/\eps)\log(1/\mu))}\bigr)$.  This result answers affirmatively the question raised in the previous work~\cite{buchin2021approximating}, which guarantees a bound of $2\ell-2$ on the output curve sizes, of whether the bound $\ell$ can be achieved with similar efficiency.

	%The previous result in~\cite{buchin2021approximating} only guarantees that each curve in $\Sigma$ has at most $2\ell-2$ vertices.
	
	There are two main ingredients of our results.  The first one is a space of configurations.  We use the grids introduced by Buchin~et~al.~\cite{buchin2021approximating} as a part of our discretization scheme; however, instead of enumerating all possible curves through the discretization vertices, we define configurations with some novel structural constraints in order to satisfy the bound $\ell$ on the size of the output curves.  Second, we design a two-phase method to construct approximate curves from the configurations.
	
	\vspace{8 pt}
	
	\noindent \textbf{Notations.}  We often denote a curve $\sigma$ as a sequence  $(u_1,u_2,\ldots,u_l)$ of its vertices.  Given  two points $x, y$ on $\sigma$, we say that $x \leq_{\sigma} y$ if $y$ is not encountered before $x$ as we walk along $\sigma$ from $u_1$.  Given two subsets $X$ and $Y$ of points on $\sigma$, we say that $X \leq_{\sigma} Y$ if $x \leq_{\sigma} y$ for all $x \in X$ and $y \in Y$.
	
	A \emph{parameterization} of $\sigma$ is a continuous function $\rho : [0,1] \rightarrow \sigma$ such that $\rho(0) = u_1$, $\rho(1) = u_l$, and for all $t, t' \in [0,1]$, $t \leq t' \iff \rho(t) \leq_{\sigma} \rho(t')$.   A \emph{matching} $g$ from a curve $\sigma$ to another curve $\tau$ is a pair of parameterizations $(\rho,\rho')$ for $\sigma$ and $\tau$, respectively, and $d_g(\sigma,\tau) = \max_{t \in [0,1]} d(\rho(t),\rho'(t))$.  For any point $x \in \sigma$, we denote the points in $\tau$ matched to $x$ by $g(x) = \bigl\{y \in \tau : \exists \, t \in [0,1] \,\, \text{s.t.} \,\, x = \rho(t) \, \wedge \, y = \rho'(t) \bigr\}$; for a subset $X \subseteq \sigma$, $g(X) = \bigcup_{x \in X} g(x)$.  For any point $y \in \tau$, we denote the points in $\sigma$ matched to $y$ by $g^{-1}(y) = \{x \in \sigma: \exists \, t \in [0,1] \,\, \text{s.t.} \, x = \rho(t) \, \wedge \, y = \rho'(t)\}$; for a subset $Y \subseteq \tau$, $g^{-1}(Y) = \bigcup_{y \in Y} g^{-1}(y)$.  The \emph{Fr\'{e}chet distance} of $\sigma$ and $\tau$ is $d_F(\sigma,\tau) = \inf_{g} d_g(\sigma,\tau)$.   A \emph{Fr\'{e}chet matching} from $\sigma$ to $\tau$ is a matching that realizes $d_F(\sigma,\tau)$.  Clearly, $d_F(\sigma,\tau) = d_F(\tau,\sigma)$. 
	
	\cancel{
	A map $f : \sigma \rightarrow \tau$ is \emph{order-respecting} if $f$ is continuous and $f(x) \leq_{\tau} f(y)$ for all points $x,y \in \sigma$ that satisfy $x \leq_{\sigma} y$.  
	%Note that $f(x)$ and $f(y)$ may be subsets of points on $\tau$.  
	The map $f$ is \emph{surjective} if $f^{-1}(p) \not= \emptyset$ for every point $p \in \tau$.  Let $\mathcal{F}$ be the class of all surjective, order-respecting map from $\sigma$ to $\tau$.  The \emph{Fr\'{e}chet distance} between $\sigma$ and $\tau$ is $d_F(\sigma,\tau) = \min_{f \in \mathcal{F}} \max_{x \in \sigma, p \in f(x)} d(p,x)$.  
	%Equivalently, one can also define $d_F(\sigma,\tau)$ to the infimum of $\max_{x \in \sigma} d(x,g(x))$ over all bijective, order-respecting functions $g$ from $\sigma$ to $\tau$.  
	We call the maps in $\cal F$ that realize the distance $d_F(\sigma,\tau)$ the \emph{Fr\'{e}chet maps} from $\sigma$ to $\tau$.  Clearly, if $f \in \mathcal{F}$ is a Fr\'{e}chet map, then $f^{-1}$ is a Fr\'{e}chet map from $\tau$ to $\sigma$.  So $d_F(\sigma,\tau) = d_F(\tau,\sigma)$.
}
	
	We will be dealing with a set $T$ of polygonal curves of $m$ vertices each.  For each curve $\tau_i \in T$, we denote its vertices in order along $\tau_i$ by $v_{i,1},\ldots,v_{i,m}$.  For all $a \in [m-1]$, $\tau_{i,a}$ denotes the edge $v_{i,a}v_{i,a+1}$.  For all $a,b \in [m]$ such that $a \leq b$, $[v_{i,a},v_{i,b}]$ denotes the vertices $\{v_{i,a},v_{i,a+1},\ldots,v_{i,b}\}$.  Given two points $x, y$ on a curve $\tau$ such that $x \leq_{\tau} y$, $\tau[x,y]$ denotes the subcurve of $\tau$ from $x$ to $y$.
	
	Given two points $x,y \in \real^d$, $xy$ denotes the closed line segment connecting $x$ and $y$.  Given any curve $\tau$, $\mathrm{int}(\tau)$ denotes the relative interior of $\tau$.  We use $B_r$ to denote the ball centered at the origin with radius $r$.   Given two subsets $S$ and $S'$ of $\real^d$, their \emph{Minkowski sum} is $S \oplus S' = \{x + y : x \in S, \, y \in S'\}$.  Given a point $p$ and $S \subseteq \real^d$, $p + S = \{p + x : x \in S\}$. For any point $x$ and any segment~$s$, $x \!\downarrow\! s$ denotes the orthogonal projection of $x$ onto the support line of $s$; so $x \!\downarrow\! s$ may not lie on $s$.  For a point set $X$, $X \!\downarrow\! s = \{x \!\downarrow\! s : x \in X\}$.
	%and the \emph{Minkowski difference} $S \ominus S'$ is $\{x - y : x \in S, \, y \in S'\}$.  
	%If $S$ consists of a single point $p$, we use $p + S'$ to denote $\{p\} \oplus S'$.  
	%For example, $p + B_r$ gives the ball centered at $p$ with radius $r$; if $C$ is an object that contains the origin, then $xy \oplus C$ is the object obtained by sweeping $C$ along $xy$, i.e., identifying the origin with every point in $xy$.

	\section{Simplified representative of a set of curves} 
		
	Let $\Delta = \{\delta_i : i \in [n]\}$ be a set of error thresholds prescribed for $T$.  Define $\Q(T, \Delta,\ell)$ to be the problem of finding a curve $\sigma$ of at most $\ell$ vertices such that $d_F(\sigma, \tau_i)\le \delta_{i}$ for $i \in [n]$.   
	%t is too challenging to compute the feasible locations of each vertex in one go.  
	%In the forward phase, we compute a superset of the feasible locations of each vertex; in the backward phase, we extract the vertices from these supersets.

	\subsection{Configurations}
	\label{sec:discrete}
	
	Imagine an infinite grid of hypercubes of side length $\alpha$.  Given a subset $R \subset \real^d$, let $G(R,\alpha)$ be the subset of grid cells that intersect $R$.
	
	Let $\min = \mathrm{argmin}_{i \in [n]} \delta_i$.  We compute $\mathcal{L} = \bigcup_{a \in [m-1]} L_a$, where $L_a$ is a set of segments that are parallel to $\tau_{\min,a}$.  First, compute the convex hull $C_a$ of $G(v_{\min,a} + B_{\delta_{\min}},\eps\delta_{\min}) \cup G(v_{\min,a+1} + B_{\delta_{\min}},\eps\delta_{\min})$.  Second, for every grid vertex $x \in G(v_{\min,a} + B_{\delta_{\min}},\eps\delta_{\min})$, take the line through $x$ that is parallel to $\tau_{\min,a}$, clip this line within $C_a$ to a segment, and include this segment in $L_a$.  The size of $L_a$ is $O(\eps^{-d})$; it can be computed in $O(\eps^{-d^2/2})$ time.  Each point in $C_a$ is at distance $\sqrt{d}\eps\delta_{\min}$ or less from a segment in $L_a$.  The size of $\mathcal{L}$ is $O(m\eps^{-d})$; it can be computed in $O(m\eps^{-d^2/2})$ time.

	Take any integer $l \in [\ell]$. 
	We construct two sets of grid cells: $\grid_1 = \bigcup_{i\in[n]}\bigcup_{a\in[m]} G\bigl(v_{i,a}+B_{\delta_i + \sqrt{d}\eps\delta_i}, \eps\delta_i/l)$ and $\grid_2 = \bigcup_{i\in[n]}\bigcup_{a\in[m]} G\bigl(v_{i,a}+B_{9\sqrt{d}\delta_{\max}}, \eps\delta_{\max})$, where $\delta_{\max} = \max_{i \in [n]}\delta_i$.  The size and construction time of $\grid_1$ are $O(mnl^d\eps^{-d})$; those of $\grid_2$ are smaller by an $l^d$ factor.
	
	%We are ready to define the configurations for the combinatorial search using $L_i$ for $i \in [n]$, $\grid_1$, and $\grid_2$.   For each configuration, we will carry out a geometric optimization.  We pick the best among all these optimization results.   
	
	%\subsection{Configurations}
	
	 Each configuration $\Psi_l$ is a 4-tuple $(\mathcal{P}, \mathcal{C}, \mathcal{S},\mathcal{A})$ designed to capture a candidate curve $\sigma = (w_1,\ldots,w_l)$ of $l$ vertices.  The component $\mathcal{P}$ is an $n$-tuple $(\pi_i)_{i\in [n]}$.  Each $\pi_i$ is a function from $[m]$ to $[0,l-1]$ that partitions the vertices of $\tau_i$ into at most $l$ contiguous subsets.  If $\pi_i(a) = 0$, it means that $v_{i,a}$ should be matched to $w_1$; if $\pi_i(a) = j \in [l-1]$, it means that $v_{i,a}$ should be matched to point(s) in $w_jw_{j+1} \setminus \{w_j\}$.  We require that 
	 %$\bigcup_{i \in [n]} \pi_i^{-1}(j)\not= \emptyset$ for every $j \in [l-1]$, and 
	 for all $i \in [n]$, $\pi_i(1) = 0$ and if $a \leq b$, then $\pi_i(a) \leq \pi_i(b)$.

	The component $\mathcal{C}$ is an $(l-1)$-tuple $((c_{j,1}, c_{j,2}))_{j\in[l-1]}$, where $c_{j,1}$ and $c_{j,2}$ are cells in $\grid_1$.  The cells $c_{j,1}$ and $c_{j,2}$ may be equal.  This component imposes the requirement that for every $j \in [l-1]$, $w_jw_{j+1}$ must intersect $c_{j,1}$ and $c_{j,2}$ in such a way that $w_jx \cap c_{j,1} \not= \emptyset$ for some point $x \in w_jw_{j+1} \cap c_{j,2}$.  The component $\mathcal{S}$ is an $l$-tuple $(s_j)_{j\in[l]}$, where each $s_j$ is a segment in $\mathcal{L}$.   This component imposes the requirement that $w_j \in s_j$.  It is possible that $s_j = s_k$ for two distinct $j, k \in [l]$.  The component $\mathcal{A}$ is an array of $l$ entries.    For each $j \in [l]$, $\mathcal{A}[j]$ is null or a cell in $\grid_2$; $\mathcal{A}[1]$ and $\mathcal{A}[l]$ must be cells in $\grid_2$; if $\mathcal{A}[j] \not= \mathrm{null}$, it imposes the constraint that $w_j \in s_j \cap \mathcal{A}[j]$.
	
	We will compute a candidate curve for each configuration.  Any candidate curve $\sigma$ that satisfies $d_F(\sigma,\tau_i) \leq \delta_i + \eps\delta_{\max}$ can be returned.  If no such curve is found, we report that $Q(T,\Delta,\ell)$ has no solution.  To satisfy the inequalities $d_F(\sigma,\tau_i) \leq \delta_i + \eps\delta_{\max}$, we will need to define $\eps' = \eps/\Theta(\sqrt{d})$ and substitute $\eps$ by $\eps'$ in the discretization scheme.  The number of configurations will go up by an $O(d^{d/2})$ factor.  If we use all input curves to form the configurations, there will be too many because there are close to $m^{nl}$ different $\mathcal{P}$'s.  We will discuss in Section~\ref{sec:accelerate} how to reduce this number.
	%There are at most $m^{l-1}$ different partitions of the vertices of $\tau_i$, so there are at most $m^{n(l-1)}$ different $\mathcal{P}$'s.  Since $|\grid_1| = O(mn\eps^{-d})$, there are at most $O((mn)^{2(l-1)}\eps^{-2d(l-1)})$ different $\mathcal{C}$'s.  As $|\mathcal{L}| = O(m\eps^{-d})$, there are at most $O(m^l\eps^{-dl})$ different $\mathcal{S}$'s.  Similarly, there are at most $O((mn)^l\eps^{-dl})$ different $\mathcal{A}$'s because $|\grid_2| = O(mn\eps^{-d})$.   The number of configurations is at most the product of all these terms which is $O(m^{(n+4)l-n-2} n^{3l-1} \eps^{-2d(2l-1)})$.  
	
	%The exponential factor $m^{nl}$ is clearly too large; it stems from enumerating the vertex partitions of the $n$ input curves.  After we present the two-phase method for constructing a simplified curve, we will prove that it suffices to work with the vertex partitions of $6l$ input curves instead of all of them.  This will allow us to drastically reduce the number of different $\mathcal{P}$'s to $O(n^{6l}m^{6l(l-1)})$.   The total number of configurations will correspondingly drop to $O(m^{6l(l-1)}n^{9l-1}\eps^{-2d(2l-1)}) = O(m^{6l^2}n^{9l}\eps^{-4dl})$.
	
	\begin{figure}
		\centering
		\includegraphics[scale=1.2]{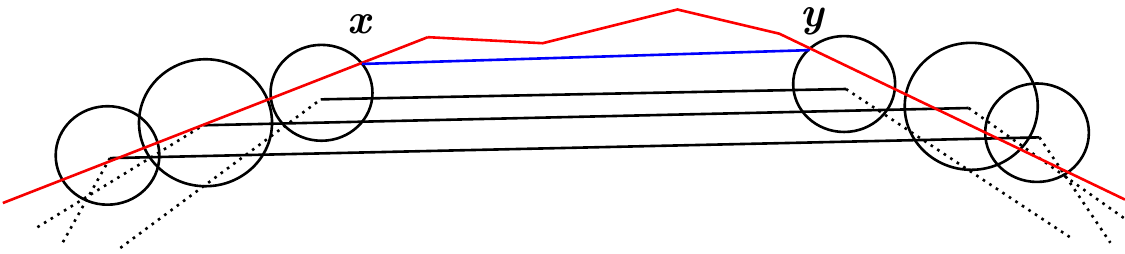}
		\caption{Visualization of the underlying idea of the component $\mathcal{C}$. The black curves are input curves. The red curve $\sigma$ is a solution for $Q(T,\Delta, \ell)$. The black circles are balls $v_{i,a}+B_{\delta_{i}}$ centered at $v_{i,a}$'s. If $\sigma[x,y]$ is matched to a single segment for every input curve, by shortcutting the subcurve $\sigma[x,y]$ to the blue segment, we can get a curve $\sigma'$ that is a solution for $Q(T, \Delta, \ell)$. In addition, every segment of $\sigma'$ intersects a ball $v_{i,a}+B_{\delta_i}$ centered around some input vertex $v_{i,a}$. It means that the intersection between every segment of $\sigma'$ and $\mathcal{G}_1$ is not empty, which motivates the definition of $(c_{j,1}, c_{j,2})$.}
	\end{figure}
	
\subsection{Constraints with respect to a configuration}
	
We describe several constraints that enforce the intuition behind the definition of a configuration.  These constraints (or their relaxations) will be verified by our algorithm.  Consider a configuration $\Psi_l = (\mathcal{P},\mathcal{C},\mathcal{S},\mathcal{A})$ and a candidate curve $\sigma = (w_1,\ldots,w_l)$ to be constructed for $\Psi_l$.   Constraint~1 requires that the cells in $\mathcal{C}$ are close to the corresponding input subcurves.  Constraint~2 restricts the locations of the vertices and edges of $\sigma$.  Constraint~3 concerns with whether  the vertices of $\sigma$ can be matched to the input curves in an order respecting manner within the error bounds.  %For any point $x$ and any segment~$s$, $x \!\downarrow\! s$ denotes the orthogonal projection of $x$ onto the support line of $s$; so $x \!\downarrow\! s$ may not lie on $s$.  For a point set $X$, $X \!\downarrow\! s = \{x \!\downarrow\! s : x \in X\}$.
 
\begin{quote}
\noindent \textbf{Constraint 1:} For every $i \in [n]$ and every $j \in [l-1]$, if $\pi_i^{-1}(j)$ is some non-empty $[a,b]$, then for every vertex $x$ of $c_{j,1}$ and every vertex $y$ of $c_{j,2}$, there exist points $p,q \in xy$ such that $d_F(pq,\tau_i[v_{i,a},v_{i,b}]) \leq \delta_i + 2\sqrt{d}\eps\delta_i$.

\vspace{4pt}

\noindent \textbf{Constraint 2:} 
\vspace{-6pt}
\begin{enumerate}[(a)]
	\item For every $j \in [l]$, if $\mathcal{A}[j]$ is null, then $w_j \in s_j$; otherwise, $w_j \in s_j \cap \mathcal{A}[j]$. 
	\item For every $j \in [l-1]$, $w_jx \cap c_{j,1} \not= \emptyset$ for some point $x \in w_jw_{j+1} \cap c_{j,2}$.  
\end{enumerate}

\vspace{4pt}

\noindent \textbf{Constraint 3:} 
\vspace{-6pt}
\begin{enumerate}[(a)]
	\item For every vertex $x$ of $\mathcal{A}[1]$ and every vertex $y$ of $\mathcal{A}[l]$, both $d(v_{i,1},x)$ and $d(v_{i,m},y)$ are at most  $\delta_i + 2\sqrt{d}\eps\delta_{\max}$ for all $i \in [n]$.
	\item Take any index $i \in [n]$.  For all $a \in [m-1]$, define $J_a = \bigl\{j :  \pi_i(a) < j \leq \pi_i(a+1) \, \wedge \, \mathcal{A}[j] \not= \mathrm{null} \bigr\}$, i.e., for $j \in J_a$, some point(s) in $\tau_{i,a}$ should be matched to $w_j$.  Constraint~3(b) requires that for all $a \in [m-1]$, if $J_a \not= \emptyset$, there exist points $\{p_j\in  \tau_{i,a} : j \in J_a \}$ such that:
	\begin{enumerate}[(i)]
		\item for all $j, k \in J_a$, if $j < k$, then $p_j \leq_{\tau_i} p_k$;
		\item for every $j \in J_a$ and every vertex $x$ of $\mathcal{A}[j]$, $d(p_j,x) \leq \delta_i + 2\sqrt{d}\eps\delta_{\max}$.
	\end{enumerate}
	\item Take any index $j \in [l-1]$ such that $\mathcal{A}[j] = \mathrm{null}$.  Note that $j > 1$ in order that $\mathcal{A}[j] = \mathrm{null}$.  Let $N_j = \bigl\{i  :  \exists \, a_i \in [m-1] \,\, \text{s.t.} \,\, \pi_i(a_i) < j \leq \pi_i(a_i+1) \bigr\}$, i.e., for $i \in N_j$, some point(s) in $\tau_{i,a_i}$ should be matched to $w_j$.  Constraint~3(c) requires that the following conditions are satisfied for all $i \in N_j$.
	\begin{enumerate}[(i)]
		\item $w_j \!\downarrow \!\tau_{i,a_i}\in \tau_{i,a_i}$ and $d(w_j,\tau_{i,a_i}) \leq \delta_i + \sqrt{d}\eps\delta_{\max}$.
		\item If $\pi_i(a_i) < j-1 \leq \pi_i(a_i+1)$ and $\mathcal{A}[j-1] = \mathrm{null}$, then $w_{j-1}\!\downarrow\! \tau_{i,a_i} \leq_{\tau_i} w_j\!\downarrow\!\tau_{i,a_i}$.
		\item  If $\pi_i(a_i) < j-1 \leq \pi_i(a_i+1)$  and $\mathcal{A}[j-1] \not= \mathrm{null}$, then \\
		$\tau_{i,a_i} \cap (\mathcal{A}[j-1] \oplus B_{\delta_i +3\sqrt{d}\eps\delta_i}) \leq_{\tau_{i}} w_j\!\downarrow\!\tau_{i,a_i}$.
		%\item If $\pi_i(a) < j+1\leq \pi_i(a+1)$ and $\mathcal{A}[j+1] = \mathrm{null}$, then $w_j \!\downarrow\! \tau_{i,a_i} \leq_{\tau_i} w_{j+1} \!\downarrow\!\tau_{i,a_i}$.
		\item If $\pi_i(a_i) < j+1 \leq \pi_i(a_i+1)$ and $\mathcal{A}[j+1] = \mathrm{null}$, then $w_j \!\downarrow\! \tau_{i,a_i} \leq_{\tau_i} w_{j+1} \!\downarrow\!\tau_{i,a_i}$.
		\item If $\pi_i(a_i) < j+1 \leq \pi_i(a_i+1)$  and $\mathcal{A}[j+1] \not= \mathrm{null}$, then \\
		$w_j\!\downarrow\!\tau_{i,a_i} \leq_{\tau_i} \tau_{i,a_i} \cap (\mathcal{A}[j+1] \oplus B_{\delta_i+3\sqrt{d}\eps\delta_i})$.
	\end{enumerate}
	%The case of $\pi_i(a) < j+1 \leq \pi_i(a+1)$  and $\mathcal{A}[j+1] = \mathrm{null}$ is covered by the application of constraint~3(c) to $j+1$.
\end{enumerate}
\end{quote}

Constraints~1--3 are justified by Lemma~\ref{Lemma: valid_cons_set} below.  It is proved by snapping the vertices of the solution curve to the discretization; the details are deferred to Appendix~\ref{app:Lemma: valid_cons_set}.

\begin{figure}
	\centering
	\includegraphics[scale=1.5]{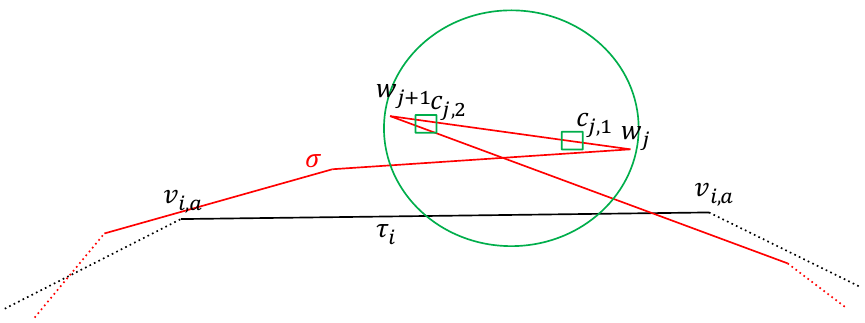}
	\caption{Visualization of the underlying idea of the component $\mathcal{A}$. The black curve is some input curve $\tau_i$. The red curve is a solution curve for $Q(T, \Delta, \ell)$ such that every segment intersects a ball centered at some input vertex. Suppose that both $w_j$ and $w_{j+1}$ are matched to some points on $\tau_{i,a}$. If $w_j \!\downarrow$\! $\tau_{i,a} \ge_{\tau_i} w_{j+1} \!\downarrow\! \tau_{i,a}$, the length of $w_jw_{j+1}$ must be no more than $2\delta_{i}$. It means that both $w_j$ and $w_{j+1}$ can be covered by $B_{9\sqrt{d}\delta_{\max}}+v_{i', a'}$ for some input vertex $v_{i', a'}$.}
\end{figure}

\begin{lemma}\label{Lemma: valid_cons_set}
If $\Q(T, \Delta,\ell)$ has a solution, there exist a configuration $\Psi_l$ and a curve $\sigma = (w_1,\ldots,w_l)$ for some $l \in [\ell]$ such that constraints~1--3 are satisfied and $d_F(\tau_{i}, \sigma)\le \delta_{i}+ \sqrt{d}\eps\delta_{\min}$ for $i\in [n]$.
%the following properties are satisfied.
\cancel{
\begin{enumerate}[{\em (i)}]
	\item Let $\sigma = (u_1,\ldots,u_l)$ be any solution curve for $Q(T,\Delta,\ell)$.  We can find a Fr\'{e}chet matching $g_i$ from $\tau_i$ to $\sigma$ for $i \in [n]$ such that for all $j \in [l-1]$, there exist $i \in [n]$ and $a \in [m]$ such that $g_i(v_{i,a}) \cap u_ju_{j+1} \not= \emptyset$ and $u_j \not\in g_i(v_{i,a})$. 
	\item There exist a configuration $\Psi_l$ and a curve $\sigma = (w_1,\ldots,w_l)$ for some $l \leq \ell$ such that constraints~1--3 are satisfied and $d_F(\tau_{i}, \sigma)\le \delta_{i}+ \sqrt{d}\eps\delta_{\min}$ for $i\in [n]$.
\end{enumerate}
}
\end{lemma}

\subsection{Forward construction}
%\subsection{Configuration checking}
\label{sec:forward}

%Lemma~\ref{Lemma: valid_cons_set} justifies the following strategy.   Enumerate all configurations.  For each configuration $\Psi$, construct a curve $\sigma$ of $l$ vertices that satisfies constraints~1--4 together with $\Psi$.  We will prove that any such curve contsructed have Fr\'{e}chet distance of at most $\delta_i + \eps\delta_{\max}$ from $\tau_i$ for all $i \in [n]$.  If we cannot construct any curve in the end, we assert that $Q(T,\Delta,l)$ has no solution.

Given a configuration $\Psi_l = (\mathcal{P},\mathcal{C},\mathcal{S},\mathcal{A})$, we check if it satisfies constraint~1, and if so, whether there exists a curve $\sigma$ that satisfies constraints~2 and~3.   It is difficult to check constraints~2,~3(c)(ii), and~3(c)(iv) exactly; therefore, we will check some relaxed versions that will be introduced later.  We will check constraints~3(b), 3(c)(i), 3(c)(iii), and~3(c)(v) exactly though.
%The subsequent backward phase will compensate for the relaxation of constraint~2.

We check constraint~1 as follows.  Take any $i \in [n]$ and any $j \in [l-1]$ such that $\pi_i^{-1}(j)$ is some non-empty $[a,b]$.  Let $x$ and $y$ be any two vertices of $c_{j,1}$ and $c_{j,2}$, respectively.   If $xy \cap (v_{i,a} + B_{\delta_i +2\sqrt{d}\eps\delta_i})$ or $xy \cap (v_{i,b} + B_{\delta_i+2\sqrt{d}\eps\delta_i})$ is empty, $\Psi_l$ does  not satisfy constraint~1.   Suppose that they are non-empty.  Let $p_1$ and $p_2$ be the points in $xy \cap (v_{i,a} + B_{\delta_i+2\sqrt{d}\eps\delta_i})$ that  are the minimum and maximum with respect to $\leq_{xy}$, respectively.  Similarly, let $q_1$ and $q_2$ be the points in $xy \cap (v_{i,b} + B_{\delta_i+2\sqrt{d}\eps\delta_i})$ that are the minimum and maximum with respect to $\leq_{xy}$, respectively.  We compute $d_F(p_1q_2,\tau_i[v_{i,a},v_{i,b}])$.  If $d_F(p_1q_2,\tau_i[v_{i,a},v_{i,b}]) > \delta_i + 2\sqrt{d}\eps\delta_i$, then $\Psi_l$ does not satisfy constraint~1.  Otherwise, we repeat the above for all vertices of $c_{j,1}$ and $c_{j,2}$, $j \in [l-1]$, and $i \in [n]$.  If the check is passed every time, then $\Psi_l$ satisfies constraint~1.   For a fixed $i$, the total time needed over all $j \in [l-1]$ is $O(m\log m)$ because $\bigl|\bigcup_{j \in [l-1]} \pi_i^{-1}(j)\bigr| = m$.  So the total time over all $i \in [n]$ and all $j \in [l-1]$ is $O(mn\log m)$.

We prove the correctness of this check.  There are points $p,q \in xy$ such that $d_F(pq,\tau_i[v_{i,a},v_{i,b}]) \leq \delta_i + 2\sqrt{d}\eps\delta_i$ if and only if there exist points $p \in xy \cap (v_{i,a} + B_{\delta_i+2\sqrt{d}\eps\delta_i})$ and $q \in xy \cap (v_{i,b} + B_{\delta_i+2\sqrt{d}\eps\delta_i})$ such that $d_F(pq,\tau_i[v_{i,a},v_{i,b}]) \leq \delta_i + 2\sqrt{d}\eps\delta_i$.  Such points $p$ and $q$ lie in $p_1p_2$ and $q_1q_2$, respectively.  All points in $p_1p$ and $qq_2$ can be matched to $v_{i,a}$ and $v_{i,b}$, respectively, within the error bound of $\delta_i +2\sqrt{d}\eps\delta_i$.   Hence, $p$ and $q$ exist if and only if  $d_F(p_1q_2,\tau_i[v_{i,a},v_{i,b}]) \leq \delta_i + 2\sqrt{d}\eps\delta_i$.

%\subsection{Forward construction}

The rest of the forward phase is to inductively compute supersets $\gamma_1,\ldots,\gamma_l$ of the feasible locations of the vertices $w_1,\ldots,w_l$ of $\sigma$ with respect to $\Psi_l$.  We will see that every $\gamma_j$ is a line segment.  We need the geometric construct $F(R,S) = \bigl\{p \in \real^d : \exists \, q \in S \,\, \text{s.t.} \,\, pq \cap R \not= \emptyset\bigr\}$, where $R$ and $S$ are two bounded convex polytopes in $\real^d$.  We can show that $F(R,S)$ is a convex polytope, and it can be constructed by computing a convex hull and a Minkowski sum.  In our usage, $|R|$ and $|S|$ are $O(2^{O(d)})$; as a result, $|F(R,S)| = O(2^{O(d)})$ and its construction time is $O(2^{O(d)})$. Refer to Appendix~\ref{app:geom} for details.  The inductive computation of $\gamma_k$ is as follows.  If $\gamma_j$ is found empty for any $j \in [l]$, we abort and do not go to the backward phase.

\vspace{8pt}

\noindent \textbf{The case of $\pmb{k=1}$.}   If every vertex of $\mathcal{A}[1]$ is within a distance of $\delta_i + 2\sqrt{d}\eps\delta_{\max}$ from $v_{i,1}$ for all $i \in [n]$,  compute $\gamma_1 = F(c_{1,1},c_{1,2}) \cap s_1 \cap \mathcal{A}[1]$.  Abort otherwise.  By constraint~2, $w_1 \in s_1 \cap \mathcal{A}[1]$, and $w_1x \cap c_{1,1} \not= \emptyset$ for some point $x \in w_1w_2 \cap c_{1,2}$. Therefore, $F(c_{1,1},c_{1,2}) \cap s_1 \cap \mathcal{A}[1]$ represents a relaxed version of constraint~2 on $w_1w_2$.  The processing time of this case is $O(n2^{O(d)})$.

\vspace{8pt}

\noindent \textbf{The case of $\pmb{k \in [2,l-1]}$.}  Suppose that $\gamma_1,\ldots,\gamma_{k-1}$ have been constructed for some $k \in [2,l-1]$.  
%There are two cases in constructing $\gamma_k$.

\vspace{6pt}

\underline{Case~1}: $\mathcal{A}[k] \not= \mathrm{null}$.  Compute $\gamma_k = F(c_{k-1,2},\gamma_{k-1}) \cap F(c_{k,1},c_{k,2}) \cap s_k \cap \mathcal{A}[k]$ in $O(2^{(O(d)})$ time.   %\label{case1}

As before, $w_k \in F(c_{k,1},c_{k,2}) \cap s_k \cap \mathcal{A}[k]$ is a relaxed version of constraint~2 on $w_kw_{k+1}$.  By constraint~2 again,
we must connect $w_k$ to $w_{k-1}$, which is in $\gamma_{k-1}$, such that $w_{k-1}w_k \cap c_{k-1,2} \not= \emptyset$, implying that $w_k \in F(c_{k-1,2},\gamma_{k-1})$.  Therefore, $\gamma_k =  F(c_{k-1,2},\gamma_{k-1}) \cap F(c_{k,1},c_{k,2}) \cap s_k \cap \mathcal{A}[k]$ satisfies a relaxed version of constraint~2.

Let $J'_a = \{j \in [k] : \pi_i(a) < j \leq \pi_i(a+1) \, \wedge \, \mathcal{A}[j] \not= \mathrm{null}\}$.  To check whether $\gamma_k$ satisfies constraint~3(b), we need to check the existence of $\{p_j \in \tau_{i,a} : j \in J'_a\}$ in increasing order of $j$ along $\tau_{i,a}$ that satisfy $d(p_j,x) \leq \delta_i + 2\sqrt{d}\eps\delta_{\max}$ for every vertex $x$ of $\mathcal{A}[j]$.  Such $p_j$'s must lie in the common intersection of $x + B_{\delta_i+2\sqrt{d}\eps\delta_i}$ over all vertices $x$ of $\mathcal{A}[j]$.  For every $j \in J'_a$, $\tau_{i,a}$ intersects this common intersection in an interval $I_{i,j}$.  Let $j_1 < \ldots < j_{|J'_a|}$ be the increasing order of indices in $J'_a$.  For $r = |J'_a|-1, \ldots, 1$ in this order, we trim $I_{i,j_r}$ to the interval $\{ p \in I_{i,j_r} : p \leq_{\tau_i} \max(I_{i,j_{r+1}}) \}$.  Afterwards, constraint~3(b) can be satisfied for $i$ if and only if $I_{i,j} \not= \emptyset$ for all $j \in J'_a$.  If the check is passed for every $i \in [n]$ and every $a \in [m-1]$, we accept $\gamma_k$; otherwise, we abort.  We spend $O(m+l2^{O(d)})$ time over all $a \in [m-1]$ for each $i \in [n]$.  The processing time is thus $O(mn2^{O(d)})$. 

\vspace{6pt}

\underline{Case~2}: $\mathcal{A}[k] = \mathrm{null}$.
%In this case, $w_k$ belongs to $s_k$ by constraint~2.  In the following, we formulate the extra restrictions due to the combination of constraint~2 and constraints~3(c)(i)--(v).   
%
%We first introduce a geometric construct.  Take the $(d-1)$-dimensional ball $B$ with radius $\delta_i + 2\sqrt{d}\eps\delta_{\max}$ that is centered at the origin, orthogonal to $\tau_{i,a_i}$.  We can put $O(\eps^{-d})$ points in the boundary of $B$ such that the boundary of the convex hull $H$ of these points are at distance at least $\delta_i + \sqrt{d}\eps\delta_{\max}$ from the center of $B$.  Note that the center of $B$ lies inside $H$; we call it the center of $H$ too.  We move the center of  $H$ along any segment $xy \subseteq \tau_{i,a_i}$ to sweep out a convex polytope $H_{i,a_i}(xy)$.  Algorithmically, it ie easier to work with the polytope $H_{i,a_i}(xy)$ than the cylinder $xy \oplus B$.  There are $O(\eps^{-d^2/2})$ facets in $H_{i,a_i}(xy)$.  
%
To satisfy a relaxed version of constraint~2 for $w_{k-1}w_k$ and $w_kw_{k+1}$, we require $w_k \in F(c_{k-1,2},\gamma_{k-1}) \cap F(c_{k,1},c_{k,2}) \cap s_k$.  Recall that $N_k = \bigl\{i : \exists \, a_i \in [m-1] \,\, \text{s.t.} \,\, \pi_i(a_i) < k \leq \pi_i(a_i+1)\bigr\}$.  Let $H_{i,a_i}$ be the cylinder with axis $\tau_{i,a_i}$ and radius $\delta_i + \sqrt{d}\eps\delta_{\max}$.  To satisfy constraint~3(c)(i), we require $w_k \in \bigcap_{i \in N_k} H_{i,a_i}$.  Altogether, we initialize $\gamma_k = F(c_{k-1,2},\gamma_{k-1}) \cap F(c_{k,1},c_{k,2}) \cap s_k \cap \bigcap_{i \in N_k} H_{i,a_i}$.  We can compute in $O(2^{O(d)})$ time the clipped segment $F(c_{k-1,2},\gamma_{k-1}) \cap F(c_{k,1},c_{k,2}) \cap s_k$.  Then we intersect the clipped segment with each $H_{i,a_i}$ in $O(1)$ time.  The total initialization time is $O(n+2^{O(d)})$.  We may trim $\gamma_k$ further as discussed below.

Case~2.1: Suppose that constraint~3(c)(ii) is applicable because $\pi_i(a_i) < k-1 \leq \pi_i(a_i+1)$ and $\mathcal{A}[k-1] = \mathrm{null}$.   
%We want the segment connecting $w_k$ to $w_{k-1} \in \gamma_{k-1}$ to intersect $c_{k-1,2}$; this is a relaxed version of constraint~2 for $w_{k-1}w_k$.  
By constraint~3(c)(i) on $w_{k-1}$, we have $w_{k-1} \in H_{i,a_i}$.  So $w_k$ satisfies constraint~3(c)(ii) if and only if $w_k-w_{k-1}$ makes a non-negative inner product with $v_{i,a_i+1} - v_{i,a_i}$.   That is, $w_k-w_{k-1} \in \Pi_{i,a_i}$, where $\Pi_{i,a_i}$ is the closed halfspace containing $v_{i,a_i+1} - v_{i,a_i}$ such that the bounding hyperplane of $\Pi_{i,a_i}$ passes through the origin and is orthogonal to $v_{i,a_i+1} - v_{i,a_i}$.   Since $w_{k-1}w_k$ intersects $c_{k-1,2}$, we relax constraint~3(c)(ii) to the restriction that $w_k \in \bigcap_{i \in N_k'} c_{k-1,2} \oplus \Pi_{i,a_i}$, where $N'_k = \bigl\{ i \in N_k : \pi_i(a_i) < k-1 \leq \pi_i(a_i+1) \, \wedge \, \mathcal{A}[k-1] = \mathrm{null}\bigr\}$.  There is no need to compute $\bigcap_{i \in N'_k} c_{k-1,2} \oplus \Pi_{i,a_i}$ because we can clip $\gamma_k$ with each $c_{k-1,2} \oplus \Pi_{i,a_i}$ in $O(2^{O(d)})$ time.

Case~2.2: Suppose that constraint~3(c)(iii) is applicable because $\pi_i(a_i) < k-1 \leq \pi_i(a_i+1)$ and $\mathcal{A}[k-1] \not= \mathrm{null}$.   We compute in $O(2^{O(d)})$ time the point $p_{i,a_i}\in \tau_{i,a_i} \cap (\mathcal{A}[k-1] \oplus B_{\delta_i + 3\sqrt{d}\eps\delta_i})$ that is maximum according to $\leq_{\tau_i}$.  We already require $w_k \in H_{i,a_i}$.  Thus, satisfying constraint~3(c)(iii) is equivalent to requiring $w_k \in p_{i,a_i} + \Pi_{i,a_i}$. The extra restriction in Case~2.2 is thus $w_k \in \bigcap_{i \in N''_k} (p_{i,a_i}+ \Pi_{i,a_i})$, where $N''_k = \bigl\{ i \in N_k : \pi_i(a_i) < k-1 \leq \pi_i(a_i+1) \, \wedge \, \mathcal{A}[k-1] \not= \mathrm{null}\bigr\}$.  We do not compute $\bigcap_{i \in N''_k} (p_{i,a_i}+ \Pi_{i,a_i})$; we clip $\gamma_k$ with each $p_{i,a_i}+ \Pi_{i,a_i}$ in $O(1)$ time instead.

Case~2.3: Suppose that $\pi_i(a_i) < k+1\leq \pi_i(a_i+1)$ and $\mathcal{A}[k+1] = \mathrm{null}$. 
%A relaxed version of constraint~2 for $w_kw_{k+1}$ requires $w_k \in F(c_{k,1},c_{k,2})$.  In addition, 
As in case~2.1 above, constraint~3(c)(iv) for $w_kw_{k+1}$ requires $w_{k+1}-w_k \in \Pi_{i,a_i}$; we relax this requirement to the extra restriction that $w_k \in \bigcap_{i \in N^*_k} c_{k,2} \oplus (-\Pi_{i,a_i})$, where $N^*_k = \bigl\{ i \in N_k : \pi_i(a_i) < k+1 \leq \pi_i(a_i+1) \, \wedge \, \mathcal{A}[k+1] = \mathrm{null}\bigr\}$.  We can clip $\gamma_k$ with each $c_{k,2} \oplus (-\Pi_{i,a_i})$ in $O(2^{O(d)})$ time.

Case~2.4: Suppose that $\pi_i(a_i) < k+1\leq \pi_i(a_i+1)$ and $\mathcal{A}[k+1] \not= \mathrm{null}$.   We compute in $O(2^{O(d)})$ time the minimum point $q_{i,a_i}$ in $\tau_{i,a_i} \cap (\mathcal{A}[k+1] \oplus B_{\delta_i + 3\sqrt{d}\eps\delta_i})$ according to $\leq_{\tau_i}$ for all $i \in N_k$.  As in case~2.2, satisfying constraint~3(c)(v) is equivalent to requiring that $w_k \in q_{i,a_i} - \Pi_{i,a_i}$.  So the extra restriction is $w_k \in \bigcap_{i \in N^{**}_k} (q_{i,a_i} - \Pi_{i,a_i})$, where $N^{**}_k = \bigl\{ i \in N_k : \pi_i(a_i) < k+1 \leq \pi_i(a_i+1) \, \wedge \, \mathcal{A}[k+1] \not= \mathrm{null}\bigr\}$.  We can clip $\gamma_k$ with each $q_{i,a_i} - \Pi_{i,a_i}$ in $O(1)$ time.

\vspace{6pt}

\underline{Summary}:  We list the different definitions of $\gamma_k$ for $k \in [2,l-1]$ in the following.
\begin{itemize}
\item $\mathcal{A}[k] \not= \mathrm{null}$: Compute $\gamma_k = F(c_{k-1,2},\gamma_{k-1}) \cap F(c_{k,1},c_{k,2}) \cap s_k \cap \mathcal{A}[k]$.  Check constraint~3(b).

\item $\mathcal{A}[k] = \mathrm{null}$: Initialize $\gamma_k = F(c_{k-1,2},\gamma_{k-1}) \cap F(c_{k,1},c_{k,2}) \cap s_k \cap \bigcap_{i \in N_k} H_{i,a_i}$.  If $N'_k \not= \emptyset$, update $\gamma_k = \gamma_k \cap \bigcap_{i \in N'_k} c_{k-1,2} \oplus \Pi_{i,a_i}$.  If $N''_k \not= \emptyset$, update $\gamma_k = \gamma_k \cap \bigcap_{i \in N''_k} (p_{i,a_i} + \Pi_{i,a_i})$.  If $N^*_k \not= \emptyset$, update $\gamma_k = \gamma_k \cap \bigcap_{i \in N^*_k} c_{k,2} \oplus (-\Pi_{i,a_i})$.  If $N^{**}_k \not= \emptyset$, update $\gamma_k = \gamma_k \cap \bigcap_{i \in N^{**}_k} (q_{i,a_i} - \Pi_{i,a_i})$. 
\end{itemize}
The total processing time for Case~2 is $O(n2^{O(d)})$.

\vspace{6pt}

\noindent \textbf{The case of $\pmb{k = l}$.}  Since $\mathcal{A}[l] \not= \mathrm{null}$, we proceed as in the case of $k \in [2,l-1]$, but we do not need to consider $(c_{l,1},c_{l,2})$.  That is, we compute $\gamma_l = F(c_{l-1,2},\gamma_{l-1}) \cap s_l \cap \mathcal{A}[l]$ in $O(2^{O(d)})$ time, and we check constraints~3(a) and 3(b) in $O(mn2^{O(d)})$ time as before.

\vspace{6pt}

%In Lemma~\ref{lem:forward} below, the correctness of the forward construction is immediate from the discussion above; the construction time of $\gamma_1,\ldots,\gamma_l$ follows from a straightforward verification of the time needed to compute the varioius geometric structures in the inductive definition of $\gamma_k$.  The details are given in Appendix~\ref{app:forward}.
	
\begin{lemma}
	\label{lem:forward}
	Given a configuration $\Psi_l$, the forward construction runs in $O(mn\log m + lmn2^{O(d)})$ time.
	If $\Psi_l$ satisfies constraint~1 and there exists a curve $\sigma = (w_1,\ldots,w_l)$ that satisfies constraints~2 and~3 with respect to $\Psi_l$, the forward construction produces a sequence of non-empty line segments $(\gamma_1,\ldots,\gamma_l)$ such that $w_j \in \gamma_j$ for all $j \in [l]$.   
\end{lemma}

\subsection{Backward extraction}
\label{sec:backward}
 
Suppose that the forward construction succeeds with the output $\gamma_1,\ldots,\gamma_l$.  The backward extraction works as follows.  Set $u_l$ to be any point in $\gamma_l$.  For $j = l-1, \ldots, 1$ in this order, set $u_j$ to be any point in $F(c_{j,2},u_{j+1}) \cap \gamma_j$ in $O(2^{O(d)})$ time.  Let $\sigma = (u_1,\ldots,u_l)$ denote the extraction output.  

The extraction succeeds in $O(l2^{O(d)})$ time if $F(c_{j,2},u_{j+1}) \cap \gamma_j$ is not empty for every $j$.  No matter which scenario was applicable in computing $\gamma_{j+1}$ in the forward phase, we always have $\gamma_{j+1} \subseteq F(c_{j,2},\gamma_j)$.  It follows that $u_{j+1} \in F(c_{j,2},\gamma_j)$, meaning that there exists a point $q \in \gamma_j$ such that $qu_{j+1} \cap c_{j,2} \not= \emptyset$.  This point $q$ belongs to $F(c_{j,2}, u_{j+1})$, which implies that $F(c_{j,2},u_{j+1}) \cap \gamma_j \not= \emptyset$.

Due to the relaxation of constraint~2 in the forward phase, we cannot ensure that $u_ju_{j+1}$ intersects $c_{j,1}$, but we can bound $d(u_ju_{j+1},c_{j,1})$ and $d(u_ju_{j+1},c_{j,2})$.  We can also bound the Fr\'{e}chet distance between an edge of $\tau_i$ and the subcurve of $\sigma$ matched to it according to $\pi_i$.  

\begin{lemma}\label{Lemma: Effect_Rel_set}
		For all $j \in [l-1]$, there exist points $p,q \in u_ju_{j+1}$ such that $p \leq_{\sigma} q$ and both $d(p,c_{j,1})$ and $d(q,c_{j,2})$ are at most $\sqrt{d}\eps\delta_{\max}$.
\end{lemma}

\begin{lemma}\label{Lemma: ver_sig_match}
	Take any $i \in [n]$ and any $a \in [m-1]$.  Suppose that $[k_1,k_2] = \{ j : \pi_i(a) < j \leq \pi_i(a+1)\}$ is non-empty.  There exist points $p,q \in \tau_{i,a}$ such that $d_F(pq,\sigma[u_{k_1},u_{k_2}]) \leq \delta_i + 4\sqrt{d}\eps\delta_{\max}$.
\end{lemma}

%We are ready to bound $d_F(\sigma,\tau_i)$ for all $i \in [n]$.
	
\begin{lemma}
	For all $i \in [n]$, $d_F(\sigma, \tau_i)\le \delta_i+4\sqrt{d}\epsilon\delta_{\max}$.
\end{lemma}
\begin{proof}
To prove the lemma, we define a matching $f_i$ from $\tau_i$ to $\sigma$ as follows.  First, we define $f_i^{-1}$ at the vertices of $\sigma$.  Take any $a \in [m-1]$.  If~$\{j : \pi_i(a) < j \leq \pi_i(a+1)\}$ is some non-empty $[k_1,k_2]$, by Lemma~\ref{Lemma: ver_sig_match}, there is a Fr\'{e}chet matching $g_{i,a}$ from $\sigma[u_{k_1},u_{k_2}]$ to a segment $pq \subseteq \tau_{i,a}$ such that $d_{g_{i,a}}(pq,\sigma[u_{k_1},u_{k_2}]) \leq \delta_i + 4\sqrt{d}\eps\delta_{\max}$; we define $f_i^{-1}(u_j) = g_{i,a}(u_j)$ for all $j \in [k_1,k_2]$.  Repeating the above for all $a \in [m-1]$ makes $f_i^{-1}(u_j) \leq_{\tau_i} f_i^{-1}(u_k) \iff j \leq k$.  Moreover, for every $j \in [l]$ and every $x \in f_i^{-1}(u_j)$, $d(u_j,x) \leq \delta_i + 4\sqrt{d}\eps\delta_{\max}$.  Any unprocessed $u_j$ must satisfy $j > \pi_i(m)$; therefore, $u_j$ should be matched to $v_{i,m}$, and we define $f_i^{-1}(u_j) = \{v_{i,m}\}$.
		
Next, we define $f_i$ at the vertices of $\tau_i$.  First, define $f_i(v_{i,a}) = u_1$ for all $a \in \pi_i^{-1}(0)$.  Take any $j \in [l-1]$ such that $\pi_i^{-1}(j)$ is some non-empty $[a,b]$.  Let $x$ and $y$ be two vertices of $c_{j,1}$ and $c_{j,2}$, respectively.  By constraint~1, there exist $pq \subseteq xy$ such that $d_F(pq,\tau_i[v_{i,a},v_{i,b}]) \leq \delta_i + 2\sqrt{d}\eps\delta_i$.  Let $h_i$ be a Fr\'{e}chet matching from $\tau_i[v_{i,a},v_{i,b}]$ to $pq$.
		
By Lemma~\ref{Lemma: Effect_Rel_set}, there exist points $p',q' \in u_ju_{j+1}$ such that $p' \leq_{\sigma} q'$ and both $d(p',c_{j,1})$ and $d(q',c_{j,2})$ are at most $\sqrt{d}\eps\delta_{\max}$.  It follows that $d(x,p') \leq 2\sqrt{d}\eps\delta_{\max}$ and $d(y,q') \leq 2\sqrt{d}\eps\delta_{\max}$.  As $pq \subseteq xy$, we can take the linear interpolation $h$ from the oriented segment $xy$ to the oriented segment $p'q'$, which ensures that $d(z,h(z)) \leq 2\sqrt{d}\eps\delta_{\max}$ for every point $z \in pq$.

For every $c \in \pi_i^{-1}(j)$, we define $f_i(v_{i,c}) = h \circ h_i(v_{i,c})$.  The Fr\'{e}chet matching $h_i$ and the linear interpolation $h$ guarantee that $f_i(v_{i,c}) \leq_{\sigma} f_i(v_{i,c'}) \iff c \leq c'$.  According to the previous discussion, for every $c \in \pi_i^{-1}(j)$, $d(v_{i,c},f_i(v_{i,c})) \leq \delta_i + 4\sqrt{d}\eps\delta_{\max}$.  Repeating the above for all $j \in [l-1]$ such that $\pi_i^{-1}(j) \not= \emptyset$ defines $f_i$ at all vertices of $\tau_i$.

Thanks to the property of $\pi_i$, the definitions of $f_i^{-1}$ at the vertices of $\sigma$ and the definitions of $f_i$ at the vertices of $\tau_i$ do not cause any conflict or order violation along $\sigma$ and $\tau_i$.
	
We have taken care of the vertices of $\sigma$ and $\tau_i$.  We use linear interpolation to match all other points between $\tau_i$ and $\sigma$;  it also maintains the distance bound of  $\delta_i+4\sqrt{d}\eps\delta_{\max}$.
\end{proof}

\cancel{
\begin{lemma}
	\label{kem:simplify}
	Let $l$ be a fixed positive integer.   Let $T=\{\tau_1,...,\tau_n\}$ be a set of $n$ curves, each containing $m$ curves.  Let $\Delta = \{\delta_1,...,\delta_n\}$ be a set of prescribed error thresholds.  Let $\Gamma$ denote the number of configurations.  There is an approximation algorithm for $Q(T,\Delta,l)$ that can return in $O(\Gamma \cdot (mn\log m + l\eps^{-d^2/2}n + ln^{(d-1)(d+2)/4}))$ time a null output or a curve $\sigma$ of $l$ vertices.  	If a null output is returned, then $Q(T,\Delta,l)$ has no solution; otherwise, $d_F(\sigma,\tau_i) \leq \delta_i + \eps\delta_{\max}$ for all $i \in [n]$;
\end{lemma}
\begin{proof}
	The correctness of the algorithm follows from the previous discussion.  It remains to analyze the running time.  Let $\Psi$ be a configuration.  As described in Section~\ref{}, the time to check constraint~1 is $O(mn\log m)$.  By Lemma~\ref{lem:forward}, the computation of $\gamma_1,\ldots,\gamma_l$ takes $O(l\eps^{-d^2/2} n + ln^{(d-1)(d+2)/4})$ time.  By Lemma~\ref{lem:backward}, the backward extraction takes $O(l)$ time.  Therefore, processing one configuration takes $O(mn\log m + l\eps^{-d^2/2}n + ln^{(d-1)(d+2)/4})$ time.  Multiply this time bound with the number of configurations gives the total running time.
\end{proof}
}
	
\subsection{Accelerating the algorithm}
\label{sec:accelerate}

Observe that a configuration only needs to guarantee that the endpoints of the segments $\gamma_1,\ldots,\gamma_l$ can be produced by the components $\mathcal{C}$, $\mathcal{S}$, $\mathcal{A}$, and other linear constraints induced by the input curves.  As there are only $2l$ segment endpoints, only $O(l)$ input curves are involved in defining $\gamma_1,\ldots,\gamma_l$.  In Appendix~\ref{app:reduce}, we show that $5l$ are sufficient.  

We do not know which $5l$ input curves to sample, so we enumerate all subsets of $5l$ input curves.  For each subset, the number of configurations drops to $O\bigl(m^{O(\ell^2)} \cdot (\ell/\eps)^{O(d\ell)}\bigr)$, and the running time of the forward construction reduces to $O(ml\log m + ml^22^{O(d)})$.   For each candidate output curve $\sigma$, we need to verify whether $\sigma$ works for the remaining $n-5l$ input curves that are not used for constructing $\sigma$.  We check by computing $d_F(\sigma,\tau_i)$ and comparing it with $\delta_i + 4\sqrt{d}\eps\delta_{\max}$ for all $i \in [n]$.  The time needed for this check is $O(lmn\log (ml))$.  To go from the error bounds $\delta_i + 4\sqrt{d}\eps\delta_{\max}$ to $\delta_i + \eps\delta_{\max}$, we need to reduce $\eps$ to $\eps/(4\sqrt{d})$ in the definition of $\grid_1$ and $\grid_2$.  Finally, we repeat the above for each $l \in [\ell]$ in order to solve $Q(T,\Delta,\ell)$ approximately.

\begin{theorem}
	\label{thm:simplify}
	Let $T = \{\tau_1,\ldots,\tau_n\}$ be $n$ polygonal curves of $m$ vertices each in $\real^d$.  Let $\Delta = \{\delta_1,\ldots,\delta_n\}$ be $n$ error thresholds.  Let $\ell \leq m$ be a positive integer.  Let $\eps$ be a fixed value in $(0,1)$.  There is an algorithm that returns a null output or a polygonal curve $\sigma$ of at most $\ell$ vertices such that $d_F(\sigma,\tau_i) \leq \delta_i + \eps\delta_{\max}$ for $i \in [n]$.  If the output is null, there is no curve $\sigma$ of at most $\ell$ vertices such that $d_F(\sigma,\tau_i) \leq \delta_i$ for $i \in [n]$.  The running time of the algorithm is $\tilde{O}\bigl(n^{O(\ell)} \cdot m^{O(\ell^2)} \cdot (d\ell/\eps)^{O(d\ell)}\bigr)$.
\end{theorem}

When $n=1$, we can use Theorem~\ref{thm:simplify}---the version without picking subsets of $5l$ input curves---to approximately minimize both the error and the output size in curve simplification.

\begin{theorem}
	\label{thm:first}
	Let $\tau$ be a polygonal curve of $m$ vertices in $\real^d$.  Let $\delta$ be an error threshold.  Let $\alpha$ and $\eps$ be some fixed values in $(0,1)$.  There is an algorithm that computes a polygonal curve $\sigma$ such that $d_F(\sigma,\tau) \leq (1+\eps)\delta$ and $|\sigma| \leq (1+\alpha)\kappa(\tau,\delta)$.  The running time is $\tilde{O}\bigl((m^{O(1/\alpha)} \cdot (d/(\alpha\eps))^{O(d/\alpha)}\bigr)$.
\end{theorem}
\begin{proof}
Let $\tau = (v_1,\ldots,v_m)$.  Compute the smallest $i \in [m]$ such that $Q(\{\tau[v_1,v_{i+1}]\},\{\delta\},\lceil 1/\alpha \rceil)$ has no solution.  We have spent $\tilde{O}\bigl(i \cdot m^{O(1/\alpha)} \cdot (d/(\alpha\eps))^{O(d/\alpha)}\bigr)$ time so far and obtain a curve $\sigma_1$ such that $|\sigma_1| \leq 1/\alpha$ and $d_F(\sigma_1,\tau[v_1,v_i]) \leq (1+\eps)\delta$.    Then we repeat the above for $\tau[v_{i+1},v_m]$ to obtain another curve $\sigma_2$.  In the end, we connect $\sigma_1, \sigma_2,\ldots$ to form the output curve $\sigma$.  The distance between $v_i$ and the last endpoint of $\sigma_1$ is at most $(1+\eps)\delta$, and so is the distance between $v_{i+1}$ and the first endpoint of $\sigma_2$; therefore, a linear interpolation shows that we can connect $\sigma_1$ and $\sigma_2$ without violating the $(1+\eps)\delta$ Fr\'{e}chet distance bound.  The same analysis applies to the connections between $\sigma_2$ and $\sigma_3$ and so on.  The greedy process means that we introduce at most one extra edge for every $1/\alpha$ edges in the optimal solution, implying that $|\sigma| \leq (1+\alpha)\kappa(\tau,\delta)$.  
\end{proof}

\section{$\pmb{(k,\ell)}$-median clustering}
\label{sec:median}

Take any $\alpha, \beta \in [1,\infty)$ and $\mu \in (0,1)$.  An algorithm is an \emph{$(\alpha+\eps)$-approximate candidate finder} with success probability at least $1-\mu$ if it computes a set $\Sigma$ of curves, each of $\ell$ vertices, and for every subset $S \subseteq T$ that has size $\frac{1}{\beta}|T|$ or more, it holds with probability at least $1-\mu$ that $\Sigma$ contains an $(\alpha+\eps)$-approximate $(1,\ell)$-median of $S$.  The following result of Buchin~et~al.~\cite{buchin2021approximating} says that a finder can be used for $(k,\ell)$-median clustering.  We use $\mathrm{cost}(T,\Sigma)$ to denote $\sum_{\tau_i \in T} \min_{\sigma \in \Sigma} d_F(\sigma,\tau_i)$.  If $\Sigma$ consists of a single curve $\sigma$, we will just write $\mathrm{cost}(T,\sigma)$ for $\sum_{\tau_i \in T} d_F(\sigma,\tau_i)$.

\begin{lemma}[Theorem 7.2~\cite{buchin2021approximating}]\label{Lemma: Framework}
	One can use an $(\alpha+\varepsilon)$-approximate candidate finder $A$ with success probability at least $1-\mu$ to compute a set $\Sigma$ of $k$ curves, each of $\ell$ vertices, such that $\mathrm{cost}(T, \Sigma) \leq (1+\frac{4k^2}{\beta-2k})(\alpha+\varepsilon) \cdot \mathrm{OPT}$ with probability at least $1-\mu$, where $\mathrm{OPT}$ is the optimal $(k,\ell)$-median cost for $T$.  The running time is $O\bigr(T_A \cdot C_A^{k+1} + n \cdot C_A^{k+2}\bigr)$, where $T_A$ is the running time of $A$ and $C_A$ is the number of curves returned by $A$.
\end{lemma}

We will present a $(1+\eps$)-approximate candidate finder such that $T_A$ and $C_A$ are $\tilde{O}\bigl((m/\mu)^{O(\ell)} \cdot  (d\beta\ell/\eps)^{O((d\ell/\eps)\log(1/\mu))}\bigr)$.
%linear in $nm^{\mathrm{poly}(d,\ell,\ln(1/\eps))}$ and exponential in $\mathrm{poly}(d,\ell,\beta,1/\eps,\ln(1/\mu))$, assuming that 
Using our finder with $\beta = \Theta(k^2/\eps)$ and adjusting $\eps$ by a constant factor, Lemma~\ref{Lemma: Framework} gives a $(1+\eps)$-approximation algorithm for the $(k,\ell)$-median clustering problem.  
%Then Lemma~\ref{Lemma: Framework} yields a $1 + O(\eps)$ approximation ration which can be decreased to $1+\eps$ by adjusting constant factors.  

Our finder makes heavy use of the configurations in Section~\ref{sec:discrete}.  Some notations are needed for the exposition.  Let $C(R,l,\alpha)$ be the set of all configurations with respect to a subset $R \subset T$, the target size $l$ of the simplified curve, and the approximation ratio $\alpha \in (0,1)$.    There is a given set of error thresholds for $R$ that we do not specify explicitly in order not to clutter the notation.  
%The error thresholds for $R$ may be inherited from those fixed for $T$ or specified particularly for $R$; 
It will be clear from the context what these error thresholds are.

We enhance the finder in~\cite{buchin2021approximating} to enumerate certain configurations and compute the corresponding curves using the two-phase construction.  
%These curves constitute the set $C$ in Lemma~\ref{Lemma: Framework}.  
Algorithm~\ref{ALG:1_APP} shows this finder.  Since we aim for a probabilistic result, we sample a subset $Y \subset T$ of size $\frac{80\beta\ell}{\eps}\ln\frac{80\ell}{\mu}$ and work with the configurations for all subsets of $Y$ of size $\frac{1}{2\beta}|Y|$.  This will allow us to capture $\Theta(l)$ input curves that induce almost all configurations necessary.  This is formalized in Lemma~\ref{lem:first-level} below.  

\begin{lemma}
	\label{lem:first-level}
Take any subset $S  \subseteq T$ with at least $n/\beta$ curves for any $\beta \geq 1$.  Let $\Delta = \{\delta_1,\ldots,\delta_{|S|}\}$ be a set of error thresholds for $S$ such that $Q(S,\Delta,\ell)$ has a solution.   Relabel elements, if necessary, so that $\delta_1 \leq  \ldots \leq \delta_{|S|}$.  For $i \in \bigl[|S|\bigr]$, let $S_i = \{\tau_i, \tau_{i+1},\ldots \}$ and let $\Delta_i = \{\delta_i,\delta_{i+1},\ldots\}$.   Take any $\alpha, \eps \in (0,1)$ and any $r \in \bigl[\frac{\eps|S|}{5\ell}\bigr]$.  There exists $\mathcal{H}_r \subseteq 2^{S_r}$ such that $|\mathcal{H}_r| = 3l-1$ for some $l \in [\ell]$, every subset in $\mathcal{H}_r$ has $\frac{\eps|S|}{5\ell}$ curves, and for every subset $R \subseteq S_r$, if $\tau_r \in R$ and $R \cap H \not= \emptyset$ for all $H \in \mathcal{H}_r$, there exist a configuration $\Psi = (\mathcal{P},\mathcal{C},\mathcal{S},\mathcal{A}) \in C(\overline{\mathcal{H}}_r \cup R ,h,\alpha)$ and a curve $\sigma = (w_1,\ldots,w_h)$ for some $h \in [l]$, where $\overline{\mathcal{H}}_r = S_r \setminus \bigcup_{H \in \mathcal{H}_r} H$, that satisfy the following properties.
\begin{enumerate}[{\em (i)}]
\item $\Psi$ and $\sigma$ satisfy constraints~1--3 with respect to $\overline{\mathcal{H}}_r \cup R$. 
\item There exists a configuration in $C(R,h,\alpha)$ that shares the components $\mathcal{C}$, $\mathcal{S}$, and $\mathcal{A}$ with $\Psi$.
\end{enumerate}
%One subset in $\mathcal{H}_S$ contains the $\frac{\eps n}{5\beta\ell}$ curves in $S$ that have the smallest error thresholds for any fixed value $\beta \geq 1$.  The other subsets in $\mathcal{H}_S$ contain $\frac{\eps |S|}{5\ell}$ curves each.
\end{lemma}

Lemma~\ref{lem:first-level} is formulated for different $r \in \bigl[\frac{\eps |S|}{5\ell}\bigr]$ because the sample of curves in Algorithm~\ref{ALG:1_APP} will include some curves among $\tau_1,\ldots, \tau_{\eps |S|/(5\ell)}$ with good probability, but we do not know \emph{a priori} which ones.  By the Chernoff bound, we can sample a small subset $R \subseteq S_r$ that satisfies Lemma~\ref{lem:first-level} with high probability.  The set $R$ acts like the $\Theta(l)$ curves that induce the components of a configuration in  Section~\ref{sec:accelerate}.  Indeed, the proof of Lemma~\ref{lem:first-level} in Appendix~\ref{app:first-level} uses a similar argument.  How should the error thresholds for $R$ be set?  In lines~\ref{alg:34} and~\ref{alg:LU}, Algorithm~\ref{ALG:1_APP} computes a 34-approximate $(1,\ell)$-median $c$, with success probability at least $1 - \mu/4$, to identify an upper bound $U$ and a lower bound $L$ on the error thresholds.  Then, we try all possible sets of integral multiples of $L$ in the range $[L,U]$ in line~\ref{alg:error}.   There are at most $3\eps |S|/5$ curves in $S_r \setminus (\overline{\mathcal{H}}_r \cup R)$ that Lemma~\ref{lem:first-level}(i) says nothing about; the analysis will take care of them separately.  The specific values $n/\beta$ and $\frac{\eps |S|}{5\ell}$ are not critical for the proof of Lemma~\ref{lem:first-level}.  They are chosen to interface with the subsequent analysis of the approximation ratio of Algorithm~\ref{ALG:1_APP}.

\begin{algorithm}[h]
	\SetAlgoLined
	\KwData{$T=\{\tau_1, ..., \tau_n\}, \, \ell \in \mathbb{Z}_{>0}, \, \beta\ge 1, \, \text{and}\; \mu, \varepsilon\in(0, 1)$.}
	\KwResult{A set $\Sigma$ of curves, each of $\ell$ vertices; for every subset $S \subseteq T$ of size $\frac{1}{\beta}|T|$ or more, it holds with probability at least $1-\mu$ that there exists a curve $\sigma \in \Sigma$ such that $\mathrm{cost}(S,\sigma) \leq (1+O(\sqrt{d}\ell\eps)) \cdot \text{optimal $(1,\ell)$-median cost of $S$}$.}
	\Begin{
		$\Sigma \gets \emptyset$\;
		$Y \gets$ a multiset of $\bigl\lceil \frac{80\beta\ell}{\eps}\ln\frac{80\ell}{\mu}\bigr\rceil$ curves that are uniformly, independently sampled from $T$ with replacement\;  \label{alg:sample}
		\For{\emph{each multi-subset} $X \subseteq Y$ \emph{such that} $|X|= |Y|/(2\beta)$}{ \label{alg:enum}
			$c \gets $ $(1,\ell)$-median-34-approximation$(X,\mu/4)$;  \quad\quad \text{/* Algorithm 1 in \cite{buchin2021approximating} */} \label{alg:34} \\
			$\Sigma \gets \Sigma \cup\{c\}$\;
			$U \gets \frac{10\ell}{\eps^2}\mathrm{cost}(X, c)$; $L \gets \frac{\eps\mu}{34|X|}\mathrm{cost}(X, c)$\;   \label{alg:LU}
			\For{\emph{each} $l \in [\ell]$ \emph{and each} $h \in [l]$} { \label{alg:index}
				\For{\emph{each possible simple subset} $W \subseteq X$ \emph{of size at most} $3l+2h$}{ \label{alg:subset}
					\For{\emph{each possible set} $\Delta_W =\{ \delta_i : \tau_i \in W, \, \delta_i = b_iL \,\, \text{\emph{for some $b_i \in \bigl[\lceil U/L\rceil\bigr]$}}\}$}{  \label{alg:error}
						$\Sigma' \gets$ curves produced by the two-phase method on all configurations in $C(W,h,\eps^2)$ with respect to $\Delta_W$; note that the approximation ratio is $\eps^2$ instead of $\eps$\;  \label{alg:curve}
						\If{$h <\ell$}{
							\For{\emph{each} $\sigma \in \Sigma'$}{
								arbitrarily make $\ell-h$ points on $\sigma$ as extra vertices\;	
							}
						}		
						$\Sigma \gets \Sigma \cup \Sigma'$
					}
				}
			}   \label{alg:end-itr}
		}
	}
	\caption{$(1+\varepsilon)$-approximate candidate-finder}\label{ALG:1_APP}
\end{algorithm}

Lemma~\ref{lem:first-level}(ii) does not immediately allow us to use $R$ to approximate an optimal $(1,\ell)$-median curve for an arbitrary subset $S \subseteq T$.  To produce a candidate curve using the two-phase construction, we need to know the curves in $S$ so that we can apply constraints~1--3 (or their relaxations) in the forward construction.  Although we do not know $S$, sampling comes to our rescue.  Lemma~\ref{lem:second-level} below shows that a small sample can capture the configurations and the effects of the two-phase construction for almost the entire $S$.   %To facilitate the discussion, 
We introduce a notation for the output of the forward construction.  Given a subset $Z \subseteq T$ and a configuration $\Psi \in C(Z,l,\alpha)$, the forward construction produces $l$ line segments; we use $\gamma_j(Z,\Psi)$ to denote the $j$-th output segment for $j \in [l]$.

\begin{lemma}
	\label{lem:second-level}
	Take any subset $S \subseteq T$ with at least $n/\beta$ curves for any $\beta \geq 1$.   Let $\Delta$ be a set of error thresholds for $S$ such that $Q(S,\Delta,\ell)$ has a solution.   Assume the notation in Lemma~\ref{lem:first-level}.  Take any $\alpha, \eps \in (0,1)$, any $r \in \bigl[\frac{\eps|S|}{5\ell}\bigr]$, and any subset $R \subseteq S_r$ such that $\tau_r \in R$ and $R \cap H \not= \emptyset$ for all $H \in \mathcal{H}_r$.  Let $\hat{S}_r = \overline{\mathcal{H}}_r \cup R$.  Let $\Psi = (\mathcal{P},\mathcal{C},\mathcal{S},\mathcal{A}) \in C(\hat{S}_r,h,\alpha)$ for some $h \in [\ell]$ be a configuration that satisfies Lemma~\ref{lem:first-level}.  There exist $\mathcal{H}_{\Psi} \subseteq 2^{\hat{S}_r}$ and a configuration $\Psi'' = (\mathcal{P}'',\mathcal{C},\mathcal{S},\mathcal{A}) \in C(\overline{\mathcal{H}}_{\Psi} \cup R,h,\alpha)$, where $\overline{\mathcal{H}}_{\Psi} = \hat{S}_r \setminus \bigcup_{H \in \mathcal{H}_{\Psi}} H$, such that $|\mathcal{H}_{\Psi}| = 2h$, every subset in $\mathcal{H}_{\Psi}$ contains $\frac{\eps|S|}{5\ell}$ curves, and for all subset $R' \subseteq \hat{S}_r$, if $R' \cap H \not= \emptyset$ for all $H \in \mathcal{H}_{\Psi}$, then there exists a configuration $\Psi' = (\mathcal{P}',\mathcal{C},\mathcal{S},\mathcal{A}) \in C(R \cup R',h,\alpha)$ that satisfies the following properties.
	\begin{enumerate}[{\em (i)}]
		\item For all $j \in [h]$, $\gamma_j(\hat{S}_r,\Psi) \subseteq \gamma_j(R\cup R',\Psi') \subseteq \gamma_j(\overline{\mathcal{H}}_{\Psi} \cup R,\Psi'')$.
		\item The backward extraction using $\{\gamma_j(R\cup R',\Psi') : j \in [h]\}$ produces a curve $\sigma$ such that $d_F(\sigma,\tau_i) \leq \delta_i + 4\sqrt{d}\alpha\cdot \max\{\delta_i : \tau_i \in R \cup R' \}$ for all $\tau_i \in \overline{\mathcal{H}}_{\Psi} \cup R$.
	\end{enumerate}
%The size of $\mathcal{H}_{\Psi}$ is $2h$.  Each subset of $\hat{S}$ in $\mathcal{H}_{\Psi}$ contains $\frac{\eps|S|}{5\ell}$ curves.
\end{lemma}

It is an important feature of Lemma~\ref{lem:second-level} that $\Psi$, $\Psi'$ and $\Psi''$ share the components $\mathcal{C}$, $\mathcal{S}$ and $\mathcal{A}$.  Lemma~\ref{lem:second-level} serves the following purpose.   The much smaller subsets $R$ and $R'$ give line segments $\gamma_j(R \cup R',\Psi')$ that capture the feasible locations of output curve vertices for the much bigger set $\hat{S}_r = \overline{\mathcal{H}}_r \cup R$.  However, this property alone does not say anything about the approximation offered by $\gamma_j(R \cup R',\Psi')$ with respect to the curves in $\hat{S}_r$.  

%Note that $\overline{\mathcal{H}}_{\Psi} \cup R \subseteq \hat{S}$ and there are $(1-\eps)|S|$ or more curves in $\overline{\mathcal{H}}_{\Psi} \cup R$.  
The two-phase construction using $\{\gamma_j(\overline{\mathcal{H}}_{\Psi} \cup R,\Psi'') : j \in [h]\}$ will produce an output curve $\sigma$ such that $d_F(\sigma,\tau_i) \leq \delta_i + 4\sqrt{d}\alpha \cdot \max\{\delta_i : \tau_i \in \overline{\mathcal{H}}_{\Psi} \cup R\}$ for all $\tau_i \in \overline{\mathcal{H}}_{\Psi} \cup R$.  One effect of $\Psi'$ and $\Psi''$ sharing the component $\mathcal{A}$ is that $\max\{\delta_i  : \tau_i \in \overline{\mathcal{H}}_{\Psi} \cup R\} = \max\{\delta_i : \tau_i \in R \cup R'\}$.  Recall that the backward extraction can start with any point in $\gamma_h(\overline{\mathcal{H}}_{\Psi} \cup R,\Psi'')$ as the $h$-th output vertex $u_h$, and for $j \in [h-1]$, the $j$-th output vertex $u_j$ can be any point in $F(c_{j,2},u_{j+1}) \cap \gamma_j(\overline{\mathcal{H}}_{\Psi} \cup R,\Psi'')$.  We can start with $u_h \in \gamma_h(R\cup R',\Psi') \subseteq \gamma_h(\overline{\mathcal{H}}_{\Psi} \cup R,\Psi'')$.  For $j \in [h-1]$, since $\mathcal{C}$ is shared by $\Psi'$ and $\Psi''$, we can pick $u_j$ from $F(c_{j,2},u_{j+1})  \cap \gamma_j(R\cup R',\Psi') \subseteq F(c_{j,2},u_{j+1}) \cap \gamma_j(\overline{\mathcal{H}}_{\Psi} \cup R,\Psi'')$.  Therefore, the same approximation guarantees apply to the backward extraction using $\gamma_j(R\cup R',\Psi')$.  The proof of Lemma~\ref{lem:second-level} is in Appendix~\ref{app:second-level}.  The values $n/\beta$ and $\frac{\eps |S|}{5\ell}$ are not critical for the proof of Lemma~\ref{lem:second-level}.  They are chosen to interface with Lemma~\ref{lem:first-level} and the analysis of the approximation ratio of Algorithm~\ref{ALG:1_APP}.

Lemma~\ref{lem:second-level} allows us to approximate using the small subset $R \cup R'$, provided that we know the right error thresholds.  Of course, we do not, but we will encounter the right choice in the enumeration of the possible error thresholds in line~\ref{alg:error} of Algorithm~\ref{ALG:1_APP}.  The capability of  Algorithm~\ref{ALG:1_APP} is analogous to that provided by the approximate candidate finder of Buchin~et~al.~\cite{buchin2021approximating}; the major difference being our guarantee that every output curve has $\ell$ vertices, whereas an upper bound of $2\ell-2$ is guaranteed in~\cite{buchin2021approximating}.  We can adapt the analysis in~\cite{buchin2021approximating} for their approximate candidate finder to show that Algorithm~\ref{ALG:1_APP} is a $(1+\eps)$-approximate candidate finder.  The details are given in Appendix~\ref{app:finder}.
 
 \begin{lemma}
 	\label{lem:finder}
 	For $\eps < 1/9$, Algorithm~\ref{ALG:1_APP} is a $(1+\eps)$-approximate candidate finder with success probability at least $1-\mu$.   The algorithm outputs a set $\Sigma$ of curves, each of $\ell$ vertices; for every subset $S \subseteq T$ of size $\frac{1}{\beta}|T|$ or more, it holds with probability at least $1-\mu$ that there exists a curve $\sigma \in \Sigma$ such that $\mathrm{cost}(S,\sigma) \leq (1+\eps)\mathrm{cost}(S,c^*)$, where $c^*$ is the optimal $(1,\ell)$-median of $S$.  The running time and output size of Algorithm~\ref{ALG:1_APP} are $\tilde{O}\bigl(m^{O(\ell^2)} \cdot \mu^{-O(\ell)} \cdot (d\beta\ell/\eps)^{O((d\ell/\eps)\log(1/\mu))}\bigr)$.
 \end{lemma}

Combining Lemmas~\ref{lem:finder} and~\ref{Lemma: Framework} gives the following result.
 
 \begin{theorem}
 	Let $T$ be a set of $n$ polygonal curves with $m$ vertices each in $\real^d$.  For any $k,\ell \in \mathbb{Z}_{>0}$ and any $\mu,\eps \in (0,1)$, one can compute a set $\Sigma$ of $k$ curves, each of $\ell$ vertices, and it holds with probability at least $1-\mu$ that $\mathrm{cost}(T,\Sigma)$ is within a factor $1+\eps$ of the optimal $(k,\ell)$-median cost of $T$.  The running time is
 	$\tilde{O}\bigl(n \cdot m^{O(k\ell^2)} \cdot \mu^{-O(k\ell)} \cdot (dk\ell/\eps)^{O((dk\ell/\eps)\log(1/\mu))}\bigr)$.
  \end{theorem}

\newpage

\bibliography{ref.bib}
\bibliographystyle{plain}

\newpage

\appendix

\section{Proof of Lemma~\ref{Lemma: valid_cons_set}}
\label{app:Lemma: valid_cons_set}

\noindent\textbf{Lemma~\ref{Lemma: valid_cons_set}.}~~\emph{If $\Q(T, \Delta,\ell)$ has a solution, there exist a configuration $\Psi_l$ and a curve $\sigma = (w_1,\ldots,w_l)$ for some $l \in [\ell]$ such that constraints~1--3 are satisfied and $d_F(\tau_{i}, \sigma)\le \delta_{i}+ \sqrt{d}\eps\delta_{\min}$ for $i\in [n]$.}

\vspace{2pt}

\begin{proof}
Let $\sigma = (u_1,\ldots,u_l)$ be a solution curve for $Q(T,\Delta,\ell)$ for some $l \in [\ell]$.  For $i \in [n]$, let $g_i$ be a Fr\'{e}chet matching from $\tau_i$ to $\sigma$.  So $d_{g_i}(\tau_i,\sigma) \leq \delta_i$ for all $i \in [n]$.

We require further that for all $j \in [l-1]$, there exist $i \in [n]$ and $a \in [m]$ such that $g_i(v_{i,a}) \cap u_ju_{j+1} \not= \emptyset$.  Suppose that this requirement is not met for $u_ju_{j+1}$.  Then, for every $i \in [n]$, there exists $a_i \in [m-1]$ such that $u_ju_{j+1} \subseteq g_i(\mathrm{int}(\tau_{i,a_i}))$.  Let $p$ be the maximum of $\{\max(g_i(v_{i,a_i})) : i \in [n]\}$ with respect to $\leq_{\sigma}$.  Let $q$ be the minimum of $\{\min(g_i(v_{i,a_i+1})) : i \in [n]\}$ with respect to $\leq_{\sigma}$.  Our choice of $p$ and $q$ means that $p \in g_s(v_{s,a_s})$ and $q \in g_t(v_{t,a_{t}+1})$ for some possibly non-distinct $s,t \in [n]$, and $\mathrm{int}(\sigma[p,q]) \subseteq g_i(\mathrm{int}(\tau_{i,a_i}))$ for all $i \in [n]$.  We update $\sigma$ by substituting $\sigma[p,q]$ with the edge $pq$, possibly making $p$ and $q$ new vertices of $\sigma$.  The number of edges of $\sigma$ is not increased by the replacement; it  may actually be reduced.  For all $i \in [n]$, we update $g_i$ by a linear interpolation along $pq$; since $\mathrm{int}(\sigma[p,q]) \subseteq g_i(\mathrm{int}(\tau_{i,a_i}))$, the replacement of $\sigma[p,q]$ and the linear interpolation ensure that after the update, $g_i$ remains a matching and $d_{g_i}(\tau_i,\sigma) \leq \delta_i$.  Our choice of $p$ and $q$ means that the update does not affect the subset of vertices of $\tau_i$ that are matched by any $g_i$ to the edges of $\sigma$ other than $pq$.  The update also ensures that $pq$ will not trigger another shortcutting, and $pq$ will not be shortened by other shortcuttings.  If necessary, we repeat the above to convert $\sigma$ to $(u'_1,u'_2,\ldots)$ such that for every edge $u'_ju'_{j+1}$ there exist $i \in [n]$ and $a \in [m]$ such that $g_i(v_{i,a}) \cap u'_ju'_{j+1} \not= \emptyset$.
%one of the following properties is satisfied: (i)~$g_i(v_{i,a}) \cap \mathrm{int}(u'_ju'_{j+1}) \not= \emptyset$ for some $i \in [n]$ and some $a \in [m]$, or (ii)~$u'_j  \in g_i(v_{i,a}) \, \wedge \, u'_{j+1} \in g_{i'}(v_{i',b})$ for some possibly non-distinct $i,i' \in [n]$ and some possibly non-distinct $a, b \in [m]$.
	
The modified $\sigma$ may have fewer vertices; if so, we decrease the value of $l$ to the current number of vertices in $\sigma$.  We use $(u'_1,u'_2,\ldots,u'_l)$ to denote the modified $\sigma$.  Let $\min = \mathrm{argmin}_{i \in [n]} \delta_i$.  Next, we snap the vertices of $\sigma$ to segments in the set $\cal L$.  Since $d_{g_{\min}}(\sigma,\tau_{\min}) \leq \delta_{\min}$, $\sigma$ lies inside $\tau_{\min} \oplus B_{\delta_{\min}}$, which implies that every vertex $u'_j$ of $\sigma$ is within a distance of $\sqrt{d}\eps\delta_{\min}$ from the nearest segment $s_j$ in $\cal L$.  We modify $\sigma$ by moving $u'_j$ to its nearest point $w_j \in s_j$ for every $j \in [l]$.  For every $i \in [n]$, we update $g_i$ using the linear interpolations between $u'_ju'_{j+1}$ and $w_jw_{j+1}$ for all $j \in [l-1]$.  Consequently, $\sigma = (w_1,\ldots,w_l)$ and $d_{g_i}(\tau_i,\sigma) \leq \delta_i + \sqrt{d}\eps\delta_{\min}$ for $i \in [n]$.   This establishes the second half of the lemma.  Notice that this updating of the $g_i$'s preserves the property that for all $j \in [l-1]$, there exist $i \in [n]$ and $a \in [m]$ such that $g_i(v_{i,a}) \cap w_jw_{j+1} \not= \emptyset$.
%one of the following properties is satisfied: (i)~$g_i(v_{i,a}) \cap \mathrm{int}(w_jw_{j+1}) \not= \emptyset$ for some $i \in [n]$ and some $a \in [m]$, or (ii)~$w_j  \in g_i(v_{i,a}) \, \wedge \, w_{j+1} \in g_{i'}(v_{i',b})$ for some possibly non-distinct $i,i' \in [n]$ and some possibly non-distinct $a, b \in [m]$.

We construct a configuration $\Psi_l$ that satisfies constraints~1--3 together with $\sigma$.
	
For every $i \in [n]$ and every $a \in [m]$, if $w_1 \in g_i(v_{i,a})$, we define $\pi_i(a) = 0$; otherwise, let $j \in [l-1]$ be the smallest index such that $g_i(v_{i,a}) \cap w_jw_{j+1} \not= \emptyset$ and we define $\pi_i(a) = j$.   This induces an $n$-tuple $\mathcal{P} = (\pi_i)_{i \in [n]}$ of the partitions of the vertices of $\tau_1,\ldots,\tau_n$ such that for all $i \in [n]$, $\pi_i(1) = 0$ and if $a \leq b$, then $\pi_i(a) \leq \pi_i(b)$.  %Moreover, for every $j \in [l-1]$, since there exist $i \in [n]$ and $a \in [m]$ such that $g_i(v_{i,a}) \cap w_jw_{j+1} \not= \emptyset$ and $w_j \not\in g_i(v_{i,a})$, it must be the case that $a \in \pi_i^{-1}(j)$.  As a result, for every $j \in [l-1]$, $\bigcup_{i \in [n]} \pi_i^{-1}(j) \not= \emptyset$ as required by the definition of $\mathcal{P}$.

Next, we define $\mathcal{C}$ as follows.  Take any $j \in [l-1]$.  Let $x_j$ and $y_j$ be the minimum and maximum points in $\bigcup_{i \in [n]} \bigcup_{a \in [m]} g_i(v_{i,a}) \cap w_jw_{j+1}$ with respect to $\leq_{\sigma}$, respectively.  Note that $\bigcup_{i \in [n]} \bigcup_{a \in [m]} g_i(v_{i,a}) \cap w_jw_{j+1}$ is non-empty because there exist $i \in [n]$ and $a \in [m]$ such that $g_i(v_{i,a}) \cap w_jw_{j+1} \not= \emptyset$.  By definition, $x_j \leq_{\sigma} y_j$.  There exists $i \in [n]$ such that $x_j$ is within a distance of $\delta_i + \sqrt{d}\eps\delta_{\min}$ from a vertex of $\tau_i$.   We can make the same conclusion about $y_j$.  It follows that $x_j$ and $y_j$ belong to cells in $\grid_1$.  Choose $c_{j,1}$ and $c_{j,2}$ to be any cells in $\grid_1$ that contain $x_j$ and $y_j$, respectively.   This gives the $(l-1)$-tuple $\mathcal{C} = ((c_{j,1},c_{j,2}))_{j \in [l-1]}$.

The components $\mathcal{S}$ and $\mathcal{A}$ of $\Psi_l$ are defined as follows.  By construction, we know that $w_j$ lies on some segment $s_j \in \mathcal{L}$.  We simply set $\mathcal{S} = (s_j)_{j \in [l]}$. For every $j \in [l]$, if $w_j$ lies in some grid cell in $\grid_2$, set $\mathcal{A}[j]$ to be that cell; otherwise, set $\mathcal{A}[j]$ to be null. 

We verify constraint~1.  Recall the definitions of $x_j$ and $y_j$ in defining $(c_{j,1},c_{j,2})$ for every $j \in [l-1]$.  Suppose that $\pi_i^{-1}(j)$ is some non-empty $[a,b]$.  It means that $g_i(\tau_i[v_{i,a},v_{i,b}]) \subseteq x_jy_j \subseteq w_jw_{j+1}$.  Let $x$ and $y$ be any vertices of $c_{j,1}$ and $c_{j,2}$, respectively.  Therefore, both $d(x,x_j)$ and $d(y,y_j)$ are at most $\sqrt{d}\eps\delta_i$.   Let $h$ be the linear interpolation from the oriented segment $x_jy_j$ to the oriented segment $xy$.  For any point $z \in \tau_i[v_{i,a},v_{i,b}]$, $d(z,h\circ g_i(z)) \leq d_{g_i}(\tau_i,\sigma) + \sqrt{d}\eps\delta_i \leq \delta_i + 2\sqrt{d}\eps\delta_i$.  It follows that $h \circ g_i$ is a matching from $\tau_i[v_{i,a},v_{i,b}]$ to some segment $pq \subseteq xy$ such that $d_{h \circ g_i}(\tau_i[v_{i,a},v_{i,b}],pq) \leq \delta_i + 2\sqrt{d}\eps\delta_i$.  This proves that constraint~1 is satisfied.

Constraint~2 is clearly satisfied by construction.

%We verify that constraint~2 is satisfied.  By construction, whenever $\pi_i(j)$ is equal to some non-empty $[v_{i,a},v_{i,b}]$ for some $a,b \in [m]$, we have $c_{j,1} \cap w_jw_{j+1} \leq_{\sigma} \{p_{ij},q_{ij}\} \leq_{\sigma} c_{j,2} \cap w_jw_{j+1}$.  For every vertex $x$ of $c_{j,1}$ and every vertex $y$ of $c_{j,2}$, both $d(x,c_{j,1} \cap w_jw_{j+1})$ and $d(y,c_{j,2} \cap w_jw_{j+1})$ are at most $\sqrt{d}\eps\delta_{\min}$.  It means that we can map $\tau_i[v_{i,a},v_{i,b}]$ by $g_i$ to $w_jw_{j+1}$ followed by a linear interpolation to $xy$ such that the additional error is at most $2\sqrt{d}\eps\delta_{\min}$.  Constraint~2 is thus satisfied.

We verify constraint~3(a) as follows.  Recall that $d_{g_i}(\tau_i,\sigma) \leq \delta_i + \sqrt{d}\eps\delta_{\min}$ for $i \in [n]$.  Since the cell $\mathcal{A}[1]$ has side length $\eps\delta_{\max}$, for every vertex $x$ of $\mathcal{A}[1]$, we have $d(v_{i,1}, x) \leq d(v_{i,1},w_1) + \sqrt{d}\eps\delta_{\max}\leq \delta_i + 2\sqrt{d}\eps\delta_{\max}$.  The same analysis applies to $d(v_{i,m},y)$ for every vertex $y$ of $\mathcal{A}[l]$.  This establishes constraint~3(a).   

Consider constraint~3(b).  The set $J_a = \bigr\{j : \pi_i(a) < j \leq \pi_i(a+1) \, \wedge \, \mathcal{A}[j] \not= \mathrm{null} \bigr\}$ contains the indices of the vertices $w_j$'s with non-null $\mathcal{A}[j]$'s such that $g_i^{-1}(w_j) \cap \tau_{i,a} \not= \emptyset$.  For all $j \in J_a$, pick any point in $g_i^{-1}(w_j) \cap \tau_{i,a}$ to be $p_j$.  Then, constraint~3(b)(i) follows from the fact that $g_i$ is a matching. Constraint~3(b)(ii) follows from the inequalities $d(w_j,p_j) \leq d_{g_i}(\tau_i,\sigma) \leq \delta_i + \sqrt{d}\eps\delta_{\min}$ and the $\sqrt{d}\eps\delta_{\max}$ bound on the distance from $w_j$ to any vertex of the cell $\mathcal{A}[j]$ that contains $w_j$.

Consider constraint~3(c).   Take any index $j \in [l-1]$ such that $\mathcal{A}[j] = \mathrm{null}$.  Recall that $N_j = \bigl\{ i : \exists \, a_i \in [m-1] \,\, \text{s.t.} \,\, \pi_i(a_i) < j \leq \pi_i(a_i+1)\bigr\}$.  We need to show that constraints~3(c)(i)--(v) are satisfied for all $i \in N_j$ whenever they are applicable.  

Since $\pi_i(a_i) < j \leq \pi_i(a_i+1)$, $g_i^{-1}$ matches $w_j$ to point(s) in $\tau_{i,a_i}$ within a distance of $\delta_i + \sqrt{d}\eps\delta_{\min}$.  It shows  that $d(w_j, \tau_{i,a_i}) \leq \delta_i + \sqrt{d}\eps\delta_{\min}$.  The reason why we set $\mathcal{A}[j]$ to be null is that $w_j$ lies outside all cells in $\grid_2$.  It implies that both $d(v_{i,a_i},w_j)$ and $d(v_{i,a_i+1},w_j)$ are greater than $9\sqrt{d}\delta_{\max}$.  Under these circumstances, the condition $w_j \!\downarrow\!\tau_{i,a_i} \in \tau_{i,a_i}$ must hold as it is the only way for $w_j$ to be matched to point(s) in $\tau_{i,a_i}$ within a distance of $\delta_i + \sqrt{d}\eps\delta_{\min}$.    So constraint~3(c)(i) is satisfied.

Before we verify constraints~3(c)(ii) and (iii), we claim that $d(w_{j-1},w_j) \geq 7\sqrt{d}\delta_{\max}$.   Take any $r \in [n]$ and any $a \in [m]$ such that $g_r(v_{r,a}) \cap w_{j-1}w_j \not= \emptyset$.  If $d(w_{j-1},w_j) < 7\sqrt{d}\delta_{\max}$, then $d(v_{r,a},w_j) \leq 8\sqrt{d}\delta_{\max} + \sqrt{d}\eps\delta_{\min}$, which implies that $w_j$ lies in a grid cell in $G(v_{r,a}+B_{9\sqrt{d}\delta_{\max}},\eps\delta_{\max}) \subseteq \grid_2$.  But this is a contradiction to the fact that $\mathcal{A}[j] = \mathrm{null}$.  This proves the claim.

Suppose that $\pi_i(a_i) < j -1 \leq \pi_i(a_i+1)$ and $\mathcal{A}[j-1] = \mathrm{null}$.   So $g_i^{-1}$ matches both $w_{j-1}$ and $w_j$ to point(s) in $\tau_{i,a_i}$.  As argued previously for $w_j$, we must have $w_{j-1} \!\downarrow\!\tau_{i,a_i} \in \tau_{i,a_i}$.   
%The perpendiculars from $w_{j-1}$ and $w_j$ to $\tau_{i,a_i}$ do not cross.  Moreover, since $d(w_{j-1},w_j) \geq 7\sqrt{d}\delta_i$ and both $\max\{d(w_{j-1},x): x \in g_i^{-1}(w_{j-1})\}$ and $\max\{d(w_j,x) : x \in g_i^{-1}(w_j)\}$ are at most $\delta_i + \sqrt{d}\eps\delta_{\min}$, no segment from $w_{j-1}$ to $g_i^{-1}(w_{j-1})$ cross any segment from $w_j$ to $g_i^{-1}(w_j)$.  It follows that the ordering of $g_i^{-1}(w_{j-1})$ and $g_i^{-1}(w_j)$ along $\tau_{i}$ is the same as that of $w_{j-1} \!\downarrow\!\tau_{i,a_i}$ and $w_j \!\downarrow\!\tau_{i,a_i}$.  Therefore, $w_{j-1} \!\downarrow\!\tau_{i,a_i} \leq_{\tau_i} w_j \!\downarrow\!\tau_{i,a_i}$.  
The angle $\angle (v_{i,a_i+1}-v_{i,a_i},w_j-w_{j-1})$ is either at most $\pi/2$ or greater than $\pi/2$.  In the first case, $w_{j-1} \!\downarrow\!\tau_{i,a_i} \leq_{\tau_i} w_j \!\downarrow\!\tau_{i,a_i}$, so constraint~3(c)(ii) is satisfied.  In the latter case, $d(w_j,g_i^{-1}(w_j) \cap \tau_{i,a_i}) \geq d(w_{j-1},w_j) - \max\{d(w_{j-1},x) : x \in g_i^{-1}(w_{j-1})\} \geq 7\sqrt{d}\delta_{\max}-\delta_i-\sqrt{d}\eps\delta_{\min} > \delta_i  + \sqrt{d}\eps\delta_{\min}$, a contradiction.

Suppose that $\pi_i(a_i) < j -1 \leq \pi_i(a_i+1)$ and $\mathcal{A}[j-1] \not= \mathrm{null}$.  
%Note that $i \in \hat{N}_j \subseteq N_j$.  
Since $d(w_{j-1},w_j) \geq 7\sqrt{d}\delta_{\max}$, the distance from $w_j \!\downarrow\!\tau_{i,a_i}$ to the boundary of $\mathcal{A}[j-1] \oplus B_{\delta_i+3\sqrt{d}\eps\delta_i}$ is at least $5\sqrt{d}\delta_{\max} - 5\sqrt{d}\eps\delta_i \geq 0$.  Therefore, $w_j \!\downarrow\!\tau_{i,a_i}$ does not lie inside $\mathcal{A}[j-1] \oplus B_{\delta_i+3\sqrt{d}\eps\delta_i}$.  Since $\max\{d(w_j,x): x \in g_i^{-1}(w_j)\}$ is at most $\delta_i + \sqrt{d}\eps\delta_{\min}$, $g_i^{-1}(w_j)$ does not overlap with $\mathcal{A}[j-1] \oplus B_{\delta_i + 3\sqrt{d}\eps\delta_i}$.  Similarly, we have $g_i^{-1}(w_{j-1}) \subset  \mathcal{A}[j-1] \oplus B_{\delta_i+3\sqrt{d}\eps\delta_i}$.  
%As argued in the previous paragraph, no segment from $w_{j-1}$ to $g_i^{-1}(w_{j-1})$ cross any segment from $w_j$ to $g_i^{-1}(w_j)$.  
It follows that the ordering of $g_i^{-1}(w_{j-1})$ and $g_i^{-1}(w_j)$ along $\tau_{i}$ is the same as that of $g_i^{-1}(w_{j-1})$ and $w_j \!\downarrow\!\tau_{i,a_i}$.  Hence, $\tau_{i,a_i} \cap (\mathcal{A}[j-1] \oplus B_{\delta_i+3\sqrt{d}\eps\delta_i}) \leq_{\tau_i} w_j \!\downarrow\!\tau_{i,a_i}$.  So constraint~3(c)(iii) is satisfied.

We can similarly show that constraint~3(c)(iv) and (v) are satisfied if applicable.
\end{proof}

\section{A geometric construct}
\label{app:geom}

Let $R$ and $S$ be two non-empty, bounded convex polytopes.   Recall that $F(R,S) = \{p \in \real^d : \exists \, q \in S \,\, \text{s.t.} \,\, pq \cap R \not= \emptyset\}$.  Let $\beta_1, \beta_2, ..., \beta_{|R|}$ be the vertices of $R$.  Let $\phi_1, \phi_2, ..., \phi_{|S|}$ be the vertices of $S$.   Let $\psi_{ij}$ be the vector $\beta_i-\phi_j$ for all $i \in [|R|]$ and $j \in [|S|]$.  We prove that $F(R,S) = R  \oplus \psi_{RS}$, where $\psi_{RS}$ denotes the set of conical combinations of $\bigl\{\psi_{ij} : i \in [|R|], \, j \in [|S|] \bigr\}$, i.e., $\sum_{i,j} x_{ij}\psi_{ij}$ for some non-negative $x_{ij}$'s.

\begin{lemma}\label{Lemma: Decomposition}
$F(R,S) = R \oplus \psi_{RS}$.
\end{lemma}
\begin{proof}
	%Since $R$ and $S$ are non-empty, we can always find a line that intersects both $R$ and $S$.  This shows that $F(R,S)$ is non-empty and unbounded. 
	%
	First, we prove that any point $p \in F(R,S)$ can be written as $r + \psi$ for some point $r \in R$ and some conical combination $\psi \in \psi_{RS}$.   If $p \in R$, this property is trivially true as we can take $\psi = 0$.  Suppose that $p \not\in R$.  According to the definition of $F(R,S)$, we can find a point $s \in S$ such that the line segment $ps$ intersects $R$. Let $r$ be a point in $ps \cap R$.  We have $p = r +\lambda(r-s)$ for some non-negative $\lambda$.  Observe that $r \not= s$; otherwise, $p = r \in R$, contradicting the assumption that $p \not\in R$.
	
	We show that $r-s \in \psi_{RS}$. Given that $r \in R$ and $s \in S$, we have $r = \sum_{i=1}^{|R|} a_i\beta_i$ and $s = \sum_{j=1}^{|S|}b_j\phi_j$ for some $a_i, b_j \geq 0$ such that $\sum_i a_i = \sum_j b_j = 1$. To determine whether $r-s \in \psi_{RS}$, we need to verify whether we can find non-negative coefficients $x_{ij}$'s such that $r - s = \sum_{i,j}x_{ij}\psi_{ij}$.  We have
	\[
		r - s = \sum_{1\le i\le |R|}a_i\beta_i-\sum_{1\le j\le |S|}b_j\phi_j.
	\]
	By expanding $\sum_{i,j}x_{ij}\psi_{ij}$ we get
	\begin{align*}
		&\sum_{i,j}x_{ij}\psi_{ij}\\
		=&\sum_{i,j}x_{i,j}(\beta_i-\phi_j)\\
		=&\sum_{1\le i\le |R|} \Bigl(\sum_{1\le j\le |S|}x_{ij}\Bigr)\beta_i - \sum_{1\le j\le |S|} \Bigl(\sum_{1\le i\le |R|}x_{i,j}\Bigr)\phi_j
	\end{align*}
	By comparing terms, we obtain the following linear system.
	\begin{center}
	\begin{tabular}{lll}
		$x_{i,1}+x_{i,2}+...+x_{i,|S|} = a_i$, & & $\forall\, 1 \leq i \leq |R|$, \\
		$x_{1,j}+x_{2,j}+...+x_{|R|,j} = b_j$, & & $\forall\, 1 \leq j\leq |S|$,  \\
		$x_{ij} \ge 0$, & & $\forall\, 1 \leq i \leq |R|, \, \forall\, 1 \leq j \leq |S|$.
	\end{tabular} 
	\end{center}
	According to Farkas' Lemma~\cite{schrijver1998theory}, there exists a vector $x\ge0$ that satisfies a linear system $Ax=e$ if and only if $ye\ge0$ for each row vector $y$ such that $yA\ge 0$. In our case, $e$ is the column vector $(a_1,\ldots,a_{|R|}, b_1,\ldots, b_{|S|})^t$, $x$ is the column vector $(x_{1,1}, \ldots, x_{1,|S|}, \ldots)^t$, and $A$ is a $(|R|+|S|)\times (|R||S|)$ matrix.  For $0 \leq k \leq |R||S|-1$, in the $(k+1)$-th column of $A$, the $(\lfloor k/|R| \rfloor + 1)$-th and the $(|R|+1+ k~\mathrm{mod}~|S|)$-th entries are equal to 1, and all other entries are zero.  Thus, given a row vector $y = (y_1, y_2, ..., y_{|R|+|S|})$, it satisfies $yA\ge 0$ if and only if 
	\[
	y_i + y_{|R|+j} \ge 0, \quad \forall\, 1 \leq i \leq |R|, \, \forall\, 1 \leq j \leq |S|.
	\]
	
	Next, we prove that $ye \ge 0$ if $yA\ge 0$. Expanding $ye$ we get:
	\[
	ye = a_1y_1+a_2y_2+...+b_1y_{|R|+1}+...+b_{|S|}y_{|R|+|S|}.
	\]
	
	Without loss of generality, assume that $a_{i'}$ is the smallest element among $a_i$'s and $b_j$'s. Let $k = \min\{i', |S|\}$. Then,
	\begin{eqnarray*}
	ye & = & a_{i'}(y_{i'}+y_{|R|+k}) + \Lambda, \quad \text{where} \\
	\Lambda & = & a_1y_1+...+a_{i'-1}y_{i'-1}+a_{i'+1}y_{i'+1}+...+a_{|R|}y_{|R|} + \\
		             & & b_1y_{|R|+1}+...+(b_k -a_{i'})y_{|R|+k}+...+b_{|S|}y_{|R|+|S|}.
	\end{eqnarray*}
	
	Note that $b_k - a_{i'} \geq 0$.  Both $(\sum_i a_i) - a_{i'}$ and $(\sum_j b_j) - a_{i'}$ are equal to $1 - a_{i'}$ as $\sum_i a_i = \sum_j b_j = 1$.  We repeat the above to pair another $y_i$ and $y_{|R|+j}$.  Since the sums of the coefficients decrease by the same amount for the $a_i$'s and $b_j$'s, there must be nothing left by the repeated pairing in the end.  That is, $\Lambda = 0$ in the end.  As a result, $ye$ can be written as a sum of terms like $y_i+y_{|R|+j}$ with non-negative coefficients.  Therefore, $ye\ge 0$, and $Ax = e$ has a feasible solution. This shows that $F(R,S) \subseteq R \oplus \psi_{RS}$.
	
	Next, we prove that $R \oplus \psi_{RS} \subseteq F(R,S)$.  Consider a point $r + \lambda\sum_{i,j}x_{ij}\psi_{ij}$, where $r \in R$, $\lambda$ is a non-negative real number, and $x_{ij}$'s are non-negative coefficients.  If $\sum_{i,j} x_{ij} = 0$, then $x_{ij} = 0$ for all $i$ and $j$, and the point is just $r$.  One can always draw a line through $r$ that intersects $S$, which shows that $r \in F(R,S)$.  Suppose that $\sum_{i,j} x_{ij} > 0$.  Without loss of generality, we assume that $\sum_{i,j}x_{ij}=1$ because, if necessary, $\lambda\sum_{i,j} x_{ij}\psi_{ij}$ can be written as $\lambda\bigl(\sum_{i,j} x_{ij}\bigr) \cdot \sum_{i,j} \bigl(\frac{x_{ij}}{\sum_{i,j} x_{ij}}\bigr) \psi_{ij}$.
	
	By expanding $\sum_{i,j}x_{ij}\psi_{ij}$ we get
	\[
		\sum_{i,j}x_{ij}\psi_{ij} = \sum_{1\le i\le |R|} \Bigl(\sum_{1\le j\le |S|}x_{ij}\Bigr)\beta_i-\sum_{1\le j\le |S|} \Bigl(\sum_{1\le i\le |R|}x_{i,j}\Bigr)\phi_j.
	\]
	Since $r \in R$, we can write it as $\sum_{i=1}^{|R|}a_i\beta_i$, where $a_i \geq 0$ and $\sum_{i=1}^{|R|} a_i = 1$.  We define
	\[
		a_i' = \frac{1}{1+\lambda} \left(a_i+\lambda\sum_{1\le j\le|S|}x_{ij}\right), \quad\quad b'_j = \sum_{i=1}^{|R|} x_{ij}, \quad\quad
	r' = \sum_{i=1}^{|R|}a_i'\beta_i, \quad\quad s = \sum_{j = 1}^{|S|} b'_j \phi_j.
	\]
	Clearly, $a'_i$ and $b'_j$ are non-negative for all $i$ and $j$.  Since $\sum_i a_i = \sum_{i,j} x_{ij} = 1$, we have $\sum_i a' = \frac{1}{1+\lambda}\sum_i a_i + \frac{\lambda}{1+\lambda}\sum_{i,j} x_{ij} = 1$.  Also, $\sum_{j} b'_j = \sum_{i,j} x_{ij} = 1$.  Therefore, $r'$ and $s$ are points in $R$ and $S$, respectively.  The point $r +\lambda\sum_{i,j}x_{ij}\psi_{ij}$ is equal to $r' +\lambda(r' - s)$. Hence, $r +\lambda\sum_{i,j}x_{ij}\psi_{ij} \in F(R,S)$.
\end{proof}

Notice that $F(R,S)$ is unbounded.  In fact, if $R \cap S \not= \emptyset$, then $F(R,S) = \real^d$.   The Minkowski sum of $R$ and $\psi_{RS}$ has complexity proportional to the product of their complexities, and so is the construction time of the Minkowski sum.  There are at most $|R||S|$ directions in $\{\psi_{ij}\}$.  Computing $\psi_{RS}$ boils down to a convex hull computation in $\real^{d-1}$ which has $O\bigl((|R||S|)^{\lfloor (d-1)/2 \rfloor}\bigr)$ complexity and can be constructed in $O\bigl((|R||S|)^{\lfloor (d-1)/2 \rfloor} + |R||S|\log (|R||S|)\bigr)$ time.

%Still, if $R \cap S = \emptyset$, we can construct $F(R,S)$ as follows.  Let $\Pi$ be any hyperplane such that it does not separate $R$ and $S$, it is far from $R$ and $S$, and it is orthogonal to the line through the means of the vertices of $R$ and $S$.\footnote{The mean of the vertices of $R$ is the point whose coordinates are the averages of those of the vertices of $R$.   The mean of the vertices of $S$ is defined similarly.}  Any ray that shoots from a vertex of $R$ in the direction of any $\psi_{ij}$ intersects $\Pi$.   Let $Q$ be the set of all intersection points between $\Pi$ and rays that shoot from the vertices of $R$ in directions $\{\psi_{ij}\}$.  So $|Q| = |R||S|$.  Compute the convex hull of $Q \cup R$ in $O(|R|^{d/2}|S|^{d/2})$ time.  This convex hull is a good representation of $F(R,S)$ for our purposes as long as $\Pi$ is far enough that it does not interfere with any other computation.

\begin{lemma}
	The complexity of $F(R,S)$ is $O(|R|^{1 + \lfloor (d-1)/2\rfloor}|S|^{\lfloor (d-1)/2\rfloor})$.  The construction time of $F(R,S)$ is $O\bigl(|R|^{1 + \lfloor (d-1)/2\rfloor}|S|^{\lfloor (d-1)/2\rfloor}+ |R||S|\log (|R||S|)\bigr)$.
\end{lemma}

\section{Proof of Lemmas~\ref{Lemma: Effect_Rel_set}~and~\ref{Lemma: ver_sig_match}}
\label{app:ver_sig_match}

We restate Lemma~\ref{Lemma: Effect_Rel_set} and give its proof.

\vspace{10pt}

\noindent\textbf{Lemma~\ref{Lemma: Effect_Rel_set}}~~\emph{For all $j \in [l-1]$, there exist points $p,q \in u_ju_{j+1}$ such that $p \leq_{\sigma} q$ and both $d(p,c_{j,1})$ and $d(q,c_{j,2})$ are at most $\sqrt{d}\eps\delta_{\max}$.}

\begin{proof}
	We enforce that $u_{j+1} \in \gamma_{j+1}$ and $u_j \in F(c_{j,2},u_{j+1}) \cap \gamma_j$ in the backward extraction.  It follows that $u_j \in \gamma_j$ and $u_ju_{j+1} \cap c_{j,2} \not= \emptyset$.  Take any point $q \in u_ju_{j+1} \cap c_{j,2}$.  The forward phase ensures that $\gamma_j \subseteq F(c_{j,1},c_{j,2})$, so there is a point $y \in c_{j,2}$ such that $yu_j \cap c_{j,1} \not= \emptyset$.  As $q$ and $y$ belong to $c_{j,2}$, $d(q,y) \leq \sqrt{d}\eps\delta_{\max}$.   Take any point $x \in yu_j \cap c_{j,1}$.  By a linear interpolation between $yu_j$ and $qu_j$, $x$ is mapped to a point $p \in qu_j$ such that $d(p,x) \leq d(q,y) \leq \sqrt{d}\eps\delta_{\max}$.  Clearly, $p \leq_{\sigma} q$.
\end{proof} 

We restate Lemma~\ref{Lemma: ver_sig_match} and give its proof.

\vspace{10pt}

\noindent\textbf{Lemma~\ref{Lemma: ver_sig_match}.}~~\emph{Take any $i \in [n]$ and any $a \in [m-1]$.  Suppose that $[k_1,k_2] = \{ j : \pi_i(a) < j \leq \pi_i(a+1)\}$ is non-empty.  There exist points $p,q \in \tau_{i,a}$ such that $d_F(pq,\sigma[u_{k_1},u_{k_2}]) \leq \delta_i + 4\sqrt{d}\eps\delta_{\max}$.}

\vspace{2pt}

\begin{proof}
	 To prove the lemma, we specify a matching $g$ from $\sigma[u_{k_1},u_{k_2}]$ to a segment in $\tau_{i,a}$.  The idea is to define $g$ on the vertices $u_k$ for $k \in [k_1,k_2]$ first and then use linear interpolation to extend $g$ to other points in $\sigma[u_{k_1},u_{k_2}]$.  
	
	Let $J_a = \{j : \pi_i(a) < j \leq \pi_i(a+1) \, \wedge \, \mathcal{A}[j] \not= \mathrm{null}\}$.  By constraint~3(b), there exist points $\{p_j \in \tau_{i,a} : j \in J_a\}$ such that these $p_j$'s appear in order along $\tau_{i,a}$, and every $p_j$ is within a distance of $\delta_i + 2\sqrt{d}\eps\delta_{\max}$ from every vertex of $\mathcal{A}[j]$.  
	%One would be tempted to define $g(u_k) = p_k$ if $\mathcal{A}[k] \not= \mathrm{null}$ and $g(u_k) = u_k \!\downarrow\! \tau_{i,a}$ otherwise.  The ordering of the $g(u_k)$'s would be correct by constraints~3(b)(i) and~3(c)(ii)--(iv), if constraints~3(ii)~and~(iv) hold.  However, the forward construction only checks relaxed versions of constraints~3(c)(ii)~and~(iv), so we need to do more work to get the right matching $g$.
	We first define $g(u_k)$ for $k \in [k_1,k_2]$ inductively.  Consider the base case of $u_{k_1}$.  If $\mathcal{A}[k_1] \not= \mathrm{null}$, let $g(u_{k_1}) = p_1$; otherwise, let $g(u_{k_1}) = u_{k_1} \!\downarrow\! \tau_{i,a}$ which lies on $\tau_{i,a}$ by constraint~3(c)(i).  In general, take any $k \in [k_1+1,k_2]$.  If $\mathcal{A}[k] \not= \mathrm{null}$, define $g(u_k)$ to be the maximum of $g(u_{k-1})$ and $p_k$ according to $\leq_{\tau_i}$.  If $\mathcal{A}[k] = \mathrm{null}$, define $g(u_k)$ to be the maximum of $g(u_{k-1})$ and $u_k \!\downarrow\! \tau_{i,a}$ according to $\leq_{\tau_i}$.  
	
	The above definition of $g(u_k)$ automatically guarantees that $g(u_{k-1}) \leq_{\tau_i} g(u_k)$ for $k \in [k_1+1,k_2]$.  We need to bound the distance between $u_k$ and $g(u_k)$.   For every $k \in [k_1,k_2]$, if $\mathcal{A}[k] \not= \mathrm{null}$, let $y_k = p_k$; otherwise, let $y_k = u_k \!\downarrow\! \tau_{i,a}$.   
	
	We first prove a claim that for all $k \in [k_1+1,k_2]$, $y_{k-1}\leq_{\tau_i} y_k $ or $d(y_{k-1},y_k) \leq \sqrt{d}\eps\delta_{\max}/l$.  Suppose that $\mathcal{A}[k-1] \not= \mathrm{null}$.   So $y_{k-1} = p_{k-1}$.  If $\mathcal{A}[k] \not= \mathrm{null}$, then $y_k = p_k$ and constraint~3(b)(ii) ensures that $y_{k-1} \leq_{\tau_i} y_k$.  If $\mathcal{A}[k] = \mathrm{null}$, then $y_k = u_k \!\downarrow\! \tau_{i,a}$.  By constraint~3(b)(ii), $y_{k-1} = p_{k-1}$ is within a distance of $\delta_i + 2\sqrt{d}\eps\delta_{\max}$ from any vertex $x$ of $\mathcal{A}[k-1]$, which implies that $y_{k-1} \in \mathcal{A}[k-1] \oplus B_{\delta_i + 3\sqrt{d}\eps\delta_i}$.  In this case, constraint~3(c)(iii) ensures that $y_{k-1} \leq_{\tau_i} y_k$.  Suppose that $\mathcal{A}[k-1] = \mathrm{null}$.  If $\mathcal{A}[k] \not= \mathrm{null}$, by constraint~3(b)(ii), $y_k = p_k$ is within a distance of $\delta_i + 2\sqrt{d}\eps\delta_{\max}$ from any vertex $x$ of $\mathcal{A}[k]$, which implies that $y_k \in \mathcal{A}[k] \oplus B_{\delta_i + 3\sqrt{d}\eps\delta_i}$.  In this case, constraint~3(c)(v) ensures that $y_{k-1} \leq_{\tau_i} y_k$.  The remaining case is that both $\mathcal{A}[k-1]$ and $\mathcal{A}[k]$ are null.  Constraints~3(c)(ii) and~3(c)(iv) are what we need, but we did not check them exactly.  When computing $\gamma_{k-1}$, we check that $\gamma_{k-1} \subseteq c_{k-1,2} \oplus (-\Pi_{i,a})$, which ensures that $y_{k-1} = u_{k-1} \!\downarrow\! \tau_{i,a}$ does not strictly follow $c_{k-1,2} \!\downarrow\! \tau_{i,a}$ in the order $\leq_{\tau_i}$.  Similarly, when computing $\gamma_k$, we check that $\gamma_k \subseteq c_{k-1,2} \oplus \Pi_{i,a}$, which ensures that $y_k = u_k \!\downarrow\! \tau_{i,a}$ does not strictly precede $c_{k-1,2} \!\downarrow\! \tau_{i,a}$ in the order $\leq_{\tau_i}$.  As a result, if $y_{k-1} \leq_{\tau_i} y_k$ does not hold, both $y_{k-1}$ and $y_k$ must belong to $c_{k-1,2} \!\downarrow\! \tau_{i,a}$.  The side length of $c_{k-1,2}$ is at most $\eps\delta_{\max}/l$, which implies that $d(y_{k-1},y_k) \leq \sqrt{d}\eps\delta_{\max}/l$.  This completes the proof of our claim.
	
	Next, we prove by induction a second claim that for all $k \in [k_1,k_2]$, $d(y_k,g(u_k)) \leq (k-k_1)\sqrt{d}\eps\delta_{\max}/l$.  The base case of $k = k_1$ is easy because $y_{k_1} = g(u_{k_1})$ by definition.  Assume that the claim is true for some $k \in [k_1,k_2-1]$.   By definition, $g(u_{k+1})$ is equal to either $y_{k+1}$ or $g(u_k)$.  If $g(u_{k+1}) = y_{k+1}$, we are done.  Suppose that $g(u_{k+1}) = g(u_k)$.  The point $y_{k+1}$ must strictly precede $g(u_k)$ with respect to $\leq_{\tau_i}$ in order that we do not set $g(u_{k+1})$ to be $y_{k+1}$.  By induction assumption, $d(y_k,g(u_k)) \leq (k-k_1)\sqrt{d}\eps\delta_{\max}/l$.  If $y_{k} \leq_{\tau_i} y_{k+1}$, then $y_k \leq_{\tau_i} y_{k+1} \leq_{\tau_i} g(u_k)$ and $d(y_{k+1},g(u_{k+1})) = d(y_{k+1},g(u_k)) \leq d(y_k,g(u_k)) \leq (k-k_1)\sqrt{d}\eps\delta_{\max}/l$.  On the other hand, if $y_{k+1} \leq_{\tau_i} y_k$, then $d(y_{k+1},g(u_{k+1})) = d(y_{k+1},g(u_k)) \leq d(y_{k+1},y_k) + d(y_k,g(u_k))$, which by our first claim and the induction assumption is at most $(k+1-k_1)\sqrt{d}\eps\delta_{\max}/l$.   This completes the proof of our second claim.
	
	We finish bounding the distance between $u_k$ and $g(u_k)$ as follows.  Suppose that $\mathcal{A}[k] \not= \mathrm{null}$.  Let $x$ be any vertex of $\mathcal{A}[k]$.  We have $d(u_k,y_k) = d(u_k,p_k) \leq d(u_k,x) + d(p_k,x) \leq \sqrt{d}\eps\delta_{\max} + d(p_k,x)$ which by constraint~3(b)(ii) is at most $\delta_i + 3\sqrt{d}\eps\delta_{\max}$.   If $\mathcal{A}[k] = \mathrm{null}$, then $d(u_k,y_k) = d(u_k, u_k \!\downarrow\! \tau_{i,a})$ which by constraint~3(c)(i) is at most $\delta_i + \sqrt{d}\eps\delta_{\max}$.    As a result, $d(u_k,g(u_k)) \leq d(u_k,y_k) + d(y_k,g(u_k)) \leq \delta_i + 3\sqrt{d}\eps\delta_{\max} + (k-k_1)\sqrt{d}\eps\delta_{\max}/l \leq \delta_i +  4\sqrt{d}\eps\delta_{\max}$.
\end{proof}

\section{Number of useful input curves}
\label{app:reduce}

\begin{lemma}\label{Lemma: subset_conf}
	Let $\Psi_l = (\mathcal{P},\mathcal{C},\mathcal{S},\mathcal{A})$ be a configuration that satisfies constraint~1.  If the forward construction returns non-empty $\gamma_1,\ldots,\gamma_l$ with respect to $\Psi_l$,  there exist a subset $T' \subseteq T$ of size at most $5l$ and a configuration $\Psi'_l = (\mathcal{P}',\mathcal{C},\mathcal{S},\mathcal{A})$ for the problem $Q(T',\Delta',\ell)$, where $\Delta' = \{\delta_i \in \Delta : \tau_i \in T'\}$, such that $\Psi'_l$ satisfies constraint~1, and the forward construction with respect to $\Psi'_l$ returns $\gamma_1,\ldots,\gamma_l$.
\end{lemma}
\begin{proof}
	Let $\Psi_l = (\mathcal{P},\mathcal{C},\mathcal{S},\mathcal{A})$ be a configuration of $\Q(T, \Delta,\ell)$ that satisfies constraint~1.  We construct the required subset $T'$ of $T$ incrementally.  Initialize $T'  = \{\tau_{\min}, \tau_{\max}\}$, where $\min = \mathrm{argmin}_{i \in [n]} \delta_i$ and $\max = \mathrm{argmax}_{i \in [n]}\delta_i$.
	
	%By the definition of $\mathcal{P}$, for every $j \in [l-1]$, there is at least one index $i$ such that $\pi_i^{-1}(j)$ is non-empty.  Pick one such index $i_j$ for every $j \in [l-1]$; if $\tau_{i_j} \not\in T'$, insert $\tau_{i_j}$ into $T'$.
	
	For every pair $(c_{j,1},c_{j,2}) \in \mathcal{C}$, let $i$ and $r$ be two indices in $[n]$ such that $c_{j,1} \in \bigcup_{a \in [m]} G(v_{i,a} + B_{\delta_i+\sqrt{d}\eps\delta_i},\eps\delta_i/l)$ and $c_{j,2} \in \bigcup_{a \in [m]} G(v_{r,a} + B_{\delta_{r}+\sqrt{d}\eps\delta_r},\eps\delta_{r}/l)$, insert $\tau_{i}$ and $\tau_{r}$ into $T'$ if they do not belong to $T'$.  
	
	For every $j \in [l]$ such that $\mathcal{A}[j] \not= \mathrm{null}$, pick an index $i \in [n]$ such that $\mathcal{A}[j]$ is a cell in $\bigcup_{a\in[m]} G(v_{i,a} + B_{9\sqrt{d}\delta_{\max}}, \eps\delta_{\max})$, insert $\tau_i$ into $T'$ if $\tau_i$ does not belong to $T'$.
	
	There are at most $3l$ curves in $T'$ so far.  We are not done with growing $T'$ yet for defining the configuration $\Psi'_l $; nevertheless, we can now define $\mathcal{C}$, $\mathcal{S}$, and $\mathcal{A}$ to be the second, third, and fourth components of $\Psi'_l$, respectively, because $T'$ already includes the necessary input curves for generating the segments in $\mathcal{S}$ and the cells in $\mathcal{C}$ and $\mathcal{A}$.
	
	We expand $T'$ further as follows. One important observation is that $\gamma_1,\ldots,\gamma_l$ are segments.  Since all components of $\Psi_l$ and $\Psi'_l$ are the same except for $\mathcal{P}$ and $\mathcal{P}'$, the definition of $\gamma_1$ with respect to $\Psi_l$ will also valid with respect to $\Psi'_l$.  Consider $\gamma_k$ for any $k > 1$.  Assume inductively that the definitions of $\gamma_j$ for all $j \in [1,k-1]$ with respect to $\Psi_l$ are also valid with respect to $\Psi'_l$.  If the definition of $\gamma_k$ with respect to $\Psi_l$ is not valid with respect to $\Psi'_l$, the forward construction with respect to $\Psi_l$ must produce an endpoint of $\gamma_k$ by clipping $s_k$ with a linear constraint that is induced by a curve $\tau_i \in T \setminus T'$.  Insert $\tau_i$ into $T'$ in this case.  In all, we insert at most $2l$ more curves into $T'$, making its size at most $5l$.
	
	Finally, we define $\mathcal{P}' = (\pi_i)_{\tau_i \in T'}$, completing the definition of $\Psi'_l$.
\end{proof}

\section{Proof of Lemma~\ref{lem:first-level}}
\label{app:first-level}
	
We restate Lemma~\ref{lem:first-level} and give its proof.

\vspace{10pt}
	
\noindent \textbf{Lemma~\ref{lem:first-level}}~~\emph{Take any subset $S  \subseteq T$ with at least $n/\beta$ curves for any $\beta \geq 1$.  Let $\Delta = \{\delta_1,\ldots,\delta_{|S|}\}$ be a set of error thresholds for $S$ such that $Q(S,\Delta,\ell)$ has a solution.   Relabel elements, if necessary, so that $\delta_1 \leq  \ldots \leq \delta_{|S|}$.  For $i \in \bigl[|S|\bigr]$, let $S_i = \{\tau_i, \tau_{i+1},\ldots \}$ and let $\Delta_i = \{\delta_i,\delta_{i+1},\ldots\}$.   Take any $\alpha, \eps \in (0,1)$ and any $r \in \bigl[\frac{\eps|S|}{5\ell}\bigr]$.  There exists $\mathcal{H}_r \subseteq 2^{S_r}$ such that $|\mathcal{H}_r| = 3l-1$ for some $l \in [\ell]$, every subset in $\mathcal{H}_r$ has $\frac{\eps|S|}{5\ell}$ curves, and for every subset $R \subseteq S_r$, if $\tau_r \in R$ and $R \cap H \not= \emptyset$ for all $H \in \mathcal{H}_r$, there exist a configuration $\Psi = (\mathcal{P},\mathcal{C},\mathcal{S},\mathcal{A}) \in C(\overline{\mathcal{H}}_r \cup R ,h,\alpha)$ and a curve $\sigma = (w_1,\ldots,w_h)$ for some $h \in [l]$, where $\overline{\mathcal{H}}_r = S_r \setminus \bigcup_{H \in \mathcal{H}_r} H$, that satisfy the following properties.}
	\begin{enumerate}[(i)]
		\item \emph{$\Psi$ and $\sigma$ satisfy constraints~1--3 with respect to $\overline{\mathcal{H}}_r \cup R$. }
		\item \emph{There exists a configuration in $C(R,h,\alpha)$ that shares the components $\mathcal{C}$, $\mathcal{S}$, and $\mathcal{A}$ with $\Psi$.}
	\end{enumerate}

%\vspace{2pt}
	
\begin{proof}
	Let $\sigma$ be a solution of $Q(S,\Delta,\ell)$, which also makes $\sigma$ a solution for $Q(S_r,\Delta_r,\ell)$.  For every $\tau_i \in S_r$, let $g_i$ be a Fr\'{e}chet matching from $\tau_i$ to $\sigma$.  As at the beginning of the proof of Lemma~\ref{Lemma: valid_cons_set}, we shorten $\sigma$ to $(u_1,u_2,\ldots,u_l)$ for some $l \in [\ell]$, if necessary, such that for all $j \in [l-1]$, $g_i(v_{i,a}) \cap u_ju_{j+1} \not= \emptyset$ for some $i \in [n]$ and some $a \in [m]$.  We first construct the subsets $H_1, \ldots, H_{3l-1}$ in $\mathcal{H}_r$.  Afterwards, we show how to modify $\sigma$ and construct the configuration $\Psi = (\mathcal{P},\mathcal{C},\mathcal{S},\mathcal{A}) \in C(\overline{\mathcal{H}}_r \cup R,h,\alpha)$ to satisfy the lemma.
	
	\vspace{6pt}
	
	\noindent \textbf{Define $\pmb{H_1,\ldots,H_{3l-1}}$.}~For every $j \in [l-1]$, let $I_j = \{ \tau_i \in S_r : \exists \, a \in [m] \,\, \text{s.t.} \,\, g_i(v_{i,a}) \cap u_ju_{j+1} \not= \emptyset \}$, which is non-empty as explained in the first paragraph.  
	
	If $|I_j| < \frac{\eps |S|}{5\ell}$, define $H_{2j-1}$ and $H_{2j}$ so that both include all curves in $I_j$ and another $\frac{\eps |S|}{5\ell} - |I_j|$ arbitrary curves from $S_r \setminus I_j$.  Suppose that $|I_j| \geq \frac{\eps |S|}{5\ell}$.  Define $H_{2j-1}$ to be the subset of $I_j$ that induce the $\frac{\eps |S|}{5\ell}$ minimum points in $\bigl\{\min(g_i(v_{i,a}) \cap u_ju_{j+1}) : \tau_i \in I_j, \, a \in [m], \, g_i(v_{i,a}) \cap u_ju_{j+1} \not= \emptyset\bigr\}$ with respect to $\leq_{\sigma}$.   Similarly, define $H_{2j}$ to be the subset of $I_j$ that induce the $\frac{\eps |S|}{5\ell}$ maximum points in $\bigl\{\max(g_i(v_{i,a}) \cap u_ju_{j+1}) : \tau_i \in I_j, \, a \in [m], \, g_i(v_{i,a}) \cap u_ju_{j+1} \not= \emptyset\bigr\}$ with respect to $\leq_{\sigma}$.   
	
	The subsets $H_{2l-1},\ldots, H_{3l-2}$ are constructed as follows.  Recall that the configurations include segments from a set $\mathcal{L}$ that are constructed using the curve $\tau_r \in S_r$ that has the minimum error threshold.   
	%By the definition of $S_r$, this curve is $\tau_r$.  
	For every $j \in [l]$, find the segment in $\mathcal{L}$ nearest to $u_j$, let $w_j$ be the point in this segment that is nearest to $u_j$, and let $V_j = \bigl\{\min_{a \in [m]} d(w_j,v_{i,a}) : \tau_i \in S_r \bigr\}$.  Define $H_{2l-2+j}$ to be the set of curves in $S_r$ that induce the $\frac{\eps|S|}{5\ell}$ minimum distances in $V_j$.
	
	The last subset $H_{3l-1}$ consists of the curves in $S_r$ with the $\frac{\eps|S|}{5\ell}$ largest error thresholds.  
	%The last subset $H_{3l}$ consists of the curves in $S$ with the $\frac{\eps n}{5\beta\ell}$ smallest error thresholds.
	
	\vspace{6pt}
	
	\noindent \textbf{Modify $\pmb{\sigma}$ and define $\pmb{\Psi}$.}~Let $R$ be a subset of $S_r$ that contains $\tau_r$ and intersects every subset in $\mathcal{H}_r$.   Since $\overline{\mathcal{H}}_r \cup R \subseteq S_r$, $\sigma$ is a solution for $\overline{\mathcal{H}}_r \cup R$ too.  However, as we restrict from $S_r$ to $\overline{\mathcal{H}}_r \cup R$, it may no longer be true that for all $j \in [l-1]$, there exist $\tau_i \in \overline{\mathcal{H}}_r \cup R$ and $a \in [m]$ such that $g_i(v_{i,a}) \cap u_ju_{j+1} \not= \emptyset$.  This is a hindrance to defining the desired configuration $\Psi$.   We encounter the same issue in proving Lemma~\ref{Lemma: valid_cons_set}, so we can apply the same shortcutting technique in the proof of Lemma~\ref{Lemma: valid_cons_set}.  We repeat the details below because we need to argue later that a configuration in $C(R,h,\alpha)$ for some $h\in [l]$ shares the second, third and fourth components with $\Psi$.
	
	Suppose that this requirement is not met for $u_ju_{j+1}$.  Then, for every $\tau_i \in \overline{\mathcal{H}}_r \cup R$, there exists $a_i \in [m-1]$ such that $u_ju_{j+1} \subseteq g_i(\mathrm{int}(\tau_{i,a_i}))$.  Let $p$ be the maximum of $\bigl\{ \max(g_i(v_{i,a_i})) : \tau_i \in \overline{\mathcal{H}}_r \cup R \bigr\}$ with respect to $\leq_{\sigma}$.  Let $q$ be the minimum of $\bigl\{ \min(g_i(v_{i,a_i+1})) : \tau_i \in \overline{\mathcal{H}}_r \cup R \bigr\}$ with respect to $\leq_{\sigma}$.  Our choice of $p$ and $q$ means that $p \in g_s(v_{s,a_s})$ and $q \in g_t(v_{t,a_{t}+1})$ for some possibly non-distinct $\tau_s, \tau_t \in \overline{\mathcal{H}}_r \cup R$, and $\mathrm{int}(\sigma[p,q]) \subseteq g_i(\mathrm{int}(\tau_{i,a_i}))$ for all $\tau_i \in \overline{\mathcal{H}}_r \cup R$.  We update $\sigma$ by substituting $\sigma[p,q]$ with the edge $pq$, possibly making $p$ and $q$ new vertices of $\sigma$.  The number of edges of $\sigma$ is not increased by the replacement; it  may actually be reduced.  For all $\tau_i \in \overline{\mathcal{H}}_r \cup R$, we update $g_i$ by a linear interpolation along $pq$; since $\mathrm{int}(\sigma[p,q]) \subseteq g_i(\mathrm{int}(\tau_{i,a_i}))$, the replacement of $\sigma[p,q]$ and the linear interpolation ensure that after the update, $g_i$ remains a matching and $d_{g_i}(\tau_i,\sigma) \leq \delta_i$.  Our choice of $p$ and $q$ means that the update does not affect the subset of vertices of any $\tau_i$ that are matched by $g_i$ to the edges of $\sigma$ other than $pq$.  The update also ensures that $pq$ will not trigger another shortcutting, and $pq$ will not be shortened by other shortcuttings.  If necessary, we repeat the above to convert $\sigma$ to $(u'_1,\ldots,u'_h)$ for some $h \in [l]$ such that for all $j \in [h-1]$, there exist $\tau_i \in \overline{\mathcal{H}}_r \cup R$ and $a \in [m]$ such that $g_i(v_{i,a}) \cap u'_ju'_{j+1} \not= \emptyset$.
	
	Next, we snap the vertices of $\sigma$.  
	%Let $\tau_r$ be the curve with the minimum error threshold in $R$.  Note that $\tau_r$ must come from $R \cap H_{3l}$ by definition.  Since $\overline{\mathcal{H}}_S$ excludes the curves in $H_{3l}$, $\tau_r$ has the minimum error threshold among the curves in $\overline{\mathcal{H}}_S \cup R$ too.  
	By assumption, the curve $\tau_r$ with the minimum error threshold in $S_r$ belongs to $R$.
	We have $d_{g_r}(\sigma,\tau_r) \leq \delta_r$ which implies that $\sigma$ lies inside $\tau_r \oplus B_{\delta_r}$.  In $C(\overline{\mathcal{H}}_r \cup R,h,\alpha)$, the line segments come from the set $\mathcal{L}$ obtained by discretizing $\tau_r \oplus B_{\delta_r}$.  Therefore, every vertex $u'_j$ is at distance no more than $\sqrt{d}\alpha\delta_r$ from its nearest line segment in $\mathcal{L}$.  We snap $u'_j$ to the nearest point $w_j$ on that line segment in $\mathcal{L}$.  This converts $\sigma$ to another curve $(w_1,\ldots,w_h)$.  We update $g_i$ using the linear interpolations between $u'_ju'_{j+1}$ and $w_jw_{j+1}$ for all $j \in [h-1]$.  This update ensures two properties.   First, for every $j \in [h-1]$, there exist $\tau_i \in \overline{\mathcal{H}}_r \cup R$ and $a \in [m]$ such that $g_i(v_{i,a}) \cap w_jw_{j+1} \not= \emptyset$.  Second, $d_{g_i}(\sigma,\tau_i) \leq \delta_i + \sqrt{d}\alpha\delta_r$ for every $\tau_i \in \overline{\mathcal{H}}_r \cup R$.
	
	This completes the description of the curve $\sigma = (w_1,\ldots,w_h)$.  Next, we show how to construct a configuration $\Psi = (\mathcal{P},\mathcal{C}, \mathcal{S},\mathcal{A}) \in C(\overline{\mathcal{H}}_r \cup R,h,\alpha)$ that satisfies constraints~1--3 together with $\sigma$.  It is exactly the same analysis as in the proof of Lemma~\ref{Lemma: valid_cons_set}.  We repeat the construction of $\Psi$ below because we need to argue that there exists a configuration in $C(R,h,\alpha)$ that shares $\mathcal{C}$, $\mathcal{S}$, and $\mathcal{A}$.
	
	We first define $\mathcal{P}$.  For every $\tau_i \in \overline{\mathcal{H}}_r \cup R$ and every $a \in [m]$, if $w_1 \in g_i(v_{i,a})$, we define $\pi_i(a) = 0$; otherwise, let $j$ be the smallest index in $[h-1]$ such that $g_i(v_{i,a}) \cap w_jw_{j+1}$ and we define $\pi_i(a) = j$.   This induces $\mathcal{P} = \bigl\{\pi_i : \tau_i \in \overline{\mathcal{H}}_r \cup R \bigr\}$ such that for all $\tau_i \in \overline{\mathcal{H}}_r \cup R$, $\pi_i(1) = 0$ and if $a \leq b$, then $\pi_i(a) \leq \pi_i(b)$.  
	%Moreover, for every $j \in [k-1]$, since there exist $\tau_i \in \overline{\mathcal{H}}_r \cup R$ and $a \in [m]$ such that $g_i(v_{i,a}) \cap w_jw_{j+1} \not= \emptyset$ and $w_j \not\in g_i(v_{i,a})$, it must be the case that $a \in \pi_i^{-1}(j)$.  As a result, for every $j \in [k-1]$, $\bigcup_{\tau_i \in \overline{\mathcal{H}}_r \cup R} \pi_i^{-1}(j) \not= \emptyset$ as required by the definition of $\mathcal{P}$.
	
	Next, we define $\mathcal{C}$ as follows.  Take any $j \in [h-1]$.  Let $x_j$ and $y_j$ be the minimum and maximum points in $\bigcup_{\tau_i \in \overline{\mathcal{H}}_r \cup R } \bigcup_{a \in [m]} g_i(v_{i,a}) \cap w_jw_{j+1}$ with respect to $\leq_{\sigma}$, respectively.  Note that $\bigcup_{\tau_i \in \overline{\mathcal{H}}_r \cup R } \bigcup_{a \in [m]} g_i(v_{i,a}) \cap w_jw_{j+1}$ is non-empty because there exist $\tau_i \in \overline{\mathcal{H}}_r \cup R $ and $a \in [m]$ such that $g_i(v_{i,a}) \cap w_jw_{j+1} \not= \emptyset$.  By definition, $x_j \leq_{\sigma} y_j$.  There exists $\tau_i \in \overline{\mathcal{H}}_r \cup R $ such that $x_j$ is within a distance of $\delta_i + \sqrt{d}\alpha\delta_r$ from a vertex of $\tau_i$.   We can make the same conclusion about $y_j$.  It follows that $x_j$ and $y_j$ belong to cells in $\grid_1$.  Choose $c_{j,1}$ and $c_{j,2}$ to be any cells in $\grid_1$ that contain $x_j$ and $y_j$, respectively.   This gives the $(h-1)$-tuple $\mathcal{C} = ((c_{j,1},c_{j,2}))_{j \in [h-1]}$.
	
	The components $\mathcal{S}$ and $\mathcal{A}$ are defined as follows.  By construction, we know that $w_j$ lies on some segment $s_j \in \mathcal{L}$.  We simply set $\mathcal{S} = (s_j)_{j \in [h]}$. For every $j \in [h]$, if $w_j$ lies in some grid cell in $\grid_2$ defined with respect to $\overline{\mathcal{H}}_r \cup R$, we set $\mathcal{A}[j]$ to be that cell; otherwise, we set $\mathcal{A}[j]$ to be null. 
	
	This completes the definition of $\Psi = (\mathcal{P},\mathcal{C},\mathcal{S},\mathcal{A})$.  We can verify exactly as in the proof of Lemma~\ref{Lemma: valid_cons_set} that $\Psi$ and $\sigma$ satisfy constraints~1--3 with respect to $\overline{\mathcal{H}}_r \cup R$.  The details are omitted here.  We proceed to verify that $C(R,h,\alpha)$ contains a configuration that shares $\mathcal{C}$, $\mathcal{S}$, and $\mathcal{A}$.
	
	The component $\mathcal{S}$ can be generated by $R$ because $R$ contains the curve $\tau_r$, and $\tau_r$ generates the superset $\mathcal{L}$ of $\mathcal{S}$.
	
	Before we verify that $\mathcal{C}$ and $\mathcal{A}$ can be generated by $R$, we first establish a property for $(u'_1,\ldots,u'_h)$, the result of converting $(u_1,\ldots,u_l)$ by shortcutting.   Take any $j \in [h-1]$.   Let $\tilde{g}_i$ refer to the matching from $\tau_i$ to $(u'_1,\ldots,u'_h)$ obtained immediately after the conversion.  Let $E_j = \bigl\{\min(\tilde{g}_i(v_{i,a}) \cap u'_ju'_{j+1}) : \tau_i \in \overline{\mathcal{H}}_r \cup R, \, a \in [m], \, \tilde{g}_i(v_{i,a}) \cap u'_ju'_{j+1} \not= \emptyset \bigr\}$, and let $E'_j = \bigl\{\max(\tilde{g}_i(v_{i,a}) \cap u'_ju'_{j+1}) : \tau_i \in \overline{\mathcal{H}}_r \cup R, \, a \in [m], \, \tilde{g}_i(v_{i,a}) \cap u'_ju'_{j+1} \not= \emptyset \bigr\}$.  We claim that the minimum point in $E_j$ lies in the image of $\tilde{g}_i(v_{i,a})$ for some $\tau_i \in R$ and some $a \in [m]$, and so does the maximum point in $E'_j$.  
	%Let $p$ be the minimum point in $E_j$.  Assume to the contrary that $p  \not= \min(\tilde{g}_i(v_{i,a}) \cap u'_ju'_{j+1})$ for any $\tau_i \in R$ and any $a \in [m]$.  
	By the shortcutting procedure, $u'_j \in u_{j_0}u_{j_0+1}$ for some $j_0 \in [l-1]$.   When we shortcut some subcurve $\sigma[p,q]$ to produce the vertex $u'_j$, either we reach $u'_j$ by searching from $q$ in the order $\leq_\sigma$, or we reach $u'_j$ by searching from $p$ in the reverse order of $\leq_\sigma$.  Without loss of generality, assume that we reach $u'_j$ by searching from $q$ in the order $\leq_\sigma$.  Recall that we identify the set $I_{j_0} = \{\tau_i \in S_r : \exists \, a \in [m] \; \text{s.t.} \; \hat{g}_i(v_{i,a}) \cap u_{j_0}u_{j_0+1} \not= \emptyset\}$ to produce $H_{2j_0-1}$, where $\hat{g}_i$ refers to the matching from $\tau_i$ to $(u_1,\ldots,u_l)$.  The subset $H_{2j_0-1}$ contains the curves that induce the $\min\bigl\{\frac{\eps|S|}{5\ell}, |I_{j_0}|\bigr\}$ minimum points in $\{\min(\hat{g}_i(v_{i,a}) \cap u_{j_0}u_{j_0+1}) : \tau_i \in I_{j_0}, \, a \in [m], \, \hat{g}_i(v_{i,a}) \cap u_{j_0}u_{j_0+1} \not= \emptyset\}$.  If $R \cap I_{j_0} \not= \emptyset$, then since all curves in $H_{2j_0-1}$ are excluded from $\overline{\mathcal{H}}_r$, some curve in $R \cap H_{2j_0-1}$ must induce the minimum point $z$ in $\{\min(\hat{g}_i(v_{i,a}) \cap u_{j_0}u_{j_0+1}) : \tau_i \in \overline{\mathcal{H}}_r \cup R, \, a \in [m], \, \hat{g}_i(v_{i,a}) \cap u_{j_0}u_{j_0+1} \not= \emptyset\}$.  The search for $u'_j$ cannot go past $z$, so $u'_j \leq_\sigma z$.  No subsequent shortcutting can cause the removal of $z$, so $z$ belongs to $u'_ju'_{j+1}$.  After updating $\hat{g}_i$ to $\tilde{g}_i$ for all $\tau_i \in \overline{\mathcal{H}}_r \cup R$ following all shortcuttings, $z$ would still exist as the minimum point in $E_j$.  Hence, the minimum point in $E_j$ lies in the image of $\tilde{g}_i(v_{i,a})$ for some $\tau_i \in R$ and some $a \in [m]$.  The other possibility is that $R \cap I_{j_0} = \emptyset$.  In this case, $|I_{j_0}| < \frac{\eps|S|}{5\ell}$ and all curves in $I_{j_0}$ are excluded from $\overline{\mathcal{H}}_r$.  Therefore, for every $\tau_i \in \overline{\mathcal{H}}_r \cup R$ and every $a \in [m]$, $\hat{g}_i(v_{i,a}) \cap u_{j_0}u_{j_0+1} = \emptyset$.  However, this is a contradiction to the fact that $u'_j \in u_{j_0}u_{j_0+1}$ because the shortcutting procedure ensures that $u'_j$ belongs to $\hat{g}_i(v_{i,a})$ for some $\tau_i \in \overline{\mathcal{H}}_r \cup R$ and some $a \in [m]$.
	
	%the property that we established for $(u_1,\ldots,u_l)$ at the beginning of the proof.
	%which implies that $\hat{g}_i(v_{i,a}) \cap u'_ju'_{j+1}= \emptyset$.  It follows that $\tilde{g}_i(v_{i,a}) \cap u'_ju'_{j+1} = \emptyset$ for all $\tau_i \in \overline{\mathcal{H}}_S$ after updating $\hat{g}_i$ to $\tilde{g}_i$ for all $\tau_i \in \overline{\mathcal{H}}_S \cup R$.  The definition of $p$ forces $p$ to be induced by some curve in $R$, contradicting the assumption that $p \not= \min(\tilde{g}_i(v_{i,a}) \cap u'_ju'_{j+1})$ for any $\tau_i \in R$ and any $a \in [m]$.  The analysis for $E'_j$ is similar.	This completes the proof of our claim.
	
    Next, we make a second claim.  Take any $j \in [h-1]$. We claim that the minimum point in $\bigl\{\min(g_i(v_{i,a}) \cap w_jw_{j+1}) : \tau_i \in \overline{\mathcal{H}}_r \cup R, \, a \in [m], \, g_i(v_{i,a}) \cap w_jw_{j+1} \not= \emptyset \bigr\}$ lies in the image of $g_i(v_{i,a})$ for some $\tau_i \in R$, and so does the maximum point in $\bigl\{\max(g_i(v_{i,a}) \cap w_jw_{j+1}) : \tau_i \in \overline{\mathcal{H}}_r \cup R, \, a \in [m], \, g_i(v_{i,a}) \cap w_jw_{j+1} \not= \emptyset \bigr\}$, where $g_i$ refers to the matching from $\tau_i$ to $\sigma$ obtained after making $\sigma = (w_1,\ldots,w_h)$.  This claim immediately follows from the claim in the previous paragraph and the fact that we snap $u'_j$ to $w_j$ and then use linear interpolations to obtain $g_i$ from $\tilde{g}_i$.
	
	Consider the component $\mathcal{C}$.  Take any $j \in [h-1]$.  In defining $c_{j,1}$, we identify a cell in $\grid_1$ defined with respect to $\overline{\mathcal{H}}_r \cup R$ that contains the minimum point in $\bigcup_{\tau_i \in \overline{\mathcal{H}}_r \cup R } \bigcup_{a \in [m]} g_i(v_{i,a}) \cap w_jw_{j+1}$.  By our second claim, we can pick $c_{j,1}$ to be a cell in $\bigcup_{a \in [m]} G(v_{i,a} + B_{\delta_i + \sqrt{d}\alpha\delta_i},\alpha\delta_i/h)$ for some curve $\tau_i \in R$.  A similar conclusion holds for the cell $c_{j,2}$.  We may have picked $c_{j,1}$ and $c_{j,2}$ to be some other cells in defining $\mathcal{C}$; if so, we change them so that they can be generated using $R$.  Hence, $\mathcal{C}$ can be generated using $R$.

	We show that $R$ can generate the component $\mathcal{A}$.   First, $\overline{\mathcal{H}}_r$ excludes all curves in $H_{3l-1}$ whereas $R$ intersects $H_{3l-1}$.  Therefore, $\max\{\delta_i : \tau_i \in \overline{\mathcal{H}}_r \cup R\} = \max\{\delta_i : \tau_i \in R\}$, which means that $\grid_2$ uses the same $\delta_{\diamond} = \max\{\delta_i : \tau_i \in \overline{\mathcal{H}}_r \cup R\} = \max\{\delta_i : \tau_i \in R\}$ in the discretization of the vertex neighborhoods regardless of whether $\grid_2$ is defined with respect to $\overline{\mathcal{H}}_r \cup R$ or $R$.

	For every $j \in [h]$, $u'_j$ is either $u_{j_0}$ for some $j_0 \in [l]$, or $u'_j$ is a newly created vertex.  Consider the case that $u'_j = u_{j_0}$.  In this case, the vertex $w_j$ in the final $\sigma$ is equal to the projection $w_{j_0}$ of $u_{j_0}$ that we check in defining $H_{2l-1}, \ldots, H_{3l-2}$.  Recall the set $V_{j_0}$ that we use in defining $H_{2l-2+j_0}$; $V_{j_0}$ contains the distances from $w_j = w_{j_0}$ to the nearest vertex of $\tau_i$ for all $\tau_i \in S_r$; $H_{2l-2+j_0}$ contains the curves that induce the $\frac{\eps|S|}{5\ell}$ minimum distances in $V_j$.  Therefore, if $w_j \in B_{\delta_i + 9\sqrt{d}\delta_{\diamond}}$ for some $\tau_i \in \overline{\mathcal{H}}_r \cup R$, then $w_j$ also lies in $B_{\delta_s + 9\sqrt{d}\delta_{\diamond}}$ for some $\tau_s \in R \cap H_{2l-2+j_0}$.  Hence, $w_j$ is contained in a cell in $\grid_2$ defined with respect to $\overline{\mathcal{H}}_r \cup R$ if and only if $w_j$ is contained in a cell in $\grid_2$ defined with respect to $R$ alone.  The cells induced by the curves in $R$ exist in $\grid_2$ defined with respect to $\overline{\mathcal{H}}_r \cup R$ too.  We may have set $\mathcal{A}[j]$ to be a different cell; if so, we change $\mathcal{A}[j]$ to be a cell induced by a curve in $R \cap H_{2l-2+j_0}$.
	
	The remaining case is that $u'_j$ is a new vertex created by the shortcutting.  Recall that $\hat{g}_i$ refers to the matching from $\tau_i$ to $(u_1,\ldots,u_l)$ before the conversion to $(u'_1,\ldots,u'_h)$.  In this case, there exists $j_0 \in [l-1]$ such that $u'_j$ is either the minimum point in $\bigl\{\min(\hat{g}_i(v_{i,a}) \cap u_{j_0}u_{j_0+1}) : \tau_i \in \overline{\mathcal{H}}_r \cup R, \, a \in [m], \, \hat{g}_i(v_{i,a}) \cap u_{j_0}u_{j_0+1} \not= \emptyset \bigr\}$, or the maximum point in $\bigl\{\max(\hat{g}_i(v_{i,a}) \cap u_{j_0}u_{j_0+1}) : \tau_i \in \overline{\mathcal{H}}_r \cup R, \, a \in [m], \, \hat{g}_i(v_{i,a}) \cap u_{j_0}u_{j_0+1} \not= \emptyset \bigr\}$.  By the proof of the first claim that we showed previously, we know that $u'_j \in \tilde{g}_s(v_{s,b})$ for some $\tau_s \in R$ and some $b \in [m]$, where $\tilde{g}_s$ refers to the matching from $\tau_i$ to $(u'_1,\ldots,u'_h)$ obtained immediately after the conversion to $(u'_1,\ldots,u'_h)$.  We preserved the distance bounds in converting $\sigma$ from $(u_1,\ldots,u_l)$ to $(u'_1,\ldots,u_h)$.  It implies that $d(u'_j,v_{s,b}) \leq \delta_{s}$.  It follows that $d(w_j,v_{s,b}) \leq d(w_j,u'_j) + d(u'_j,v_{s,b})\leq \sqrt{d}\alpha\delta_s + \delta_s < 9\sqrt{d}\delta_s$.  Therefore, $w_j$ is contained in a cell in $G(v_{s,b},B_{9\sqrt{d}\delta_{\diamond}},\alpha\delta_{\diamond}) \subseteq\grid_2$ regardless of whether $\grid_2$ is defined with respect to $\overline{\mathcal{H}}_r \cup R$ or $R$.  We may have set $\mathcal{A}[j]$ to be a different cell; if so, we change $\mathcal{A}[j]$ to be the cell in $G(v_{s,b},B_{9\sqrt{d}\delta_{\diamond}},\alpha\delta_{\diamond})$ that contains $w_j$.
\end{proof}

\section{Proof of Lemma~\ref{lem:second-level}}
\label{app:second-level}

We restate Lemma~\ref{lem:second-level} and give its proof.

\vspace{10pt}

\noindent \textbf{Lemma~\ref{lem:second-level}}~~\emph{Take any subset $S \subseteq T$ with at least $n/\beta$ curves for any $\beta \geq 1$.   Let $\Delta$ be a set of error thresholds for $S$ such that $Q(S,\Delta,\ell)$ has a solution.   Assume the notation in Lemma~\ref{lem:first-level}.  Take any $\alpha, \eps \in (0,1)$, any $r \in \bigl[\frac{\eps|S|}{5\ell}\bigr]$, and any subset $R \subseteq S_r$ such that $\tau_r \in R$ and $R \cap H \not= \emptyset$ for all $H \in \mathcal{H}_r$.  Let $\hat{S}_r = \overline{\mathcal{H}}_r \cup R$.  Let $\Psi = (\mathcal{P},\mathcal{C},\mathcal{S},\mathcal{A}) \in C(\hat{S}_r,h,\alpha)$ for some $h \in [\ell]$ be a configuration that satisfies Lemma~\ref{lem:first-level}.  There exist $\mathcal{H}_{\Psi} \subseteq 2^{\hat{S}_r}$ and a configuration $\Psi'' = (\mathcal{P}'',\mathcal{C},\mathcal{S},\mathcal{A}) \in C(\overline{\mathcal{H}}_{\Psi} \cup R,h,\alpha)$, where $\overline{\mathcal{H}}_{\Psi} = \hat{S}_r \setminus \bigcup_{H \in \mathcal{H}_{\Psi}} H$, such that $|\mathcal{H}_{\Psi}| = 2h$, every subset in $\mathcal{H}_{\Psi}$ contains $\frac{\eps|S|}{5\ell}$ curves, and for all subset $R' \subseteq \hat{S}_r$, if $R' \cap H \not= \emptyset$ for all $H \in \mathcal{H}_{\Psi}$, then there exists a configuration $\Psi' = (\mathcal{P}',\mathcal{C},\mathcal{S},\mathcal{A}) \in C(R \cup R',h,\alpha)$ that satisfies the following properties.}
	\begin{enumerate}[(i)]
		\item \emph{For all $j \in [h]$, $\gamma_j(\hat{S}_r,\Psi) \subseteq \gamma_j(R\cup R',\Psi') \subseteq \gamma_j(\overline{\mathcal{H}}_{\Psi} \cup R,\Psi'')$.}
		\item \emph{The backward extraction using $\{\gamma_j(R\cup R',\Psi') : j \in [h]\}$ produces a curve $\sigma$ such that $d_F(\sigma,\tau_i) \leq \delta_i + 4\sqrt{d}\alpha\cdot \max\{\delta_i : \tau_i \in R \cup R' \}$ for all $\tau_i \in \overline{\mathcal{H}}_{\Psi} \cup R$.}
	\end{enumerate}

\vspace{2pt}

\begin{proof}
We construct the subsets $H_{2j-1}$ and $H_{2j}$ in $\mathcal{H}_{\Psi}$  inductively from $j = 1$ to $h$.  The definitions of $H_{2j-1}$ and $H_{2j}$ depend on the status of $\mathcal{A}[j]$.  

If $\mathcal{A}[j] \not= \mathrm{null}$, both $H_{2j-1}$ and $H_{2j}$ are just arbitrary subsets of $\hat{S}_r$ of size $\frac{\eps |S|}{5\ell}$.  This case covers the base case of $j = 1$ because $\mathcal{A}[1] \not= \mathrm{null}$ by constraint~1 (by Lemma~\ref{lem:first-level}, $\Psi$ satisfies constraint~1).  

Suppose that $\mathcal{A}[j] = \mathrm{null}$.  In this case, $j \in [2,k-1]$.  In the forward construction, $\gamma_j(\hat{S}_r,\Psi)$ is obtained by applying a set of clippings to the segment $\tilde{s}_j = F(c_{j-1,2},\gamma_{j-1}) \cap F(c_{j,1},c_{j,2}) \cap s_j$.   The effect is that each curve $\tau_i \in \hat{S}_r$ defines a subsegment $p_{ij}q_{ij} \subseteq \tilde{s}_j$, and $\gamma_j(\hat{S}_r,\Psi)$ is equal to $\bigcap_{\tau_i \in \hat{S}_r} p_{ij}q_{ij}$.  Recall that $\tilde{s}_j$ is parallel to the edge $\tau_{r,a}$ for some $a \in [m]$.  We assume that $p_{ij}$ precedes $q_{ij}$ in the orientation that is consistent with $\leq_{\tau_{r,a}}$.  Let $H_{2j-1}$ be the subset of curves in $\hat{S}_r$ that induce the $\frac{\eps |S|}{5\ell}$ maximum points in $\{p_{ij} : \tau_i \in \hat{S}_r\}$ with respect to $\leq_{\tau_{r,a}}$.  Symmetrically, let $H_{2j}$ be the subset of curves in $\hat{S}_r$ that induce the $\frac{\eps |S|}{5\ell}$ minimum points in $\{q_{ij} : \tau_i \in \hat{S}_r\}$ with respect to $\leq_{\tau_{r,a}}$.  This completes the definition of $\mathcal{H}_{\Psi}$.

To construct $\Psi''$, we extract the mappings $\bigl\{\pi_i \in \mathcal{P} : \tau_i \in  \overline{\mathcal{H}}_{\Psi} \cup R\bigr\}$ to form $\mathcal{P}''$.  Then, $\Psi'' = (\mathcal{P}'',\mathcal{C},\mathcal{S},\mathcal{A})$.  By Lemma~\ref{lem:first-level}(ii), the components $\mathcal{C}$, $\mathcal{S}$, and $\mathcal{A}$ can be induced by the curves in $R$.    It follows that $\Psi''$ is indeed a configuration in $C(\overline{\mathcal{H}}_{\Psi} \cup R,h,\alpha)$. Since $\Psi$ satisfies Constraint~1 by Lemma~\ref{lem:first-level}, $\Psi''$ satisfies constraint~1 too because it inherits $\pi_i$'s from $\mathcal{P}$, and $\Psi$ and $\Psi''$ share the same component $\mathcal{C}$.  Therefore, the forward construction of $\gamma_j(\overline{\mathcal{H}}_{\Psi} \cup R,\Psi'')$ for $j \in [h]$ can proceed.  
%Since $\overline{\mathcal{H}}_{\Psi} \cup R \subseteq \hat{S}_r$ and $\Psi''$ shares $\mathcal{C}$, $\mathcal{S}$ and $\mathcal{A}$ with $\Psi$, for $j \in [h]$, the forward construction of $\gamma_j(\overline{\mathcal{H}}_{\Psi} \cup R,\Psi'')$ applies a subset of the clippings of $\tilde{s}_j$ when compared with the forward construction of $\gamma_j(\hat{S}_r,\Psi)$.  By Lemmas~\ref{lem:forward} and~\ref{lem:first-level}, $\gamma_j(\hat{S}_r,\Psi)$ is non-empty for all $j \in [h]$.  The forward construction of $\gamma_j(\overline{\mathcal{H}}_{\Psi} \cup R,\Psi'')$ will not be aborted, so $\gamma_j(\hat{S}_r,\Psi) \subseteq \gamma_j(\overline{\mathcal{H}}_{\Psi} \cup R,\Psi'')$ for all $j \in [h]$.

Next, take any subset $R' \subseteq \hat{S}_r$ such that $R' \cap H \not= \emptyset$ for all $H \in \mathcal{H}_{\Psi}$.   We extract the mappings $\{\pi_i \in \mathcal{P} : \tau_i \in R \cup R'\}$ to form $\mathcal{P}'$.  Then, $\Psi' = (\mathcal{P}',\mathcal{C},\mathcal{S},\mathcal{A})$.  By Lemma~\ref{lem:first-level}(ii), the components $\mathcal{C}$, $\mathcal{S}$, and $\mathcal{A}$ can be induced by the curves in $R$.   It follows that $\Psi'$ is indeed a configuration in $C(R \cup R',h,\alpha)$.  Since $\Psi$ satisfies constraint~1 by Lemma~\ref{lem:first-level}, $\Psi'$ satisfies constraint~1 too because it inherits $\pi_i$'s from $\mathcal{P}$, and $\Psi$ and $\Psi'$ share the same component $\mathcal{C}$.  Therefore, the forward construction of $\gamma_j(R \cup R',\Psi')$ for $j \in [h]$ can proceed.   

We prove by induction that $\gamma_j(\hat{S}_r,\Psi) \subseteq \gamma_j(R\cup R',\Psi') \subseteq \gamma_j(\overline{\mathcal{H}}_{\Psi} \cup R,\Psi'')$ for $j \in [h]$.  Afterwards, Lemma~\ref{lem:second-level}(ii) automatically follows as we discussed in the main text.

In the base case of $j = 1$, all three of $\gamma_1(\hat{S}_r,\Psi)$, $\gamma_1(R \cup R',\Psi')$ and $\gamma_1(\overline{\mathcal{H}}_{\Psi} \cup R,\Psi'')$ are computed as $F(c_{1,1},c_{1,2}) \cap s_1 \cap \mathcal{A}[1]$.

Consider any $j \in [2,h-1]$.  If $\mathcal{A}[j] \not= \mathrm{null}$, then $\gamma_j(\hat{S}_r,\Psi)$, $\gamma_j(R \cup R',\Psi')$ and $\gamma_j(\overline{\mathcal{H}}_{\Psi} \cup R,\Psi'')$ are computed as follows:
\begin{align*}
	\gamma_j(\hat{S}_r,\Psi) &= F\bigl(c_{j-1,2}, \gamma_{j-1}(\hat{S}_r,\Psi) \bigr)\cap F(c_{j,1}, c_{j,2})\cap s_j \cap \mathcal{A}[j], \\
	\gamma_j(R \cup R',\Psi') &= F\bigl(c_{j-1,2}, \gamma_{j-1}(R \cup R',\Psi') \bigr)\cap F(c_{j,1}, c_{j,2})\cap s_j \cap \mathcal{A}[j],\\
	\gamma_j(\overline{\mathcal{H}}_{\Psi} \cup R,\Psi'') &= F\bigl(c_{j-1,2}, \gamma_{j-1}(\overline{\mathcal{H}}_{\Psi} \cup R,\Psi'')\bigr)\cap F(c_{j,1}, c_{j,2})\cap s_j \cap \mathcal{A}[j].
\end{align*}
Since $\gamma_{j-1}(\hat{S}_r,\Psi) \subseteq \gamma_{j-1}(R \cup R',\Psi') \subseteq \gamma_{j-1}(\overline{\mathcal{H}}_{\Psi} \cup R,\Psi'')$ by induction assumption, we have 
\[
F\bigl(c_{j-1,2}, \gamma_{j-1}(\hat{S}_r,\Psi) \bigr) \subseteq F\bigl(c_{j-1,2}, \gamma_{j-1}(R \cup R',\Psi') \bigr) \subseteq F\bigl(c_{j-1,2}, \gamma_{j-1}(\overline{\mathcal{H}}_{\Psi} \cup R,\Psi'') \bigr).
\]
It follows that $\gamma_j(\hat{S}_r,\Psi) \subseteq \gamma_j(R \cup R',\Psi') \subseteq \gamma_j(\overline{\mathcal{H}}_{\Psi} \cup R,\Psi'')$.

Suppose that $\mathcal{A}[j] = \mathrm{null}$.  As in defining the subsets $H_1,\ldots,H_{2h}$, the construction of $\gamma_j(\hat{S}_r,\Psi)$ can be viewed as applying some clippings to the segment $\tilde{s}_j = F\bigl(c_{j-1,2}, \gamma_{j-1}(\hat{S}_r,\Psi) \bigr)\cap F(c_{j,1}, c_{j,2})\cap s_j$.  That is, each curve $\tau_i \in \hat{S}_r$ induces a subsegment $p_{ij}q_{ij}$ on $\tilde{s}_j$, and
\[
\gamma_j(\hat{S}_r,\Psi)  = \bigcap_{\tau_i \in \hat{S}_r} p_{ij}q_{ij}.
\]
Similarly, $R \cup R'$ induces a set of subsegments on $\tilde{s}_j$ so that $\gamma_j(R \cup R',\Psi')$ is the common intersection of these subsegments.   Moreover, since $R \cup R' \subseteq \hat{S}_r$, the set of subsegments induced by $R \cup R'$ is exactly $\{ p_{ij} q_{ij} : \tau_i \in R \cup R' \} \subseteq  \{ p_{ij} q_{ij} : \tau_i \in \hat{S}_r\}$.  Therefore,
\[
\gamma_j(R \cup R',\Psi') =  \bigcap_{\tau_i \in R \cup R'} p_{ij}q_{ij} \, \supseteq \, \gamma_j(\hat{S}_r,\Psi).
\]
In the same manner, we have
\[
\gamma_j(\overline{\mathcal{H}}_{\Psi} \cup R,\Psi'') = \bigcap_{\tau_i \in \overline{\mathcal{H}}_{\Psi} \cup R} p_{ij}q_{ij}.
\]
By the definition of $\mathcal{H}_{\Psi}$ and $\overline{\mathcal{H}}_{\Psi}$, the curves in $\hat{S}_r$ that induce the $\frac{\eps |S|}{5\ell}$ maximum points in $\{p_{ij} : \tau_i \in \hat{S}_r\}$ are excluded from $\overline{\mathcal{H}}_{\Psi}$.  Therefore, if any of these curves are present in $\overline{\mathcal{H}}_{\Psi} \cup R$, they must belong to $R$.  On the other hand, $R'$ intersects every subset in $\mathcal{H}_{\Psi}$ by assumption, which implies that $R'$ contains curve(s) in $\hat{S}_r$ that induce some of the $\frac{\eps |S|}{5\ell}$ maximum points in $\{p_{ij} : \tau_i \in \hat{S}_r\}$.  Altogether, we conclude that the maximum point in $\{p_{ij} : \tau_i \in R \cup R'\}$ is equal to or follows the maximum point in $\{p_{ij} : \tau_i \in \overline{\mathcal{H}}_{\Psi} \cup R\}$.  Similarly, the minimum point in $\{q_{ij} : \tau_i \in R \cup R'\}$ is equal to or precedes the minimum point in $\{q_{ij} : \tau_i \in \overline{\mathcal{H}}_{\Psi} \cup R\}$.  As a result,
\[
\gamma_j(R \cup R',\Psi') =  \bigcap_{\tau_i \in R \cup R'} p_{ij}q_{ij}  \, \subseteq  \bigcap_{\tau_i \in \overline{\mathcal{H}}_{\Psi} \cup R} p_{ij}q_{ij} = \gamma_j(\overline{\mathcal{H}}_{\Psi} \cup R,\Psi'').
\]

Finally, in the terminating case of $j = h$, $\mathcal{A}[h] \not= \mathrm{null}$, and so $\gamma_h(\hat{S}_r,\Psi)$, $\gamma_h(R \cup R',\Psi')$ and $\gamma_h(\overline{\mathcal{H}}_{\Psi} \cup R,\Psi'')$ are computed as follows:
\begin{align*}
	\gamma_h(\hat{S}_r,\Psi) &= F\bigl(c_{h-1,2}, \gamma_{h-1}(\hat{S}_r,\Psi) \bigr)\cap F(c_{h,1}, c_{h,2})\cap s_h \cap \mathcal{A}[h], \\
	\gamma_h(R \cup R',\Psi') &= F\bigl(c_{h-1,2}, \gamma_{h-1}(R \cup R',\Psi') \bigr)\cap F(c_{h,1}, c_{h,2})\cap s_h \cap \mathcal{A}[h],\\
	\gamma_h(\overline{\mathcal{H}}_{\Psi} \cup R,\Psi'') &= F\bigl(c_{h-1,2}, \gamma_{h-1}(\overline{\mathcal{H}}_{\Psi} \cup R,\Psi'')\bigr)\cap F(c_{h,1}, c_{h,2})\cap s_h \cap \mathcal{A}[h].
\end{align*}
We conclude as before that $\gamma_h(\hat{S}_r,\Psi) \subseteq \gamma_h(R \cup R',\Psi') \subseteq \gamma_h(\overline{\mathcal{H}}_{\Psi} \cup R,\Psi'')$.
\end{proof}

\section{Proof of Lemma~\ref{lem:finder}}
\label{app:finder}

We restate Lemma~\ref{lem:finder} and give its proof, which is adapted from the proof of a similar result in~\cite{buchin2021approximating}.

\vspace{10pt}

\noindent \text{Lemma~\ref{lem:finder}}~~\emph{For $\eps < 1/9$, Algorithm~\ref{ALG:1_APP} is a $(1+\eps)$-approximate candidate finder with success probability at least $1-\mu$.   The algorithm outputs a set $\Sigma$ of curves, each of $\ell$ vertices; for every subset $S \subseteq T$ of size $\frac{1}{\beta}|T|$ or more, it holds with probability at least $1-\mu$ that there exists a curve $\sigma \in \Sigma$ such that $\mathrm{cost}(S,\sigma) \leq (1+\eps)\mathrm{cost}(S,c^*)$, where $c^*$ is the optimal $(1,\ell)$-median of $S$.  The running time and output size of Algorithm~\ref{ALG:1_APP} are $\tilde{O}\bigl(m^{O(\ell^2)} \cdot \mu^{-O(\ell)} \cdot (d\beta\ell/\eps)^{O((d\ell/\eps)\log(1/\mu))}\bigr)$.}

\vspace{8pt}

\begin{proof}
To prove that Algorithm~\ref{ALG:1_APP} is a $(1+\eps)$-approximate candidate finder with success probability at least $1-\mu$, we need to show that the algorithm returns a set $\Sigma$ of curves, each of $\ell$ vertices, and for all subset $S \subseteq T$ of size $\frac{1}{\beta}|T|$ or more, it holds with probability at least $1-\mu$ that there exists a curve $\sigma \in \Sigma$ such that $\mathrm{cost}(S,\sigma) \leq (1+\eps)\mathrm{cost}(S,c^*)$, where $c^*$ is the optimal $(1,\ell)$-median of $S$.  

In line~\ref{alg:sample}, the algorithm samples a multiset $Y \subseteq T$ of $\bigl\lceil \frac{80\beta\ell}{\eps}\ln\frac{80\ell}{\mu}\bigr\rceil$ possibly non-distinct curves.  We treat the intersection $Y \cap S$ as a multiset too, i.e., if a curve $\tau_i$ appears $x$ times in $Y$ and $\tau_i \in S$, then $\tau_i$ appears $x$ times in $Y \cap S$.  The notation $|Y|$ and $|Y \cap S|$ refer to the number of curves in $Y$ and $Y \cap S$ counting multiplicities.  Define a random variable as follows:
\begin{align*}
	Y_S = & \left\{\begin{array}{lcl}
		Y \cap S, & & \text{if $|Y \cap S| \leq  |Y|/(2\beta)$}, \\
		\text{a uniform sample of $Y \cap S$ of size $|Y|/(2\beta)$}, & & \text{otherwise}.  
		\end{array}\right.
\end{align*}
To get a uniform sample of $Y \cap S$ of size $|Y|/(2\beta)$ when $|Y \cap S| > |Y|/(2\beta)$, we treat the elements of $Y \cap S$ as distinct, generate all possible subsets of $Y \cap S$ of size $|Y|/(2\beta)$, and pick one uniformly at random to be $Y_S$.  So $Y_S$ may be a multiset.  The notation $|Y_S|$ refers to be the number of curves in $Y_S$ counting multiplicities.  We first bound the probabilities of several random events.

\vspace{10pt}

\noindent \textbf{Event about $\pmb{Y}$ and $\pmb{Y_S}$.}~Since points are independently sampled from $T$ with replacement to form $Y$, the subset $Y_S$ is a uniform, independent sample of $S$ with replacement.   Since $|S| \geq n/\beta$, the chance of picking a curve in $S$ when we are forming $Y$ is at least $1/\beta$, which implies that $\mathrm{E}\bigl[|Y \cap S|\bigr] \geq |Y|/\beta$.  Therefore, $\Pr\bigl[|Y \cap S| < |Y|/(2\beta) \bigr] \le \Pr\bigl[|Y \cap S|< \E[|Y \cap S|]/2\bigr]$.  Applying the Chernoff bound to $\Pr\bigl[|Y \cap S|< \E[|Y \cap S|]/2\bigr]$, we obtain
\begin{eqnarray*}
\Pr\bigl[|Y \cap S| < |Y|/(2\beta) \bigr] & \le & \Pr\bigl[|Y \cap S|< \E[|Y \cap S|]/2\bigr] \\
& \leq & e^{-\frac{1}{8}\mathrm{E}[|Y \cap S|]} \;\;\leq \;\; e^{-|Y|/(8\beta)} \; \leq \; \left(\frac{\mu}{80\ell}\right)^{10\ell/\eps} \; < \; \mu/80.
\end{eqnarray*}
This gives our first event:
\begin{equation*}
E_{Y_S}: |Y_S| = |Y|/(2\beta), \quad \Pr\bigl[E_{Y_S}\bigr] > 1-\mu/80.
\end{equation*}
Under event $E_{Y_S}$,  line~\ref{alg:enum} of Algorithm~\ref{ALG:1_APP} will produce a subset $X$ equal to some $Y_S$ in some iteration.

\vspace{10pt}

\noindent \textbf{Event about $\pmb{c}$.}~Consider the curve $c$ returned in line~\ref{alg:34} of Algorithm~\ref{ALG:1_APP}.  The working of the $(1,\ell)$-median-34-approximation$(X,\mu/4)$ in~\cite{buchin2021approximating} guarantees that $c$ is a 34-approximate $(1,\ell)$-median of $X$ with probability at least $1-\mu/4$.  We obtain our second event:
\begin{center}
	$E_c$: $c$ is a 34-approximate $(1,\ell)$-median of $X$,  \quad
	$\Pr\bigr[E_{c}\bigr] \geq 1- \mu/4$.
\end{center}

\vspace{6pt}

\noindent \textbf{Event about a $\pmb{(1,\ell)}$-median.}~Let $c^{*}$ be an optimal $(1, \ell)$-median of $S$.   The average Fr\'{e}chet distance between $c^*$ and a curve in $S$ is $\mathrm{cost}(S,c^*)/|S|$.  In other words, if we pick a curve $\tau_i$ uniformly at random from $S$, the expected value of $d_F(\tau_i,c^*)$ is $\mathrm{cost}(S,c^*)/|S|$.  Since $Y_S$ is a uniform, independent sample of $S$ with replacement, we know that for all $\tau_i \in Y_S$, $\mathrm{E}[d_F(\tau_i,c^*)] = \mathrm{cost}(S,c^*)/|S|$.  It follows that
\[
\mathrm{E}\bigl[\mathrm{cost}(Y_S,c^*)\bigr] = \sum_{\tau_i \in Y_S} \mathrm{E}[d_F(\tau_i,c^*)] = |Y_S| \cdot \frac{\mathrm{cost}(S,c^*)}{|S|}.
\]
By Markov's inequality, $\Pr\bigl[\mathrm{cost}(Y_S,c^*) \geq \frac{4}{\mu}\mathrm{E}[\mathrm{cost}(Y_S,c^*)]\bigr]\leq \mu/4$.  Therefore, it holds with probability greater than $1-\mu/4$ that $\mathrm{cost}(Y_S,c^*) < \frac{4}{\mu}\mathrm{E}\big[\mathrm{cost}(Y_S,c^*)\bigr] = \frac{4|Y_S|}{\mu|S|}\mathrm{cost}(S,c^*)$.  Note that $|Y_S|$ must be positive then. By rearranging terms, it holds with probability greater than $1-\mu/4$ that $\frac{\mu}{|Y_S|}\mathrm{cost}(Y_S,c^*) < \frac{4}{|S|}\mathrm{cost}(S,c^*)$.  This gives our third event:
\begin{equation*}
E_{\mathrm{cost}(Y_S,c^*)}: \frac{\mu}{|Y_S|}\mathrm{cost}(Y_S,c^*) < \frac{4}{|S|}\mathrm{cost}(S,c^*), \quad
\Pr\bigl[E_{\mathrm{cost}(Y_S,c^*)}\bigr] > 1-\mu/4.
\end{equation*}

\vspace{4pt}

\noindent \textbf{Event about Lemmas~\ref{lem:first-level} and~\ref{lem:second-level}.}~Let $\Delta = \{\delta_1,\ldots,\delta_{|S|}\}$ be a set of error thresholds for $S$ that will be specified later such that $c^*$ is a solution of $Q(S,\Delta,\ell)$.  We assume the notation in Lemma~\ref{lem:first-level}.  That is, $\delta_1 \leq \ldots \leq \delta_{|S|}$, $S_i = \{\tau_i,\tau_{i+1},\ldots\}$, and $\Delta_i = \{\delta_i,\delta_{i+1},\ldots\}$.    It follows that $c^*$ is also a solution of $Q(S_r,\Delta_r,\ell)$ for all $r \in \bigl\lceil \frac{\eps |S|}{5\ell} \bigr\rceil$.  Take any $r \in \bigl[\frac{\eps |S|}{5\ell}\bigr]$.  By Lemma~\ref{lem:first-level}, there is a family $\mathcal{H}_r$ of $3l-1$ subsets of $S_r$ for some $l \in [\ell]$, each containing $\frac{\eps |S|}{5\ell}$ curves, such that some desirable consequences will follow if some subset of $S_r$ contains $\tau_r$ and intersects every subset in $\mathcal{H}_r$.  We argue that $Y_S$ likely contains such a subset $R$ with one additional property that we explain below.  Conditioned on $E_{Y_S}$, we have $|Y_S| = \frac{40\ell}{\eps}\ln \frac{80\ell}{\mu}$.  

First, let $Z_1 = \bigl\{\tau_i : i \in \bigl[\frac{\eps|S|}{5\ell}\bigr] \bigr\}$, and we analyze $|Y_S \cap Z_1|$.  Given a curve drawn uniformly at random from $S$, the probability that it belongs to $Z_1$ is $\frac{\eps}{5\ell}$.  Therefore, the expected value of $|Y_S \cap Z_1|$ is $\frac{\eps}{5\ell}|Y_S| = 8\ln\frac{80\ell}{\mu}$.  Applying the Chernoff bound, we obtain $\Pr\bigl[|Y_S \cap Z_1| > 16\ln(80\ell/\mu)\bigr] \leq \Pr\bigl[|Y_S \cap Z_1| > 2\,\mathrm{E}\bigl[|Y_S \cap Z_1|\bigr]\bigr] \leq e^{-\frac{1}{3}\mathrm{E}[|Y_S \cap Z_1|]} < \mu/80$.
Similarly, $\Pr\bigl[|Y_S \cap Z_1| < 4\ln(80\ell/\mu)\bigr] \leq \Pr\bigl[|Y_S \cap Z_1| < \mathrm{E}\bigl[|Y_S \cap Z_1|\bigr]/2\bigr] \leq e^{-\frac{1}{8}\mathrm{E}[|Y_S \cap Z_1|]} < \mu/80$.
Therefore, 
\begin{equation}
\Pr\left[4\ln\frac{80\ell}{\mu} \leq |Y_S \cap Z_1| \leq 16\ln\frac{80\ell}{\mu}\right] > 1-\mu/40.
\label{eq:event-1}
\end{equation}

Let $r = \mathrm{argmin}\{\tau_i : Y_S \cap Z_1\}$.  We have $Y_S \subseteq S_r$ and $\tau_r \in Y_S$.

There are at most $\frac{\eps |S|}{10\ell}$ curves in $S$ that has a Fr\'{e}chet distance of at least $\frac{10\ell}{\eps|S|}\mathrm{cost}(S,c^*)$; otherwise, the total would exceed $\mathrm{cost}(S,c^*)$, an impossibility.   It means that for every $H \in \mathcal{H}_r$, at least half of the curves in $H$ have Fr\'{e}chet distances at most $\frac{10\ell}{\eps|S|}\mathrm{cost}(S,c^*)$ from $c^*$.   Since $Y_S$ is a uniform, independent sample of $S$ with replacement, the probability of $Y_S$ containing a curve from a particular $H \in \mathcal{H}_r$ that has a Fr\'{e}chet distance at most $\frac{10\ell}{\eps|S|}\mathrm{cost}(S,c^*)$ from $c^*$ is at least $1-\bigl(1-\frac{\eps}{10\ell}\bigr)^{(40\ell/\eps)\ln(80\ell/\mu)} \geq 1 - \frac{\mu}{80\ell}$.  As a result, by the union bound, 
\begin{equation}
	\Pr\left[\forall H \in \mathcal{H}_r, \, \exists \tau_i \in Y_S \cap H \; \text{s.t.} \; d_F(\tau_i,c^*) \leq \frac{10\ell}{\eps|S|}\mathrm{cost}(S,c^*)\right] > 1-\frac{\mu}{80\ell}\cdot(3l-1) > 1-\frac{\mu}{25}.
	\label{eq:event-2}
\end{equation}

By \eqref{eq:event-1} and \eqref{eq:event-2}, it holds with probability greater than $1- 13\mu/200$ that $Y_S$ contains a subset $R \subseteq S_r$ that contains $\tau_r$, has size at most $3l$, and enables us to invoke Lemma~\ref{lem:first-level}.  As a result, there exists a configuration $\Psi = (\mathcal{P},\mathcal{C},\mathcal{S},\mathcal{A}) \in C(\overline{\mathcal{H}}_r \cup R,h,\eps^2)$ that satisfies Lemma~\ref{lem:first-level}.

Given $R$ and $\Psi$, by the same argument, for any $j \in [2h]$, the probability that $Y_S$ contains a curve in a particular $H \in \mathcal{H}_{\Psi}$ that has a Fr\'{e}chet distance no more than $\frac{10\ell}{\eps|S|}\mathrm{cost}(S,c^*)$ from $c^*$ is at least $1-\frac{\mu}{80\ell}$.  As a result, by the union bound, 
\begin{equation}
	\Pr\left[\forall H \in \mathcal{H}_{\Psi}, \, \exists \tau_i \in Y_S \cap H \; \text{s.t.} \; d_F(\tau_i,c^*) \leq \frac{10\ell}{\eps|S|}\mathrm{cost}(S,c^*)\right] > 1-\frac{\mu}{80\ell}\cdot 2h > 1-\frac{\mu}{40}.
\end{equation}
That is, $Y_S$ contains a subset $R' \subseteq \overline{\mathcal{H}}_r \cup R$ that has size at most $2h$ and enables us to invoke Lemma~\ref{lem:second-level}.  We obtain our fourth event:

\vspace{8pt} 

\noindent\hspace{20pt}\parbox[t]{15cm}{
	$E_{\Psi}$: \parbox[t]{14.5cm}{
		$\bullet$~~\parbox[t]{14cm}{Given that $c^*$ is a solution for $Q(S,\Delta,\ell)$, the existence of $S_r$, $\mathcal{H}_r \subseteq 2^{S_r}$, $R \subseteq Y_S \cap S_r$, $\Psi \in C(\overline{\mathcal{H}}_r \cup R,h,\eps^2)$, $\mathcal{H}_{\Psi} \subseteq 2^{\overline{\mathcal{H}}_r \cup R}$, $R' \subseteq Y_S \cap (\overline{\mathcal{H}}_r \cup R)$, $\Psi' \in C(R \cup R',h,\eps^2)$, and $\Psi'' \in (\overline{\mathcal{H}}_\Psi \cup R,h,\eps^2)$ that satisfy Lemmas~\ref{lem:first-level} and~\ref{lem:second-level}.}
	}
}	

\vspace{4pt}\noindent\hspace{42pt}\parbox[t]{15cm}{
	$\bullet$~~\parbox[t]{14cm}{$|R| \leq 3l$ and $|R'| \leq 2h$.}
}

\vspace{4pt}

\noindent\hspace{42pt}\parbox[t]{15cm}{
		$\bullet$~~\parbox[t]{14cm}{For all $\tau_i \in R \cup R'$, $d_F(\tau_i,c^*) \leq \frac{10\ell}{\eps|S|}\mathrm{cost}(S,c^*)$.}
	}

\vspace{2pt}
\noindent\hspace{20pt}\parbox[t]{20cm}{
	$\Pr\bigr[E_{\Psi} \,|\, E_{Y_S} \bigr] \geq 1- 9\mu/100$.
}

\vspace{6pt}

\noindent \textbf{Analysis.}~We describe the analysis conditioned on the events $E_{Y_S}$, $E_c$, $E_{\mathrm{cost}(Y_S,c^*)}$, and $E_{\Psi}$. 

Conditioned on event $E_{Y_S}$, line~\ref{alg:enum} of Algorithm~\ref{ALG:1_APP} will produce a subset $X$ equal to some $Y_S$.  We are interested in this particular iteration of lines~\ref{alg:34}--\ref{alg:end-itr}.  We compute a 34-approximate $(1,\ell)$-median $c$ for $X = Y_S$ in line~\ref{alg:34}.  We also compute a st $\Sigma'$ of curves in line~\ref{alg:curve}.  Our goal is to show that some curve $c' \in \Sigma' \cup \{c\}$ satisfies $\mathrm{cost}(S,c') \leq (1+\eps)\mathrm{cost}(S,c^*)$.   Throughout this analysis, we assume that $c \not= c^*$; otherwise, there is nothing to prove.

We first define a neighborhood $N_c$ of $c$ in $S$ and another neighborhood $N_{c^*}$ of $c^*$ in $S$ in terms of $\mathrm{cost}(S,c^*)$:
\[
N_c = \left\{ \tau_i \in S : d_F(\tau_i,c) \leq \frac{\eps}{|S|}\mathrm{cost}(S,c^*)\right\}, \quad
N_{c^*} = \left\{ \tau_i \in S : d_F(\tau_i,c^*) \leq \frac{1}{\eps^2|S|}\mathrm{cost}(S,c^*)\right\}.
\]
There are no more than $\eps^2|S|$ curves in $S$ that do not belong to $N_{c^*}$; otherwise, the total cost would exceed $\mathrm{cost}(S,c^*)$, an impossibility.
\begin{equation}
	|S \setminus N_{c^*}| \leq \eps^2 |S| \, \Rightarrow \, |S| - |N_{c^*}| \leq \eps^2|S| \, \Rightarrow  \, |N_{c^*}| \geq (1-\eps^2)|S|.   \label{eq:1}
\end{equation}
The analysis is divided into two cases depending on the size of $N_{c^*}\setminus N_c$.

\vspace{10pt}

\noindent \emph{Case~1:} $|N_{c^*}\setminus N_c| \leq 2\eps|N_{c^*}|$.

Suppose that $d_F(c,c^*) \leq 4\eps \cdot \mathrm{cost}(S,c^*)/|S|$.  We can derive a good bound on $\mathrm{cost}(S,c)$ easily:
\[
\mathrm{cost}(S,c) \leq \sum_{\tau_i \in S} \bigl(d_F(\tau_i,c^*) + d_F(c^*,c)\bigr) \leq \mathrm{cost}(S,c^*) + 4\eps \cdot \mathrm{cost}(S,c^*).
\]

The other case is that $d_F(c,c^*) >  4\eps \cdot \mathrm{cost}(S,c^*)/|S|$.  We prove that this case leads to a contradiction and hence it does not happen.  First of all, $|S \setminus N_c| \leq |S \setminus N_{c^*}| + |N_{c^*} \setminus N_c|$.  We have $|S \setminus N_{c^*}| \leq \eps^2|S|$ by~\eqref{eq:1}, and $|N_{c^*} \setminus N_c| \leq 2\eps|N_{c^*}|$ in Case~1.  Therefore, 
\begin{equation}
|S \setminus N_c| \; \leq \; \eps^2|S| + 2\eps|N_{c^*}| \; \leq \; \eps^2|S|+ 2\eps|S| \; \leq \; 3\eps|S|.  \label{eq:1-1}
\end{equation}
It follows that 
\begin{equation}
|N_c| = |S| - |S \setminus N_c| \geq (1-3\eps)|S|.  \label{eq:1-2}
\end{equation}
For every $\tau_i \in N_c$, we have $d_F(\tau_i,c) \leq \eps\cdot\mathrm{cost}(S,c^*)/|S|$ by definition, which implies that
\[
d_F(\tau_i,c^*) - d_F(\tau_i,c) \geq d_F(c,c^*) - d_F(\tau_i,c) - d_F(\tau_i,c) \geq d_F(c,c^*) - 2\eps\cdot\mathrm{cost}(S,c^*)/|S|.
\]
Since we are considering the case that $d_F(c,c^*) > 4\eps\cdot\mathrm{cost}(S,c^*)/|S|$, we conclude that
\begin{equation}
\forall\, \tau_i \in N_c, \quad d_F(\tau_i,c^*) - d_F(\tau_i,c) > \frac{1}{2}d_F(c,c^*).  \label{eq:1-3}
\end{equation}
By triangle inequality,
\begin{equation}
\forall\, \tau_i \in S \setminus N_c, \quad d_F(\tau_i,c) - d_F(\tau_i,c^*) \leq d_F(c,c^*).  \label{eq:1-4}
\end{equation}
Putting \eqref{eq:1-3} and \eqref{eq:1-4} together gives:
\begin{eqnarray*}
\mathrm{cost}(S,c^*) - \mathrm{cost}(S,c) & = & \sum_{\tau_i \in N_c} \bigl(d_F(\tau_i,c^*) - d_F(\tau_i,c)\bigr) +  \sum_{\tau_i \in S \setminus N_c} \bigl(d_F(\tau_i,c^*) - d_F(\tau_i,c)\bigr) \\
& > & \frac{1}{2}|N_c|\cdot d_F(c,c^*) - |S \setminus N_c| \cdot d_F(c,c^*) \\
& \stackrel{\eqref{eq:1-1}, \eqref{eq:1-2}}{\geq} & \frac{1-9\eps}{2}\cdot d_F(c,c^*).
\end{eqnarray*}
We have $d_F(c,c^*) > 0$ as $c \not= c^*$ by assumption.  It leads to the contradiction that $\mathrm{cost}(S,c^*) > \mathrm{cost}(S,c)$ as $\eps < 1/9$ by assumption.

\vspace{10pt}

\noindent \emph{Case~2:} $|N_{c^*}\setminus N_c| > 2\eps|N_{c^*}|$.

%Under event $E_{Y_S}$, the enumeration of subsets of $Y$ of size $|Y|/(2\beta)$ will produce such a subset $X$ equal to some $Y_S$.  We consider such a subset $Y_S$.  
Our idea is to apply Lemmas~\ref{lem:first-level} and~\ref{lem:second-level} to analyze the cost of the curves produced in line~\ref{alg:curve} of Algorithm~\ref{ALG:1_APP}.  To this end, we must argue that the enumeration in lines~\ref{alg:subset} and~\ref{alg:error} of Algorithm~\ref{ALG:1_APP} will produce an appropriate $W$ and $\Delta_W$.   We first take care of $\Delta_W$ in the following.

By~\eqref{eq:1}, $|N_{c^*}| \geq (1-\eps^2)|S|$.  Since $|N_{c^*}\setminus N_c| > 2\eps|N_{c^*}|$ in Case~2, we get 
\begin{equation}
|N_{c^*}\setminus N_c| \; > \; 2\eps(1-\eps^2)|S| \; \geq \; \eps|S|, \quad\quad (\because \eps < 1/9)
\label{eq:1-5}
\end{equation}
%In the discussion of the event $E_{\Psi}$, we extracted three subsets $Y_S \cap Z_1$, $\{\tau_r\} \cup Z_2$, and $Z_3$ from $Y_S$ that have sizes at most $16\ln\frac{80\ell}{\mu}$, $6\ell\ln\frac{80\ell}{\mu} + 1$, and $4\ell\frac{80\ell}{\mu}$, respectively.   Since $|Y_S| = \frac{40\ell}{\eps}\ln\frac{80\ell}{\mu}$, we can extract another subset $Z_4$ from $Y_S \setminus (Z_1 \cup \{\tau_r\} \cup Z_2 \cup Z_3)$ that has size $\frac{1}{\eps}\ln\frac{80\ell}{\mu}$.  When we view the formation of 
Since $Y_S$ is a random sample of $S$ with replacement, the probability of picking a curve from $N_{c^*} \setminus N_c$ is at least $\eps$ by~\eqref{eq:1-5}.  It follows that 
\begin{eqnarray}
& & \Pr\left[\exists \, \tau_i \in Y_S \,\; \text{s.t.} \; d_F(\tau_i,c^*) \leq \frac{1}{\eps^2|S|}\cdot\mathrm{cost}(S,c^*) \, \wedge \, d_F(\tau_i,c) > \eps \cdot \frac{\mathrm{cost}(S,c^*)}{|S|} \;\;  \Bigl| \;\; E_{Y_S} \right] \nonumber \\
& \geq & 1 - (1-\eps)^{|Y_S|} \nonumber \\
& \geq & 1 - (1-\eps)^{(40\ell/\eps)\ln(80\ell/\mu)} \nonumber \\
& \geq & 1 - \mu/80.  \label{eq:2}
\end{eqnarray}
There are three implications conditioned on the event in~\eqref{eq:2}.  First, we have a lower bound for $\mathrm{cost}(Y_S,c)$:
\begin{equation}
	\mathrm{cost}(Y_S,c) \; \geq \; \eps \cdot \frac{\mathrm{cost}(S,c^*)}{|S|}. 
	\label{eq:3}
\end{equation}
Second, the Fr\'{e}chet distance upper bound of $\frac{10\ell}{\eps|S|}\mathrm{cost}(S,c^*)$ referenced in event $E_{\Psi}$ is bounded from above by the upper bound $U$ computed in line~\ref{alg:LU} in Algorithm~\ref{ALG:1_APP}, which means that the range of error thresholds that Algorithm~\ref{ALG:1_APP} considers is sufficiently large.
\begin{equation}
	\frac{10\ell}{\eps} \cdot \frac{\mathrm{cost}(S,c^*)}{|S|} \; \stackrel{\eqref{eq:3}}{<} \; \frac{10\ell}{\eps^2}\mathrm{cost}(Y_S,c) \; = \; U.
	\label{eq:4}
\end{equation}
Third, using the fact that $c$ is a 34-approximation of the optimal $(1,\ell)$-median of $Y_S$ and the event $E_{\mathrm{cost}(Y_S,c^*)}$, we can prove an upper bound in terms of $\mathrm{cost}(S,c^*)/|S|$ for the lower bound $L$ computed in line~\ref{alg:LU} in Algorithm~\ref{ALG:1_APP}.  The lower bound $L$ is also the discrete step size of the error thresholds that we consider.  This upper bound on $L$ will allow us to bound the error caused by the discrete step size.
\begin{eqnarray}
	L & = & \frac{\eps\mu}{34} \cdot \frac{\mathrm{cost}(Y_S,c)}{|Y_S|} \nonumber \\
	    & \leq & \eps\mu \cdot \frac{\mathrm{cost}(Y_S,c^*)}{|Y_S|} \quad\quad\quad\quad (\because \text{$c$ is a 34-approximation}) \nonumber \\
	    & < & 4\eps\cdot \frac{\mathrm{cost}(S,c^*)}{|S|}.    \quad\quad\quad\quad\; (\because \text{event $E_{\mathrm{cost}(Y_S,c^*)}$})  \label{eq:5}
\end{eqnarray}

The discrete error thresholds between $L$ and $\infty$ with step size $L$ roughly capture the Fr\'{e}chet distances of all input curves $\tau_i \in S$ from $c^*$.  This motivates us to define:
\[
\Delta = \left\{\delta_i : \tau_i \in S, \; \delta_i = \lceil d_F(\tau_i,c^*)/L \rceil \cdot L \right\}.
\]
The set $\Delta$ is the set of error thresholds referenced in event $E_{\Psi}$.  Clearly, $c^*$ is a solution for $Q(S,\Delta,\ell)$ because $\delta_i \geq d_F(\tau_i,c^*)$ for all $\tau_i \in S$.  So we fulfill the precondition of applying $E_{\Psi}$.  We cannot afford the time to compute $\Delta$ explicitly.  Fortunately, Lemma~\ref{lem:second-level} says that it is unnecessary to do so; it suffices to capture the subset of $\Delta$ for $R \cup R'$ such that $R \subseteq S_r$, $R$ contains $\tau_r$ and intersects every subset in $\mathcal{H}_r$, $R' \subseteq \overline{\mathcal{H}}_r \cup R$, and $R'$ intersects every subset in $\mathcal{H}_{\Psi}$.  The event $E_{\Psi}$ exactly provides such a $R \cup R'$ of size at most $3l+2h$.  Moreover, by event $E_{\Psi}$, for all $\tau_i \in R \cup R'$, $d_F(\tau_i,c^*) \leq \frac{10\ell}{\eps|S|}\mathrm{cost}(S,c^*)$.  Consider the iteration in which the subset $W$ equal to $R \cup R'$ is produced in line~\ref{alg:subset} of Algorithm~\ref{ALG:1_APP}.  Since the Fr\'{e}chet distance bound $\frac{10\ell}{\eps|S|}\mathrm{cost}(S,c^*)$ is no more than $U$ by~\eqref{eq:4}, the set of error thresholds $\Delta_W = \{\delta_i \in \Delta : \tau_i \in W\}$ will be produced in line~\ref{alg:error} at some point.  We perform a cost analysis in the following.

Note that $\max\{\delta_i \in \Delta_W\} \leq \frac{10\ell}{\eps|S|}\mathrm{cost}(S,c^*) + L$.  By the implication of Lemma~\ref{lem:second-level}, the output of the two-phase construction on all configurations in $C(W,h,\eps^2)$ must include a curve $c'$ such that:
\begin{eqnarray*}
	\forall\, \tau_i \in \overline{\mathcal{H}}_{\Psi}, \quad d_F(\tau_i,c') & \leq & \delta_i + 4\sqrt{d}\eps^2 \cdot \max\{\delta_i \in \Delta_W\} \nonumber \\
	& \leq & d_F(\tau_i,c^*) + L + 4\sqrt{d}\eps^2 \cdot \left(\frac{10\ell}{\eps|S|}\mathrm{cost}(S,c^*)+L\right)  \nonumber \\
	& \stackrel{\eqref{eq:5}}{=} & d_F(\tau_i,c^*) + O(\sqrt{d}\ell\eps) \cdot \frac{\mathrm{cost}(S,c^*)}{|S|}.  
\end{eqnarray*}
We still have to analyze the cost of the curves in $S \setminus \overline{\mathcal{H}}_{\Psi}$.  Note that
$S \setminus \overline{\mathcal{H}}_{\Psi} \subseteq \bigl\{\tau_i : i \in \bigl[\frac{\eps|S|}{5\ell}\bigr]\bigr\} \cup \bigcup_{H \in \mathcal{H}_r \cup \mathcal{H}_{\Psi}} H$.  The size of $S \setminus \overline{\mathcal{H}}_{\Psi}$ is thus at most $\frac{\eps |S|}{5\ell} \cdot 3l+ \frac{\eps |S|}{5\ell} \cdot 2h  \leq \eps|S|$.  There are at least $|S|/2$ curves in $S$ that have Fr\'{e}chet distances at most $2\,\mathrm{cost}(S,c^*)/|S|$ from $c^*$; otherwise, the total cost would exceed $\mathrm{cost}(S,c^*)$, an impossibility.  Since $|S \setminus \overline{\mathcal{H}}_{\Psi}| \leq \eps|S|$ and $\eps < 1/9$, we conclude that:
\[
\exists \, \tau_{i_0} \in \overline{\mathcal{H}}_{\Psi} \;\; \text{s.t.} \;\; d_F(\tau_{i_0},c^*) \leq 2\,\mathrm{cost}(S,c^*)/|S|. 
\]
We are ready to bound $\mathrm{cost}(S,c')$:
\begin{eqnarray*}
	\mathrm{cost}(S,c') & = & 
	\sum_{\tau_i \in \overline{\mathcal{H}}_{\Psi}} d_F(\tau_i,c') +
	\sum_{\tau_i \in S \setminus \overline{\mathcal{H}}_{\Psi}} d_F(\tau_i,c') \\
	& \leq & \sum_{\tau_i \in \overline{\mathcal{H}}_{\Psi}} d_F(\tau_i,c^*) + O(\sqrt{d}\ell\eps) \cdot | \overline{\mathcal{H}}_{\Psi}| \cdot \frac{\mathrm{cost}(S,c^*)}{|S|} + \\
	& & \quad\quad \sum_{\tau_i \in S \setminus \overline{\mathcal{H}}_{\Psi}} \left(d_F(\tau_i,c^*) + d_F(c^*,\tau_{i_0}) + d_F(\tau_{i_0},c')\right) \\
	& \leq & \mathrm{cost}(S,c^*) + O(\sqrt{d}\ell\eps) \cdot \mathrm{cost}(S,c^*)+  |S \setminus \overline{\mathcal{H}}_{\Psi}| \cdot \left(d_F(c^*,\tau_{i_0}) + d_F(\tau_{i_0},c')\right) \\
	& \leq & \mathrm{cost}(S,c^*) + O(\sqrt{d}\ell\eps) \cdot \mathrm{cost}(S,c^*) + \\
	& & \quad\quad\quad \eps|S| \cdot \left(2d_F(\tau_{i_0},c^*) + O(\sqrt{d}\ell\eps) \cdot \frac{\mathrm{cost}(S,c^*)}{|S|} \right) \quad\quad\quad (\because \tau_{i_0} \in \overline{\mathcal{H}}_{\Psi}) \\
	& \leq & \mathrm{cost}(S,c^*) + O(\sqrt{d}\ell\eps) \cdot \mathrm{cost}(S,c^*) + \eps|S| \cdot \frac{4\,\mathrm{cost}(S,c^*)}{|S|} \\
	& \leq &  \mathrm{cost}(S,c^*) + O(\sqrt{d}\ell\eps) \cdot \mathrm{cost}(S,c^*).
\end{eqnarray*}
This completes the analysis of Case~2.

The probability bound of $1-\mu$ follows from $\Pr[E_{Y_S}]$, $\Pr[E_c]$, $\Pr[E_{\mathrm{cost}(Y_S,c^*)}]$, $\Pr[E_{\Psi} | E_{Y_S}]$, and the probability of the event in \eqref{eq:2}.

The running time is asymptotically bounded by $N_X \cdot \bigr(T_{\text{34apx}} + \ell^2 N_W \cdot N_{\Delta_W} \cdot N_{C_W} \cdot O(m|X|\log m + hm|X|2^{O(d)})\bigr)$, where $N_X$ is the number of subsets of $Y$ of size $|Y|/(2\beta)$, $T_{\text{34apx}}$ is the running time of the $(1,\ell)$-median-34-approximation algorithm, $N_W$ is the number of subsets of $X$ of size at most $3l+2h$, $N_{\Delta_W}$ is the number of possible sets of error thresholds for $W$, and $N_{C_W}$ is the number of configurations in $C(W,h,\eps^2)$. 

One can verify that $N_X = O\bigl(|Y|^{|Y|/(2\beta)}\bigr) = \tilde{O}\bigl((\beta\ell/\eps)^{O((\ell/\eps)\log(1/\mu))}\bigr)$, $N_W = O\bigl(|X|^{5\ell}\bigr) = \tilde{O}\bigl((\ell/\eps)^{O(\ell)}\bigr)$, $N_{\Delta_W} = O\bigl((\mu^{-1}\eps^{-3}\ell|X|)^{|W|}\bigr) = \tilde{O}\bigl(\mu^{-O(\ell)}\cdot(\ell/\eps)^{O(\ell)}\bigr)$, and $N_{C_W} = O(m^{O(h|W|)} \cdot (\ell/\eps)^{O(dh)}\bigr) = O\bigl(m^{O(\ell^2)} \cdot (\ell/\eps)^{O(d\ell)}\bigr)$.  We have $T_{\text{34apx}} = \tilde{O}(m^3)$~\cite{buchin2021approximating}.  The running time bound is thus equal to $\tilde{O}\bigl((\beta\ell/\eps)^{O((\ell/\eps)\log(1/\mu))} \cdot \bigl(m^3 + m^{O(\ell^2)} \cdot \mu^{-O(\ell)} \cdot (\ell/\eps)^{O(d\ell)}\bigr)\bigr) = \tilde{O}\bigl(m^{O(\ell^2)} \cdot \mu^{-O(\ell)} \cdot  (\ell/\eps)^{O(d\ell)} \cdot (\beta\ell/\eps)^{O((\ell/\eps)\log(1/\mu))}\bigr)$.

To reduce the approximation ratio from $1 + O(\sqrt{d}\ell\eps)$ to $1+\eps$, we need to reduce $\eps$ to $\eps/\Theta(\sqrt{d}\ell)$.  In summary, the total running time is  $\tilde{O}\bigl(m^{O(\ell^2)} \cdot \mu^{-O(\ell)} \cdot (d\beta\ell/\eps)^{O((d\ell/\eps)\log(1/\mu))}\bigr)$.
\end{proof}

\end{document}